\documentclass[11pt]{article}
\usepackage[utf8]{inputenc}
\usepackage[margin=0.8in]{geometry}
\usepackage{amsmath,amsfonts,amsthm, amssymb}
\usepackage[dvipsnames]{xcolor}
\usepackage[colorlinks, linkcolor=OliveGreen, citecolor=purple,urlcolor=purple, bookmarks=true,backref=page]{hyperref}
\usepackage[capitalise,nameinlink]{cleveref}
\usepackage{bbm}
\usepackage{braket}
\usepackage{mathtools}
\usepackage{enumerate}
\usepackage[inline]{enumitem}
\usepackage{hyperref}
\usepackage{soul}
\usepackage{tikz}
\usepackage{caption}
\usepackage{float}
\usepackage{booktabs}
\usepackage{tabularx,environ}
\renewcommand{\backref}[1]{}

\renewcommand{\backrefalt}[4]{%
\ifcase #1 %
\or
[p.\ #2]%
\else
[pp.\ #2]%
\fi}

\newcommand{\trianglewedge}{%
    \begin{tikzpicture}[scale=0.1,baseline=(current bounding box.center)]
        \draw (0,0) -- (1,2) -- (2,0); 
    \end{tikzpicture}%
}

\usepackage{bm}
\newcommand{\pathsprogress}{\mathcal{P}}
\newcommand{\cycleprogress}{\mathcal{C}}
\renewcommand{\pathsprogress}{{\Lambda}}
\renewcommand{\cycleprogress}{\bm{\square}}

\usepackage{tcolorbox}
\renewcommand{\fbox}[1]{%
\tcbox[colframe=black,colback=white,arc=2mm,boxrule=0.5mm,boxrule=0.25mm]{#1}%
}

\renewcommand{\labelitemi}{\tiny$\blacksquare$}

\let\originalleft\left
\let\originalright\right
\renewcommand{\left}{\mathopen{}\mathclose\bgroup\originalleft}
\renewcommand{\right}{\aftergroup\egroup\originalright}

\DeclareMathOperator{\twodots}{..}
\DeclareMathOperator{\1}{\mathbbm{1}}
\DeclareMathOperator{\id}{\mathbb{I}}

\DeclareMathOperator{\pr}{\mathrm{Pr}}

\DeclareMathOperator{\maxdeg}{\mathrm{maxdeg}}
\DeclareMathOperator{\triedge}{\mathsf{TriangleEdge}}
\DeclareMathOperator{\trivertex}{\mathsf{TriangleVertex}}
\DeclareMathOperator{\tri}{\mathsf{Triangle}}

\newcommand{\kcycle}[1]{#1\mathsf{\text{-}CYCLE}}

\DeclareMathOperator{\ED}{\mathrm{ED}}

\DeclareMathOperator{\OR}{\mathrm{OR}}

\DeclareMathOperator{\expect}{\mathrm{E}}

\DeclareMathOperator{\poly}{\mathrm{poly}}
\DeclareMathOperator{\hide}{\mathsf{HIDE}}
\DeclareMathOperator{\shuffle}{\mathsf{SHUFFLE}}

\DeclareMathOperator{\calA}{\mathcal{A}}
\DeclareMathOperator{\calD}{\mathcal{D}}
\DeclareMathOperator{\calG}{\mathcal{G}}
\DeclareMathOperator{\calH}{\mathcal{H}}

\DeclareMathOperator{\calO}{\mathcal{O}}
\DeclareMathOperator{\calR}{\mathcal{R}}
\DeclareMathOperator{\calS}{\mathcal{S}}
\DeclareMathOperator{\calT}{\mathcal{T}}

\DeclareMathOperator{\success}{\mathrm{succ}}
\DeclareMathOperator{\record}{\mathrm{rec}}

\DeclareMathSymbol{\dash}{\mathord}{AMSa}{"39}
\DeclarePairedDelimiter\ceil{\lceil}{\rceil}
\DeclarePairedDelimiter\floor{\lfloor}{\rfloor}

\newcommand{\abs}[1]{\left|#1\right|}
\newcommand{\kdist}[1]{#1\mathsf{\text{-}DIST}}
\newcommand{\ksum}[1]{#1\mathsf{\text{-}SUM}}
\newcommand{\Bad}{\mathrm{Bad}}
\newcommand{\Prob}{\mathrm{Pr}}
\newcommand{\BarPr}{\widetilde{\mathrm{Pr}}}
\newcommand{\norm}[1]{\left\|#1\right\|}

\newcommand{\calP}{\mathcal{P}}

\newcommand{\Adv}{\mathrm{Adv}}
\newcommand{\Act}{\mathrm{Act}}
\newcommand{\Cst}{\mathrm{Cst}}
\newcommand{\pSearch}{\mathrm{pSearch}}
\newcommand{\TVD}{\mathrm{TVD}}
\newcommand{\SEEN}{\mathrm{SEEN}}
\newcommand{\DENSE}{\mathrm{DENSE}}
\newcommand{\SPARSE}{\mathrm{SPARSE}}
\newcommand{\dense}{\mathrm{dense}}
\newcommand{\sparse}{\mathrm{sparse}}
\newcommand{\MAJ}{\mathrm{MAJ}}
\newcommand{\SMAJ}{\mathrm{\Sigma MAJ}}
\newcommand{\Bin}{\mathrm{Bin}}
\newcommand{\tildeO}{\widetilde{O}}
\newcommand{\tildeD}{\widetilde{D}}
\newcommand{\tildeOmega}{\widetilde{\Omega}}
\newcommand{\tildeTheta}{\widetilde{\Theta}}
\newcommand{\polylog}{\mathrm{polylog}}
\renewcommand{\emptyset}{\varnothing}

\newcommand{\ketbra}[2]{|#1\rangle\langle#2|}
\newcommand{\ketbrasame}[1]{|#1\rangle\langle#1|}

\newtheorem{theorem}{Theorem}[section]
\newtheorem{lemma}[theorem]{Lemma}
\newtheorem{fact}[theorem]{Fact}
\newtheorem{claim}[theorem]{Claim}
\newtheorem{conjecture}[theorem]{Conjecture}
\newtheorem{corollary}[theorem]{Corollary}
\newtheorem{proposition}[theorem]{Proposition}

\theoremstyle{definition}
\newtheorem{definition}[theorem]{Definition}
\newtheorem{remark}[theorem]{Remark}

\numberwithin{equation}{section}



\makeatletter
\newcommand{\problemtitle}[1]{\gdef\@problemtitle{#1}}
\newcommand{\probleminput}[1]{\gdef\@probleminput{#1}}
\newcommand{\problempromise}[1]{\gdef\@problempromise{#1}}
\newcommand{\problemquestion}[1]{\gdef\@problemquestion{#1}}
\newcommand{\problempremise}[1]{\gdef\@problempremise{#1}}

\NewEnviron{problem}{
\probleminput{}\problemquestion{}
  \BODY
  \par\addvspace{.5\baselineskip}
  \noindent
  \begin{tabularx}{\textwidth}{@{\hspace{\parindent}} l X c}
    \textbf{Input:} & \@probleminput \\
    \textbf{Question:} & \@problemquestion
  \end{tabularx}
  \par\addvspace{.5\baselineskip}
}

\NewEnviron{promiseproblem}{
\probleminput{}\problempromise{}\problemquestion{}
  \BODY
  \par\addvspace{.5\baselineskip}
  \noindent
  \begin{tabularx}{\textwidth}{@{\hspace{\parindent}} l X c}
    \textbf{Input:} & \@probleminput \\
    \textbf{Promise:} & \@problempromise \\
    \textbf{Question:} & \@problemquestion
  \end{tabularx}
  \par\addvspace{.5\baselineskip}
}

\usepackage{centernot}

\usepackage{enumitem}
\setenumerate{label=\rm (\roman*)}
\setlist[enumerate,1]{itemsep=-1\parsep, topsep=0mm}
\setlist[description]{itemsep=-1\parsep, topsep=0mm}
\setlist[itemize]{itemsep=-1\parsep, topsep=0mm}

\makeatletter
\def\@tocline#1#2#3#4#5#6#7{\relax
  \ifnum #1>\c@tocdepth 
  \else
    \par \addpenalty\@secpenalty\addvspace{\ifnum #1=1 2mm \else #2\fi}%
    \begingroup \hyphenpenalty\@M
    \@ifempty{#4}{%
      \@tempdima\csname r@tocindent\number#1\endcsname\relax
    }{%
      \@tempdima#4\relax
    }%
    \parindent\z@ \leftskip#3\relax \advance\leftskip\@tempdima\relax
    \rightskip\@pnumwidth plus4em \parfillskip-\@pnumwidth
          \ifnum #1=1 \bfseries #5\else #5\fi 
   \leavevmode\hskip-\@tempdima
      \ifcase #1
       \or\or \hskip 1em \or \hskip 2em \else \hskip 3em \fi%
#6     \nobreak\relax
{\ifnum #1=1\hfill \else \SparseDotfill\fi}
 \hbox to\@pnumwidth{\@tocpagenum{
    \ifnum #1=1 \bfseries \fi #7}}\par
    \nobreak
    \endgroup
  \fi}
\makeatother

\renewcommand{\epsilon}{\varepsilon}
\newcommand{\I}{\mathbb{I}}

\title{\centering{Quantum Algorithms on Edge Lists\\
\large Hiding, Shuffling, and Cycle Finding}
}
\author{\normalsize 
Amin Shiraz Gilani\footnote{University of Maryland. \texttt{asgilani@umd.edu}} \;\;\;\; 
Daochen Wang\footnote{University of British Columbia. \texttt{wdaochen@gmail.com}} \;\;\;\;
Pei Wu\footnote{The Pennsylvania State University. \texttt{pei.wu@psu.edu}} \;\;\;\;
Xingyu Zhou\footnote{University of British Columbia. \texttt{zxingyu@cs.ubc.ca}} \;\;\;\;}
\date{\vspace{-5ex}}

\begin{document}

\maketitle
\thispagestyle{empty}
\begin{abstract}
The edge list model is arguably the simplest input model for graphs, where the graph is specified by a list of its edges. In this model, we study the quantum query complexity of three variants of the triangle finding problem. The first asks whether there exists a triangle containing a target edge and raises general questions about the hiding of a problem's input among irrelevant data. The second asks whether there exists a triangle containing a target vertex and raises general questions about the  shuffling of a problem's input. The third asks whether there exists a triangle; this problem bridges the $3$-distinctness and $3$-sum problems, which have been extensively studied by both cryptographers and complexity theorists. We provide tight or nearly tight results for these problems as well as some first answers to the general questions they raise.

Furthermore, given any graph with low maximum degree, such as a typical random sparse graph, we prove that the quantum query complexity of finding a length-$k$ cycle in its length-$m$ edge list is $m^{3/4-1/(2^{k+2}-4)\pm o(1)}$, which matches the best-known upper bound for the quantum query complexity of $k$-distinctness on length-$m$ inputs up to an $m^{o(1)}$ factor. We prove the lower bound by developing new techniques within Zhandry's recording query framework \cite{zhandry_compressed} as generalized by Hamoudi and Magniez \cite{Hamoudi_2023}. These techniques extend the framework to treat any non-product distribution that results from conditioning a product distribution on the absence of rare events. We prove the upper bound by adapting Belovs's learning graph algorithm for $k$-distinctness \cite{kdist_learning_graphs_12}.
Finally, assuming a plausible conjecture concerning only cycle finding, we show that the lower bound can be lifted to an essentially tight lower bound on the quantum query complexity of $k$-distinctness, which is a long-standing open question.
\end{abstract}

\newpage

\thispagestyle{empty}
\tableofcontents

\newpage
\pagenumbering{arabic}
\setcounter{page}{1}

\section{Introduction}
The study of graph problems forms a cornerstone of research in theoretical computer science. These problems have been studied when the input data structure is the \emph{adjacency matrix} or \emph{adjacency list} of the graph. In the adjacency matrix model of an $n$-vertex  undirected simple graph, one is given the presence or absence status of every one of the $\binom{n}{2}$ edges in an $n\times n$ binary matrix. In the adjacency list model, the input is given as $n$ lists of neighbors, one for each of the $n$ vertices. Both models build structural cues into the input: an adjacency matrix assigns each vertex a given row and column, while an adjacency list uses vertices to order edges. These cues allow for fast computation but are arguably not intrinsic to the underlying graph.

In this work, we study graph problems in the edge list model, a minimalist model where the graph is given as an unordered multiset of edges. When computing on edge lists, algorithms must build the graph’s structure for themselves and not rely on any input cues. Thus, studying computation on edge lists offers a new window into the relationship between a graph problem's \emph{structure} and \emph{complexity}.
\begin{center}
\vspace{3mm}
\begin{minipage}{0.2\textwidth}
\begin{tikzpicture}[scale=0.55, every node/.style={circle, draw,  minimum size=0.4cm, inner sep=0}]
    \node (1) at (0, 2) {1};
    \node (2) at (2, 2) {2};
    \node (3) at (0, 0) {3};
    \node (4) at (2, 0) {4};
    
    \draw (1) -- (2);
    \draw (3) -- (4);
    \draw (1) -- (4);
    \draw (1) -- (3);
    \draw (2) -- (4);
\end{tikzpicture}
\end{minipage}
\begin{minipage}{0.1\textwidth}
\vspace{-3.3mm}
\begin{alignat*}{2}
    &&\text{Edge list:} &\quad \stackrel{\{1,2\}}{\rule{0.8cm}{0.5pt}} \ \stackrel{\{3,4\}}{\rule{0.8cm}{0.5pt}}  \
    \stackrel{\{1,4\}}{\rule{0.8cm}{0.5pt}} \ \stackrel{\{1,3\}}{\rule{0.8cm}{0.5pt}} \, \stackrel{\{2,4\}}{\rule{0.8cm}{0.5pt}}
    \\
    &&\text{Adjacency list:} &\quad 1\colon \stackrel{2}{\rule{0.3cm}{0.5pt}} \, \stackrel{4}{\rule{0.3cm}{0.5pt}} \, \stackrel{3}{\rule{0.3cm}{0.5pt}} \quad 
    2 \colon \stackrel{4}{\rule{0.3cm}{0.5pt}} \, \stackrel{1}{\rule{0.3cm}{0.5pt}} \quad 3 \colon \stackrel{1}{\rule{0.3cm}{0.5pt}} \, \stackrel{4}{\rule{0.3cm}{0.5pt}} \quad 4 \colon \stackrel{1}{\rule{0.3cm}{0.5pt}} \, \stackrel{3}{\rule{0.3cm}{0.5pt}} \, \stackrel{2}{\rule{0.3cm}{0.5pt}}
\end{alignat*}
\end{minipage}

\vspace{2mm}\textit{An edge list can be seen as a shuffled adjacency list: both inputs specify the same displayed graph.}
\end{center}

We study the edge list model in the context of \emph{quantum query complexity}. 
In quantum query complexity, we assume access to a quantum computer and characterize the complexity of a problem by the number of times the computer queries an input edge list.

Quantum algorithms on edge lists has previously been studied in the area of streaming algorithms. For example, \cite{kallaugher_streaming} showed that triangle counting can be solved by a quantum streaming algorithm using less space than any classical rival. However, the quantum computational model employed differs significantly from the quantum query model that we study. In the former case, the edge list is still queried classically --- only the memory is quantum --- whereas we assume the edge list can be queried in \emph{quantum superposition}. As we will see in this work, this change of model leads to new and interesting questions.

Inspired by the fact that the triangle problem in the adjacency matrix model has driven innovations in quantum query complexity \cite{triangle_magniez_07,belovs_constant_sized_12,nested_quantum_walk_2013,improved_triangle_magniez_13,legall_triangle_14,yu_bendavid_hypergraph_24} over many years, we were driven to study the same problem but in the edge list model. Specifically, we study three variants of the triangle problem.
\begin{itemize}
    \item $\triedge$ asks whether there exists a triangle containing a target edge and raises general questions about how a problem's complexity increases if its input is \emph{hidden} among irrelevant data. For the specific problem, the irrelevant data consists of those edges in the input that are not incident to the target edge.
    \item $\trivertex$ asks whether there exists a triangle containing a target vertex and raises general questions about how a problem's complexity increases if its input is \emph{shuffled}. For the specific problem, two parts of the input are shuffled: one part containing edges incident to the target vertex and the other part containing the remaining edges. The notion of shuffling is also inherent in the comparison between the edge list and adjacency list models.
    \item $\tri$ asks whether a triangle exists and bridges the $3$-distinctness and $3$-sum problems. The latter two problems have been extensively studied, both in the worst-case setting by complexity theorists seeking to connect problem structure with algorithmic complexity \cite{kdist_learning_graphs_12,spalek_belovs_13,belovs_childs_time_efficient_3dist_13,certificate_structure_belovs_2014,apers_testing_25} and in the average-case setting by cryptographers seeking to understand the security of cryptosystems \cite{ksum_crypto,liu_zhandry_19}.
\end{itemize}
We provide tight or nearly tight results for these problems as well as some first answers to the general questions they raise. Furthermore, we generalize our results for $\tri$ to $k$-cycle finding, a problem motivated by its close ties to the well-studied problem of $k$-distinctness.

\subsection{Our results}
In the following, we write $Q(\cdot)$ for the worst-case (bounded-error) quantum query complexity, $\mathbb{N}$ for the set of positive integers, and $[n]\coloneqq \{1,\dots,n\}$ for every $n\in \mathbb{N}$. 

\paragraph{$\triedge$ and Hiding.} In $\triedge$, when searching for a triangle containing a target edge $\{u,v\}$, any edge in the input of the form $\{u',v'\}$ such that $\{u',v'\} \cap \{u,v\} = \emptyset$ can be viewed as irrelevant data that hides the relevant part of the input. This observation motivates a general definition of hiding.

\begin{definition}[Hiding transform]\label{def:hiding}
Let $a,b$ be integers with $b\geq a\geq 1$. Let $\tildeD,\Sigma$ be finite non-empty sets with $\tildeD\subseteq \Sigma^a$. For a function $f\colon \tildeD  \to \{0,1\}$, we define the \emph{hiding transform} of $f$ to be 
    \begin{equation}
    \hide_b[f] \colon D \subseteq (\Sigma \cup \{*\})^b \to \{0,1\},
    \end{equation}
where $*$ is a symbol outside of $\Sigma$, the set $D$ consists of all strings $y \in (\Sigma \cup \{*\})^b$ with exactly $a$ non-$*$ symbols such that the length-$a$ subsequence\footnote{In general, for integers $b\geq a \geq 1$, $\Sigma$ a finite non-empty set, and a string $x\in \Sigma^b$: a length-$a$ subsequence of $x$ is a string of the form $x_{i_1}x_{i_2}\dots x_{i_a}$ for some integers $1\leq i_1<i_2<\dots<i_a \leq b$.} $\tilde{y}$ formed by those non-$*$ symbols is an element of $\widetilde{D}$, and $\hide_b[f](y)$ is defined to be $f(\tilde{y})$.
\end{definition}

Then $\triedge$ can be seen as $\hide_m[\ED_d]$, where $\ED_d$ is the element distinctness function that decides whether an input string of length $d$ has two distinct positions containing the same symbol. We show the following.\footnote{We use the notation $\tildeO,\tildeOmega,\tildeTheta$ to denote big-$O$,$\Omega$, and $\Theta$ up to poly-logarithmic factors.}
\begin{theorem}[Informal version of \cref{prop:hided_ED,prop:triangle_edge}]
    $Q(\hide_m[\ED_d]) = \tildeTheta(\sqrt{m}d^{1/6})$. Consequently,
    \begin{equation}
        Q(\triedge) = \tildeTheta(\sqrt{m} d^{1/6}),
    \end{equation}
    where the input is a length-$m$ edge list and the target edge has $d$ neighboring edges.\footnote{ We assume $d$ is known in advance. This is for simplicity in presentation: by using approximate counting to estimate $d$, it can be seen that the bound holds even if $d$ is not known in advance. The same comment applies to $\trivertex$.}
\end{theorem}

The study of $\triedge$ naturally led us to investigate how $Q(f)$ relates to $Q(\hide_b[f])$ for arbitrary $f\colon \Sigma^a \to \{0,1\}$. As a first step in this direction, we show $Q(\hide_b[f]) = \tildeTheta(\sqrt{b/a} \cdot Q(f))$ for every symmetric $f\colon \{0,1\}^a \to \{0,1\}$, where ``symmetric'' refers to ``symmetric under permuting the \emph{positions} of input symbols'' throughout this work. The symmetry condition is necessary. For example, let \(f:\{0,1\}^a\to\{0,1\}\) be the dictator function \(f(x)=x_1\). Then
\(Q(f)=1\), but $\hide_b[f]$ requires finding the first non-\(\ast\) symbol among the \(b\) positions. On the subdomain where the last \(a-1\) positions are
fixed non-\(\ast\) symbols and exactly one of the first \(b-a+1\) positions is
non-\(\ast\), computing \(\hide_b[f]\) contains unstructured search on
\(b-a+1\) items as a subproblem. Hence
\begin{equation}
    Q(\hide_b[f])=\Omega(\sqrt{b-a+1}),
\end{equation}
which is not \(O(\sqrt{b/a}\,Q(f))\) in general.
On the other hand, we conjecture that the result holds for every symmetric $f\colon \Sigma^a \to \{0,1\}$, even if $\abs{\Sigma}>2$. We support this conjecture by showing that its randomized analogue is true in \cref{prop:symmetric_partial}.

\paragraph{$\trivertex$ and Shuffling.} 
We show in \cref{prop:triangle_vertex_upper_bound,prop:triangle_vertex_lower_bound} that the quantum query complexity of $\trivertex$, when the input is a length-$m$ edge list and the target vertex $u$ has degree $d$ satisfies
\begin{equation}
    \Omega(\sqrt{m/d} \cdot Q(\kdist{3}_d)) \leq Q(\trivertex) \leq O(\sqrt{m}d^{1/4}),
\end{equation}
where $\kdist{3}_d$ is the $3$-distinctness function that decides whether an input string of length $d$ has three distinct positions containing the same symbol. Note that $Q(\kdist{3}_d)$ is between $\Omega(d^{2/3})$ -- because it is no easier than $\ED_d$ (which could be also called $\kdist{2}_d$) -- and $O(d^{5/7})$ \cite{kdist_learning_graphs_12}.

When seeking a triangle containing the target vertex $u$, it is helpful to think of the input edge list as consisting of two parts: part $A$ containing edges that are incident to $u$, and part $B$ containing edges that are not incident to $u$. We do not a priori know where parts $A$ and $B$ are but neither part contains irrelevant data --- indeed two edges of the triangle must come from $A$ and one edge from $B$ (if it exists). The \emph{shuffling} of parts $A$ and $B$ contributes to the hardness of $\trivertex$.

The above discussion and the comparison of the edge list model with the adjacency list and adjacency matrix models naturally led us to investigate how shuffling a function's input affects its complexity. We formalize and study two concrete versions of this problem.

\underline{(I) Shuffled functions}. As discussed, the edge list can be seen as a ``shuffled version'' of the adjacency list. We now formalize this notion. 
    \begin{definition}[Shuffling transform]\label{def:shuffling}
            Let $n\in \mathbb{N}$. Let $\tildeD,\Sigma$ be finite non-empty sets with $\tildeD\subseteq \Sigma^n$. For a function $f\colon \widetilde{D} \to \{0,1\}$, we define the \emph{shuffling transform} of $f$ to be
    \begin{equation}
        \shuffle[f] \colon D \subseteq (\Sigma\times [n])^n  \to \{0,1\}, \quad \text{where}
    \end{equation}
    \begin{enumerate}
        \item the set $D$ consists of all elements $x\in (\Sigma\times [n])^n$ satisfying $x = ((v_1,\pi(1)),\dots,(v_n,\pi(n)))$ for some bijection $\pi\colon [n]\to [n]$ such that $(v_{\pi^{-1}(1)},v_{\pi^{-1}(2)},\dots,v_{\pi^{-1}(n)})\in \widetilde{D}$.
    \item $\shuffle[f](x)\coloneqq f(v_{\pi^{-1}(1)},v_{\pi^{-1}(2)},\dots,v_{\pi^{-1}(n)})$.
    \end{enumerate}
    \end{definition}
In other words, the inputs to $\shuffle[f]$ are shuffled versions of the inputs to $f$ such that each symbol of the input to $\shuffle[f]$ additionally contains its position pre-shuffling. Of course, every (worst-case) query complexity measure of $\shuffle[f]$ is at least that of $f$, since $\shuffle[f]$ contains a copy of $f$ under the restriction of the domain $D$ to those inputs of the form $(x_1,1),\dots, (x_n,n)$, where $x = x_1 \dots x_n \in \widetilde{D}$. If $f$ is symmetric, then every query complexity measure on $f$ is the same as that measure on $\shuffle[f]$ since the algorithm computing $\shuffle[f]$ could simply ignore the second coordinates of the input. If $f$ is not symmetric, then $Q(\shuffle[f])$ could be much larger than $Q(f)$: the dictator function witnesses a $\Omega(\sqrt{n})$ versus $1$ separation.

Given the above discussion, the interesting question becomes: how does the separation between $Q(\shuffle[f])$ and $Q(f)$ depend on  ``how symmetric'' $f$ is? Natural (partially) symmetric $f$ arises from graph properties. 
We show that there can be massive separations between $Q(\shuffle[f])$ and $Q(f)$ even if $f$ has significant symmetry by being defined as a graph property in either (i) the adjacency list model (exponential separation) or (ii) adjacency matrix model (unbounded separation). 

Since an edge list can be interpreted as a shuffled adjacency list, result (i) can be interpreted as ``quantum computers can compute some graph property exponentially faster given an adjacency list instead of an edge list''; the proof of result (ii) can also be adapted to show ``quantum computers can compute some graph property unboundedly faster given an adjacency matrix instead of an edge list''.
\begin{theorem}[Informal version of \cref{thm:shuffling-adj-list,prop:shuffling_adj_matrix}]
\leavevmode
    \begin{enumerate}
        \item There exists a (family of) graph property $\calP_1:A\to \{0,1\}$, where $A$ denotes a set of adjacency lists of size $n$ such that
        \begin{equation*}
            Q(\calP_1)=O(\polylog(n)) \quad \text{and} \quad Q(\shuffle[\calP_1])=n^{\Omega(1)}.
        \end{equation*}
        \item There exists a (family of) graph property $\calP_2:M\to\{0,1\}$, where $M$ denotes a set of $n\times n$ adjacency matrices such that 
        \begin{equation*}
            Q(\calP_2)=O(1) \quad \text{and} \quad  Q(\shuffle[\calP_2])=n^{\Omega(1)}.
        \end{equation*}
    \end{enumerate}
\end{theorem}

\underline{(II) Shuffled direct sum.} In $\trivertex$, the hardness arising from shuffling is intuitively not due to shuffling \emph{within} parts $A$ and $B$ but rather \emph{between} them. We formalize a toy version of this type of shuffling in the context of direct sum as follows.

\begin{definition}[Shuffled direct sum]\label{def:shuffled-direct-sum}
    Let $n,k\in \mathbb{N}$. Let $\Sigma$ be a finite non-empty set. For a function $f \colon \Sigma^n \to \{0,1\}$, we define the $k$-shuffled direct sum of $f$ to be 
    \begin{equation}
        \shuffle^k[f]\colon D \subseteq  (\Sigma\times [k])^{kn} \to \{0,1\}^k, \quad \text{where}
    \end{equation}
\begin{enumerate}
    \item the set $D$ consists of all elements $x\in (\Sigma\times [k])^{kn}$ satisfying
    $x = ((v_1,c_1),\dots, (v_{kn},c_{kn}))$ such that: for all $j\in [k]$, there are exactly $n$ indices $i\in [kn]$ with $c_i  = j$.
    \item $\shuffle^k[f](x)$ is defined to equal $(f(v^{(1)}),\dots,f(v^{(k)}))$, where $v^{(j)}$ is the subsequence of $v\coloneqq v_1\dots v_{kn}$ indexed by those $i\in [kn]$ such that $c_i = j$.
\end{enumerate}
\end{definition}

    We have $Q(\shuffle^k[f]) \geq \Omega(kQ(f))$, since computing $\shuffle^k[f]$ is at least as hard as computing $k$ independent copies of $f$, which costs $\Omega(kQ(f))$ queries by a well-known direct-sum theorem for quantum query complexity \cite{ambainis_direct_sum,reichardt_direct_sum}.
This direct sum lower bound can be far from tight for $\shuffle^k[f]$: again consider $f$ being the dictator function. At first sight, computing $\shuffle^k[f]$ seems harder than computing the $k$ independent copies of $f$ even if $f$ is symmetric. Perhaps counterintuitively, we show in \cref{prop:shufflek} that this is not the case if $f$ is symmetric and has Boolean domain $\{0,1\}^n$: for such $f$, we show $Q(\shuffle^k[f]) \leq O(kQ(f))$. In \cref{conj:shuffling_composition}, we conjecture that the same holds for symmetric $f$ with non-Boolean domain. 

\paragraph{Triangle Finding.}
Having discussed $\triedge$ and $\trivertex$, we now turn to $\tri$, where we simply ask whether a triangle exists\footnote{While $\tri$ is a decision problem, it is essentially equivalent to its search variant via a simple reduction as explained in the preliminaries section. Therefore, we often do not distinguish between the search and decision problems.}  without placing any constraints. In the edge list model, we find that $\tri$ bridges the $\kdist{3}$ and $\ksum{3}$ problems since its structure lies between theirs. Recall that in $\kdist{3}$, the goal is to decide whether the input contains three repetitions of the same symbol. In $\ksum{3}$, the input contains symbols from some abelian group and the goal is to decide if there are three symbols that sum to the zero-element of the group.

Intuitively, $\tri$ should be no easier than $\kdist{3}$ and no harder than $\ksum{3}$ by considering its ``certificate structure'' \cite{certificate_structure_belovs_2014}. In $\kdist{3}$, once we have found one symbol of a $1$-certificate, the next symbol for completing the certificate is determined. In $\ksum{3}$, once we have found one symbol of a $1$-certificate, the next symbol can be arbitrary. In comparison, in $\tri$, once we have found one edge $(i,j)$ of a $1$-certificate, the next edge must be incident to one of $i$ or $j$ so it is neither determined nor arbitrary.  In \cref{prop:two_way_reduct}, we show that this intuition is formally correct by reducing $\kdist{3}$ to $\tri$ and $\tri$ to $\ksum{3}$.

Our first main result is stated informally below. 
\begin{theorem}[Informal version of \cref{prop:triangle_search,cor:worst-case-search}]\label{prop:ftriangle_lower_intro}
    Given a length-$m$ edge list on $n = \Theta(m)$ vertices, the quantum query complexity of finding a triangle is
    \begin{equation}\label{intro_tri_query_complexity_result}
        Q(\tri)\geq \widetilde{\Omega}(m^{5/7}).
    \end{equation}
    Furthermore, the lower bound holds even if the edge list is uniformly random, which corresponds to a random sparse graph. As a corollary, \cref{intro_tri_query_complexity_result} holds even if the edge list is promised to have a maximum degree at most $O(\log(m)/\log\log(m))$.
\end{theorem}

We find \cref{prop:ftriangle_lower_intro} interesting for several reasons. Firstly, proving the theorem pushed us to develop new techniques for the recording query framework, which we believe to be of independent interest. The recording query framework was pioneered by Zhandry \cite{zhandry_compressed} to prove the security of cryptosystems against quantum adversaries. There is now a growing line of work developing new techniques within the recording query framework, including \cite{liu_zhandry_19,chung_compressed_21, Hamoudi_2023, beame_matrix_24, majenz_permutationoracles_24,ma_huang_2024,carolan_compressed_25}. Of particular relevance to our work is \cite{Hamoudi_2023} by Hamoudi and Magniez, which generalized Zhandry's original framework in order to lower bound the quantum query complexity of finding collision pairs. 

This framework is particularly well suited for proving average-case quantum query lower bounds, where the input is sampled from certain types of distributions. The framework can also yield \emph{optimal} worst-case lower bounds if the worst-case distribution is of a type it can handle --- this is the case for triangle finding on graphs of low maximum degree as we will see. Ideally, to prove our lower bound, we would like to consider a hard distribution supported \emph{only} on graphs with low max-degree. The low max-degree condition makes finding a cycle harder and hence proving the lower bound easier. To see this, consider the contrapositive in the special case of triangle finding. If there were no degree bound, then an edge queried at step $t$ could help the algorithm find a very large number of \emph{wedges}, i.e., length-$2$ paths. In the extreme case, $\Omega(t)$ wedges could be found. As each wedge can be completed to a triangle, the more wedges the algorithm finds, the easier triangle finding becomes. 

Unfortunately, a product distribution on edge lists cannot be entirely supported on those with low max-degree. In other words, while sampling a graph with high max-degree is a rare event, it is always \emph{possible}. The main technical contribution of our work is to develop tools to handle such rare events within the recording query framework. We name these tools the Mirroring and Exclusion lemmas. They can be used in conjunction to isolate and bound the contribution to the algorithm's success probability when rare events occur.
These tools extend the recording query framework to treat any non-product distribution that results from conditioning a product distribution on the absence of rare events.

\cref{prop:ftriangle_lower_intro} is also interesting because the lower bound curiously matches the best-known upper bound on the quantum query complexity of $3$-distinctness, which is $Q(\kdist{3}_m) \leq O(m^{5/7})$, up to logarithmic factors. This upper bound was first proven by Belovs \cite{kdist_learning_graphs_12} over a decade ago using a ground-breaking learning graph algorithm \cite{belovs_constant_sized_12}. There have been no improvements since. On the other hand, the best-known lower bound on $Q(\kdist{3}_m)$ is \emph{still} $\Omega(m^{2/3})$, which is simply inherited from the element distinctness lower bound of Aaronson and Shi \cite{aaronson_shi} proven over two decades ago. This is despite many attempts at raising the lower bound using the polynomial method \cite{bkt_polynomial,mande_kdist}, adversary method \cite{spalek_adversary_2013,rosmanis_adversary_2014}, and even recording queries (equivalently, compressed oracle) method~\cite{liu_zhandry_19}. In \cref{tab:complexity}, we summarize these results together with those to be discussed later. Note that $\kcycle{3}$ is equivalent to $\tri$.

\begin{table}[h!]
\centering
\resizebox{\textwidth}{!}{
\begin{tabular}{@{}lccc@{}}
\toprule
\textbf{Problem} & \textbf{$\kdist{k}$ Lower Bound} & \textbf{$\kdist{k}$ Upper Bound} & \textbf{$\kcycle{k}$ Tight Bound} \\
\midrule
\textbf{$k=3$} & $\Omega(m^{2/3})$~\cite{aaronson_shi} & $O(m^{5/7})$~\cite{kdist_learning_graphs_12} & $m^{5/7\pm o(1)}$ (this work)\\
\addlinespace
\textbf{$k\geq 4$} & $\tildeOmega(m^{3/4-1/(4k)})$~\cite{mande_kdist} & $O(m^{3/4-1/(2^{k+2}-4)})$~\cite{kdist_learning_graphs_12} & $m^{3/4-1/(2^{k+2}-4)\pm o(1)}$ (this work) \\
\bottomrule
\end{tabular}
}
\caption{\small Best known quantum query complexity bounds for $\kdist{k}$ and $\kcycle{k}$ on length-$m$ inputs. For $\kcycle{k}$, we assume the input has max-degree $\leq O(\log(m)/\log\log(m))$. For $k=3$, the  $\kdist{k}$ lower bound is inherited from that for $k=2$ \cite{aaronson_shi}. Can $\kcycle{k}$ help us close the long-standing gaps for $\kdist{k}$?}\label{tab:complexity}
\end{table}

We strongly believe that the matching of the lower bound in \cref{prop:ftriangle_lower_intro} and the best-known upper bound on $Q(\kdist{3}_m)$ is \emph{not a coincidence}. This is because the ``collision structure'' of a uniformly random edge list on $n$ vertices containing $m = \Theta(n)$ edges resembles that of a candidate worst-case distribution on inputs to $\kdist{3}_m$. By collision structure of an edge list, we mean the number of wedges and triangles it contains; by collision structure of an input to $3$-distinctness, we mean the number of $2$-collisions and $3$-collisions it contains, where an $l$-collision refers to a length-$l$ tuple of indices at which the input contains the same symbol. (For more details about the resemblance, see the overview of cycle finding below.) Indeed, the similarities between the collision structures were so striking that we were led to ask whether Belovs's learning graph algorithm for $3$-distinctness could be adapted to provide a matching upper bound to \cref{prop:ftriangle_lower_intro}. As our second main result, we answer this question affirmatively.
\begin{theorem}[Informal version of \cref{thm:worst-case-decision,cor:worst-case-decision}]\label{prop:ftriangle_upper_intro}
    Given a length-$m$ edge list on $n = \Theta(m)$ vertices with maximum degree at most $O(\log(m)/\log\log(m))$, we construct a quantum algorithm witnessing
    \begin{equation}
        Q(\tri)\leq m^{5/7+o(1)}.
    \end{equation}
    In particular, the upper bound holds if the edge list is uniformly random, which corresponds to a random sparse graph. Moreover, the algorithm is an adaptation of Belovs's learning graph algorithm for $3$-distinctness as developed in \cite{kdist_learning_graphs_12}.
\end{theorem}

Taken together, \cref{prop:ftriangle_lower_intro,prop:ftriangle_upper_intro} give the \emph{first} example of a problem for which Belovs's learning graph algorithm in \cite{kdist_learning_graphs_12} is provably optimal, up to an $m^{o(1)}$ factor.\footnote{Note that the best known quantum query lower bound for $\kdist{k}$ for every $k\geq 4$ (see \cite{bkt_polynomial,mande_kdist}) is also polynomially far from the upper bound witnessed by Belovs's learning graph algorithm.} \cref{prop:ftriangle_lower_intro} also suggests that any quantum algorithm that polynomially improves the best-known $O(m^{5/7})$ upper bound on $Q(\kdist{3}_m)$ must exploit more than just the collision structure of the problem, which seems improbable to us. 

\paragraph{Cycle Finding.} Our results above suggest that triangle finding in a graph of low maximum degree is closely tied to $3$-distinctness. Indeed, they share essentially the same quantum query upper bound and have quite similar quantum algorithms. The $3$-distinctness problem generalizes to a well-studied family of problems known as $k$-distinctness, where $k\geq 3$ is an integer (treated as a constant) and the goal is to decide whether an input list contains $k$ distinct positions with the same symbol. Is there a family of problems generalizing triangle finding in the edge list model that can be similarly tied to $k$-distinctness? 

The answer is yes with $k$-\emph{cycle} finding being the sought for family. To see why, let us return to the concept of collision structure mentioned above. Let $\kdist{k}_m$ denote the function computing $k$-distinctness on length-$m$ inputs. A candidate worst-case distribution on inputs to $\kdist{k}_m$ is the uniform distribution over inputs containing exactly one $k$-collision and $\Omega(m)$ many $(k-1)$-collisions, cf.~\cite[Section 5]{quantum_walk_ed} and~\cite[Section 3]{belovs_childs_time_efficient_3dist_13}. The intuition is that the many $(k-1)$-collisions present many dead-ends that effectively hide the location of the $k$-collision. In comparison, writing $P(n,k)\coloneqq n!/(n-k)!$, a uniformly random length-$m$ edge list on $n=\Theta(m)$ vertices contains
\begin{alignat*}{3}
&\Theta(P(n,k)\cdot (m/n^2)^{k-1}) &&= \Theta(m) \ &&\text{length-$(k-1)$ paths} \\
\text{and} \quad  &\Theta(P(n,k)\cdot (m/n^2)^k) &&= \Theta(1) \ &&\text{length-$k$ cycles}.
\end{alignat*}

    Since a length-$(k-1)$ path is analogous to a $(k-1)$-collision and a length-$k$ cycle is analogous to a $k$-collision, the collision structure of the uniformly random edge list resembles that of the aforementioned worst-case distribution on inputs to $k$-distinctness. This is also a good point to remark that the uniform distribution over inputs to $k$-distinctness, which is studied in \cite{liu_zhandry_19}, \emph{cannot} be made to resemble the worst-case distribution (for any setting of input length and alphabet size) because it has the property that the number of $(k-1)$-collisions cannot be too far from the number of $k$-collisions.

Our main result for cycle finding generalizes our results for triangle finding as follows.
\begin{theorem}[Informal version of \cref{cor:kcycle_lower_bound,cor:worst-case-decision-k}]
    Given a length-$m$ edge list on $n=\Theta(m)$ vertices with maximum degree at most $O(\log(m)/\log\log(m))$, the quantum query complexity of finding a $k$-cycle is
    \begin{equation}
        Q(\kcycle{k})= m^{3/4-1/(2^{k+2}-4)\pm o(1)}.
    \end{equation}
    Furthermore, this holds if the edge list is uniformly random, which corresponds to a random sparse graph.
\end{theorem}
The upper bound is \emph{exactly} the same as the best known upper bound for $k$-distinctness up to an $m^{o(1)}$ factor and is again shown by adapting Belovs's learning graph algorithm for $k$-distinctness in \cite{kdist_learning_graphs_12}. We have summarized the state of affairs in \cref{tab:complexity}.

The comparison of collision structures above suggests that the hardness of $k$-cycle and $k$-distinctness spring from a common source, which hints at the possibility of lifting our nearly tight lower bound for $k$-cycle to a nearly tight lower bound for $k$-distinctness. In fact, we propose \cref{conj:graph_property_partition}, stating that the quantum query complexity of $k$-cycle should not decrease too much under a restriction of the input edges, and prove in \cref{prop:conjecture_implies_lower} that it enables such lifting.

\paragraph{Note added.} Subsequent to this work, Belovs~\cite{belovs_kdistinctness_26} proved a tight quantum query lower bound for \(k\)-distinctness, using a framework inspired by Zhandry's compressed oracle technique and independently of the lifting conjecture discussed above.

\subsection{Technical overview}\label{sec:tech-overview}
We now highlight some of our work's main technical contributions. We inherit the notation above.

\paragraph{$\triedge$ and Hiding.}
To lower and upper bound $Q(\hide_m[\ED_d])$, our main observation is that the problem ``self-reduces'' to a more structured version of itself, $\hide_m'[\ED_d]$, whose query complexity is easier to characterize. 
The inputs to $\hide_m'[\ED_d]$ are restricted to have the form of $d$  blocks each of length $m/d$,\footnote{We may assume $m/d \in \mathbb{Z}$ without loss of generality. We will not make further remarks like this in the technical overview.} where each block contains $(m/d -1)$ $*$s and exactly one non-$*$ symbol. Remarkably, imposing the block structure causes no decrease in hardness.

Observe that $\hide_m'[\ED_d]$ can be viewed as the composition of $\ED_d$ with the so-called $\pSearch$ function that extracts the unique non-$*$ symbol from $m/d$ symbols. The latter function has query complexity $\Theta(\sqrt{m/d})$ and a composition theorem \cite[Theorem 9]{brassard_psearch} yields $Q(\hide_m[\ED_d]) \geq Q(\hide_m'[\ED_d]) \geq \Omega(\sqrt{m/d} \cdot d^{2/3}) = \Omega(\sqrt{m}d^{1/6})$. To upper bound $Q(\hide_m[\ED_d])$, we consider a quantum algorithm that first randomly permutes the positions of a given input string $x$. 
A standard probability argument -- like that used to bound the maximum load in a balls-into-bins experiment -- implies that the resulting string $\tilde{x}$ is highly likely to be in a block form similar to inputs of $\hide_m'[\ED_d]$, except each block may contain up to $\log(d)$ non-$*$ symbols. Since the number of non-$*$ symbols in each block is so small, we simply Grover search for all of them on the fly while running the quantum algorithm for $\ED_{d\log(d)}$ on $d\log(d)$ symbols. 
This algorithm has query complexity $\tildeO(\sqrt{m} d^{1/6})$.

By considering the block structure, we also see $Q(\hide_m[f]) = \Omega(\sqrt{m/d} \cdot Q(f))$ for \emph{every} $f\colon \Sigma^d \to \{0,1\}$. We show $Q(\hide_m[f]) =\tildeO(\sqrt{m/d} \cdot Q(f))$ for symmetric $f\colon \{0,1\}^d \to \{0,1\}$ using the characterization of the optimal quantum query algorithm for such functions in \cite{polynomial_2001}. We show $R(\hide_m[f]) = O((m/d)\cdot R(f))$ for symmetric $f\colon \Sigma^d \to \{0,1\}$, where $R(\cdot)$ denotes the worst-case (bounded-error) randomized query complexity, using the characterization of the optimal randomized query algorithm for such functions in \cite{baryoussef_2001}.

\paragraph{$\trivertex$ and Shuffling.} To obtain  $Q(\trivertex) \leq O(\sqrt{m} d^{1/4})$ in \cref{prop:triangle_vertex_upper_bound}, we use a quantum walk algorithm that walks on the Hamming graph with vertices labeled by $r \coloneqq \ceil{d^{3/4}}$  positions from part $A$ (the part containing edges incident to the target vertex) of the input. Since we do not a priori know where part $A$ is, we perform amplitude amplification in both the setup and update steps of the quantum walk to keep the walk on part $A$. Importantly, this uses the fact that the underlying random walk is uniformly random on the Hamming graph.
To lower bound $Q(\trivertex) \geq \Omega(\sqrt{m/d} \cdot Q(\kdist{3}))$, we give a reduction from $\hide_m[\kdist{3}]$ to $\trivertex$ in \cref{prop:triangle_vertex_lower_bound}.

We obtain an exponential separation between $Q(f)$ and $Q(\shuffle[f])$ when $f$ is defined by the graph property $\calP$ from \cite[Section 6]{symmetries_bcgkpw} in the adjacency list model as follows. \cite{symmetries_bcgkpw} showed that computing $f$ witnesses an exponential separation between randomized and quantum query complexities. But the quantum query complexity of computing $\shuffle[f]$ is polynomially related to its randomized query complexity since the problem is symmetric \cite{chailloux_symmetric}.

If $f$ is defined by a graph property in the adjacency matrix model, the above argument cannot work since \cite{symmetries_bcgkpw} showed that the quantum query complexity of computing $f$ is polynomially related to its randomized query complexity. Nonetheless, we found that an unbounded separation between $Q(f)$ and $Q(\shuffle[f])$ can be witnessed by the following Majority-of-Majority function on a restricted partial domain, that we name $\SMAJ$:\footnote{We chose this name because $\Sigma$ resembles a rotated M and $\SMAJ_n$ is (a restriction of) the composition of two $\MAJ_n$s.}

\begin{definition}\label{def:smaj}
\begin{equation}
     \SMAJ_n \colon D_0\dot{\cup} D_1 \subseteq \{0,1\}^{n^2} \to \{0,1\},
\end{equation}
where $\SMAJ_n(x) = 0$ if and only if $x\in D_0$ and 
    \begin{enumerate}
    \item $x = (x_{1,1},x_{1,2},\dots, x_{n,n}) \in \{0,1\}^{n^2}$ is in $D_0$ if and only if there exists a subset $S\subseteq[n]$ of size $\geq 2n/3$ such that for all $i\in S$, $x^i \coloneqq (x_{i,1},\dots,x_{i,n})$ has Hamming weight $\abs{x^i} \geq 2n/3$ and for all $i\in [n]-S$,  $\abs{x^i} \leq n/3$.
    \item $x = (x_{1,1},x_{1,2},\dots, x_{n,n}) \in \{0,1\}^{n^2}$ is in $D_1$ if and only if there exists a subset $S\subseteq[n]$ of size $\leq n/3$ such that for all $i\in S$, $x^i \coloneqq (x_{i,1},\dots,x_{i,n})$ has Hamming weight $\abs{x^i} \geq 2n/3$ and for all $i\in [n]-S$, $\abs{x^i} \leq n/3$.
    \end{enumerate}
\end{definition}

In other words, inputs in $D_0$ have at least $2n/3$ ``dense'' rows, i.e., a substring of the form $x_{i,1}x_{i,2}\dots x_{i,n}$ for some $i\in [n]$, with at least $2n/3$ ones; and at most $n/3$ ``sparse'' rows with at most $n/3$ ones. In contrast, inputs in $D_1$ have at most $n/3$ dense rows and at least $2n/3$ sparse rows.

It is not hard to see that $R(\SMAJ_n) = O(1)$ as follows. For a given row, we can test whether it is dense or sparse using $O(1)$ queries. Since the fractions of dense blocks for inputs in $D_0$ and $D_1$ differ by a constant, we can distinguish between these cases using $O(1)$ queries. Therefore $Q(\SMAJ_n) \leq R(\SMAJ_n) = O(1)$.

In contrast, computing $\shuffle[\SMAJ_n]$ seems at least as hard as finding two distinct input symbols $(x_{i_1,j_1},(i_1,j_1))$ and $(x_{i_2,j_2},(i_2,j_2))$ that came from the same row pre-shuffling, i.e., $i_1=i_2$ but $j_1\neq j_2$, which is a collision-type problem. Formally, we prove $R(\shuffle[\SMAJ_n]) = \Omega(\sqrt{n})$ in \cref{prop:shuffling_adj_matrix} by showing that two particular distributions, one supported on the set $D_0$ and another on $D_1$, are hard to distinguish by any few-query randomized algorithm using a hands-on total variation distance argument. Therefore, $Q(\shuffle[\SMAJ_n]) = \Omega(R(\shuffle[\SMAJ_n])^{1/3}) = \Omega(n^{1/6})$, where the first equality uses \cite{chailloux_symmetric}, which applies since $\shuffle[\SMAJ_n]$ is symmetric.

\paragraph{Triangle Finding (lower bound).}
We prove our lower bound on triangle finding (\cref{prop:ftriangle_lower_intro}) in Zhandry's recording query framework \cite{zhandry_compressed}. Following the framework, we define a ``progress quantity'' that tracks the progress the algorithm has made in ``recording'' the searched-for object in its internal memory. The progress quantity can be roughly thought of as the square root of the probability with which the quantum algorithm can find the searched-for object, where the probability is over randomness in \emph{both} the input distribution and the algorithm. The progress quantity depends on the number of queries the quantum algorithm makes. If this quantity is small after the last query, then the algorithm cannot find what it is searching for with high probability. 

Our proof has two steps:\footnote{The arguments here are better understood by considering $n$, which represents the number of vertices in the graph. But recall that \cref{prop:ftriangle_lower_intro} concerns the regime $m=\Theta(n)$, so all results here can also be expressed in terms of $m$.}
\begin{enumerate}
    \item we first show that the progress in recording much more than $r^*(t)\coloneqq t^{3/2}\log^2(n)/\sqrt{n}$ wedges in $t$ queries is negligible;
    \item then we show that, given we record $O(r^*(t))$ wedges in $t$ queries, the progress of recording a triangle at the $(t+1)$-th query increases by at most $O(\sqrt{r^*(t)}/n)$, which corresponds to the square root of the probability that a random edge completes one of the $r^*(t)$ recorded wedges to a triangle.
\end{enumerate}
Therefore, at the $T$-th query, the progress of recording a triangle is $\sum_{t=0}^T \sqrt{r^*(t)}/n$, which equals $O(T^{7/4}\log(n)/n^{5/4})$,
and is $o(1)$ unless $T \geq \Omega(n^{5/7}/\log^{4/7}(n))$. The ability to perform this type of step-by-step analysis is a known strength of the recording query framework.\footnote{To quantum query lower bound experts: the standard quantum adversary method \cite{A00,negative_weights_adv} is \emph{not} well-suited to performing this type of step-by-step analysis because it gives only weak lower bounds for small success probabilities, and step (ii) needs the progress in step (i) to be inverse-polynomially small to work. If we were forced to redo this analysis using the adversary method, we would have to switch to its \emph{multiplicative} version, see, e.g., \cite{ambainis_multiplicative,lee_roland_multiplicative,jeffery_zur_compressed_multiplicative_25}.} For example, it was exploited to great effect by Liu and Zhandry \cite{liu_zhandry_19} in proving their tight lower bound on the quantum query complexity of \emph{average-case} $k$-distinctness. 

What is new to our work is how we perform step (i) above. As previously discussed, the issue is that a newly queried edge could contribute to more than one wedge. Let us now see how this issue manifests itself at a technical level. We begin by following the recording queries framework and define a progress quantity $\Lambda_{t,r} \in [0,1]$ for integer $t,r$ with $t\geq 0$ where $\Lambda_{t,r}^2$ represents the probability a quantum query algorithm has recorded at least $r$ wedges immediately after the $t$-th query. Directly using existing techniques in the framework gives the following recurrence for $\Lambda_{t,r}$:
\begin{equation}\label{eq:naive_recurrence}
    \Lambda_{t,r} \leq \Lambda_{t-1,r} + O(\sqrt{t/n}) \cdot \Lambda_{t-1,r-t+1},
\end{equation}
where the factor $O(\sqrt{t/n}) = O(\sqrt{tn/n^2})$ arises as the square root of the probability that a randomly chosen edge is incident to one of the at most $t-1$ edges that can be recorded after the $(t-1)$-th query; the subscript $r-t+1 = r - (t-1)$ arises from the possibility of the new edge recorded at the $t$-th query contributing $t-1$ additional wedges. However, solving \cref{eq:naive_recurrence} leads to a trivial lower bound for triangle finding that does not even beat the $\Omega(m^{2/3})$ lower bound it inherits from element distinctness.

The main problem with \cref{eq:naive_recurrence} is the subscript $r-t+1$ on the second term on the right-hand side. However, the event it corresponds to seems unlikely to happen when the input is a sparse graph and $t$ is large: if the new edge contributes $t-1$ additional wedges, it must be incident to a degree-$\Omega(t)$ vertex recorded by the quantum query algorithm. Now, our input is a random sparse graph whose maximum degree is at most $O(\log(n)/\log\log(n)) \leq O(\log(n))$ with high probability, \emph{independent of $t$}. Does this property also hold for the internal memory of the quantum query algorithm doing the recording? Our first technical contribution, the Mirroring Lemma (\cref{lem:mirror}), answers this question affirmatively. This lemma allows us to transfer, or mirror, properties of the initial input distribution onto the internal memory of the quantum algorithm, independently of the value of $t$. Directly using this technique allows us to improve \cref{eq:naive_recurrence} to
\begin{equation}\label{eq:intermediate_recurrence}
    \Lambda_{t,r} \leq \Lambda_{t-1,r} + O(\sqrt{t/n}) \cdot \Lambda_{t-1,r-O(\log^2(n))} + \epsilon,
\end{equation}
where $\epsilon>0$ is a small number corresponding to the tail probability of the input graph having a vertex of degree $\Omega(\log(n))$.

Unfortunately, \cref{eq:intermediate_recurrence} still does not yield the desired result: to see this, note that the solution to a similar recurrence $A_{t,r} = A_{t-1,r} + p A_{t-1,r-1} + \epsilon$ (with $p,\epsilon \in [0,1]$ and boundary conditions $A_{0,0} = 1$ and $A_{0,r} = 0$ for all $r>0$) is $A_{t,r} = \binom{t}{r}p^r + \epsilon(1 + (1+p) + \dots + (1+p)^{t-1})$. Even for an exponentially small $\epsilon$, the term $\epsilon(1+p)^{t-1}$ blows up for large $t$. Our second technical contribution, the Exclusion Lemma (\cref{lem:exclusion_lemma}), allows us to overcome this problem. To employ this lemma, we introduce a new progress quantity called $\Lambda_{t,r}'$ that is defined like $\Lambda_{t,r}$ except we additionally require the quantum query algorithm to \emph{not} have recorded a degree-$\Omega(\log(n))$ vertex at any point before the $t$-th query. By definition, $\Lambda_{t,r}'$ satisfies recurrence \cref{eq:intermediate_recurrence} with $\epsilon$ set to $0$, that is,
\begin{equation}\label{eq:final_recurrence}
    \Lambda_{t,r}' \leq \Lambda_{t-1,r}' + O(\sqrt{t/n}) \cdot \Lambda_{t-1,r-O(\log^2(n))}'.
\end{equation}
The Exclusion Lemma allows us to upper bound $\Lambda_{t,r}$ by $\Lambda_{t,r}' + O(t \epsilon)$, where the second term no longer blows up for large $t$ and is easy to make negligible. Therefore, solving \cref{eq:final_recurrence} first for $\Lambda_{t,r}'$ and then using $\Lambda_{t,r} \leq \Lambda_{t,r}' + O(t \epsilon)$ yields the claimed result of step (i).

\paragraph{Triangle Finding (upper bound).}
We prove our upper bound on triangle finding (\cref{prop:ftriangle_lower_intro}) 
by adapting Belovs's learning graph algorithm for $3$-distinctness from \cite{kdist_learning_graphs_12}. Formally, a ``learning graph algorithm'' is a directed acyclic graph that encapsulates a solution to a semi-definite program that characterizes quantum query complexity \cite{reichardt_adv_tight,reichardt_direct_sum}. Our main adaptation of Belovs's algorithm pertains to its handling of so-called \emph{faults} in \cite[Section 6]{kdist_learning_graphs_12}. 

The notion of a fault is easier to explain in Jeffery and Zur's interpretation of Belovs's algorithm as a quantum walk \cite{jeffery_zur}.\footnote{\cite{jeffery_zur} goes much beyond merely interpreting Belovs's algorithm. However, in this paper, we will only use \cite{jeffery_zur} to aid our explanation of Belovs's algorithm.} The following explanation is based on \cite[Section 1.3]{jeffery_zur}. The quantum walk first creates a uniform superposition over subsets $R_1$ of indices of some size $r_1$, and queries all $r_1$ indices. Then, for each $R_1$ in the superposition, the algorithm creates a uniform superposition over all subsets $R_2$ (disjoint from $R_1$) of indices of some size $r_2$. But rather than querying every index in $R_2$, the algorithm \emph{only} queries those $i_2 \in R_2$ that have a \emph{match} in $R_1$, i.e., $x_{i_2} = x_{i_1}$ for some $i_1\in R_1$, where $x$ is the input. This significantly reduces query complexity by exploiting the structure of $3$-distinctness: any two unequal symbols could not be part of a $1$-certificate. Unfortunately, it also leads to the aforementioned faults. The issue is that when performing the update step of the quantum walk by adding a new index $j_1$ to $R_1$, we cannot afford the queries needed to update a corresponding $R_2$ by searching for a $j_2 \in R_2$ such that $x_{j_2} = x_{j_1}$ because $R_2$ was not fully queried. But if we do not search and there does exist $j_2\in R_2$ with $x_{j_2} = x_{j_1}$, then the set of queried indices in $R_2$ becomes incorrect, introducing a fault. 

In our setting, there can be more faults because ``matching'' in our case naturally needs to be redefined to mean: $i_2\in R_2$ matches with $i_1\in R_1$ if and only if  $x_{i_2}$ (which is an edge in our case) is \emph{incident} to $x_{i_1}$. In particular, $i_2$ could match $i_1$ even if $x_{i_2} \neq x_{i_1}$. However, the number of faults introduced is bounded above by using the maximum degree $d$ of the input, which for a random sparse graph satisfies $d\leq O(\log(n)/\log\log(n))$. Then we adapt the ``error-correcting'' technique of \cite{kdist_learning_graphs_12} to correct $O(d)$ faults by paying a multiplicative factor of $2^{O(d)} = n^{o(1)}$ on the quantum query complexity, which leads to the theorem. Along the way, we construct a learning graph algorithm that may be easier to understand than that in \cite{kdist_learning_graphs_12}; for example, our algorithm genuinely corresponds to a graph, unlike that in \cite{kdist_learning_graphs_12}.

\paragraph{Cycle Finding.} We prove our lower and upper bounds on $k$-cycle finding by generalizing our proofs in the $k=3$ (triangle) case.\footnote{While these generalized proofs fully generalize those in the case $k=3$, we encourage the interested reader to read the proofs in the $k=3$ case first as they contain most of the key ideas and should make the generalized proofs easier to follow.} Here is a sketch of how this works.

To prove the lower bound, we consider a uniformly random length-$m$ edge list on $n$ vertices such that $n=\Theta(m)$. Our proof proceeds in $(k-1)$ steps. We use the Mirroring Lemma in the first step for the same reason it was used in the triangle case. On the other hand, the role of the Exclusion Lemma is expanded ``by a factor of order $k$'' as it serves to ``glue'' together the $(k-1)$ steps. For $i\in \{1,\dots,k-2\}$, we show in the $i$-th step that it is hard for a $t$-query algorithm to record much more than $r_{i+1}(t)$ length-$(i+1)$ paths without having recorded at least $r_{j}(t)$ length-$j$ paths for some $1\leq j \leq i$. Here, the $r_l(t)$s are certain positive integers essentially \emph{defined} by the relation $r_l(t) = t\sqrt{r_{l-1}(t)/n}$ and boundary condition $r_1(t) = t+1$. The term $\sqrt{r_{l-1}(t)/n}$ corresponds to the square root of the probability that a single uniformly random edge $e$ extends one of the $r_{l-1}(t)$ length-$(l-1)$ paths to a length-$l$ path. (Note that the precise expression of this probability is $r_{l-1}(t)\cdot 2(n-l)/(n(n-1)/2)$, which is $\Theta(r_{l-1}(t)/n)$.) One technical challenge here (that was not present in the triangle case) is the need to account for the scenario of edge $e$ creating a length-$l$ path by joining together a length-$a$ path with a length-$b$ path such that $a,b$ are non-negative integers and $a+b+1 = l$ \emph{but} neither $a$ nor $b$ equals $0$. However, we find that this scenario can be neglected as its effect is dominated by that of the extension scenario, i.e., the scenario where $a$ or $b$ is $0$. After concluding the first $(k-2)$ steps, we show in the $(k-1)$-th step that it is hard for a $t$-query algorithm to record a $k$-cycle without having recorded at least $r_{k-1}(t)$ length-$(k-1)$ paths. Finally, we use the Exclusion Lemma $(k-1)$ times to glue together the results of all steps and show that it is hard for a few-query algorithm to record a $k$-cycle.

On the upper bound side, our algorithm for $k$-cycle remains an adaptation of Belovs's algorithm for $k$-distinctness. The adaptation again pertains to the handling of faults due to a natural redefinition of what matching means. By the degree bound, the maximum number of faults that can occur is fewer than $2d$ where $d$ is the maximum degree of the input edge list. We again adapt Belovs's error correcting technique to handle them by paying a multiplicative factor of $2^{O(dk)}$, which is $m^{o(1)}$ when $d\leq O(\log(m)/\log\log(m))$. 

Interestingly, the exact same $r_{l}(t)$s mentioned in the lower bound overview play a ``dual'' role in both Belovs's learning graph algorithm for $k$-distinctness and our adaptation of it to $k$-cycle: given $t$ queries, the learning graph can be thought of as ``recording'' $r_j(t)$ number of $j$-collisions (in the case of $k$-distinctness) or length-$j$ paths (in the case of $k$-cycle) at the $j$-th step, or ``stage'' as Belovs calls it. Therefore, our lower bound can be interpreted as saying that the learning graph algorithm for $k$-cycle is tight at every stage.

\section{Preliminaries}\label{sec:prelims}
\paragraph{Notation.} $\mathbb{N}$ denotes the set of positive integers. Notation such as $\mathbb{Z}_{\geq 0}$ denotes the set of non-negative integers, $\mathbb{R}_{>0}$ denotes the set of positive reals, for $a\geq 0$, $\mathbb{Z}_{[0,a]}$ denotes the set $\mathbb{Z}\cap[0,a]$, and so on. For $n\in \mathbb{N}$, $[n]$ denotes the set $\{1,\dots,n\}$. For two complex square matrices $X,Y$ of the same dimension, we write $X\leq Y$ to mean $Y-X$ is positive semidefinite, and $X\geq Y$ to mean $X-Y$ is positive semidefinite. We use $\id$ to denote an identity matrix of a context-appropriate dimension.

An alphabet is a finite non-empty set. For an alphabet $A$ and $k\in \mathbb{Z}_{\geq 0}$, $\binom{A}{k}$ denotes the set of all size-$k$ subsets of $A$. For $n\in \mathbb{N}$ and $\gamma \geq 0$, we write $\binom{n}{<\gamma}$ for the sum $\sum_{i \in \{0,1,\dots, n\} \cap [0,\gamma)} \binom{n}{i}$.  Given two disjoint sets $A,B$, we sometimes write $A\dot{\cup}B$ for their union, where the dot \emph{emphasizes} that $A$ and $B$ are disjoint. (If we write $A\cup B$ instead, $A$ and $B$ could still be disjoint.) For an alphabet $\Sigma$, we write $x\leftarrow \Sigma$ for sampling $x$ uniformly at random from $\Sigma$. For $x\in \Sigma^n$ and $a,b\in \mathbb{N}$, $x[a\twodots b]$ denotes the substring of $x$ from index $a$ to $b$ inclusive; if $b<a$, then this denotes the empty string. For $k, n\in \mathbb{N}$ with $k\geq 3$, we say $\{u_1,v_1\}, \dots, \{u_k,v_k\} \in \binom{[n]}{2}$ form a $k$-cycle if there exists distinct $w_1,\dots,w_k\in [n]$ such that $\{u_1,v_1\} = \{w_1,w_2\}$, $\dots$, $\{u_{k-1},v_{k-1}\} = \{w_{k-1},w_k\}$, and $\{u_k,v_k\} = \{w_k,w_1\}$. For $k,n,m\in \mathbb{N}$ with $k\geq 3$, we say $x\in \binom{[n]}{2}^m$ contains a $k$-cycle if there exist distinct $i_1,\dots,i_k \in [m]$ such that $x_{i_1},\dots,x_{i_k}$ form a $k$-cycle. The word ``triangle'' formally refers to a $3$-cycle.

For an alphabet $U$, we denote the symmetric group acting on $U$ by $\mathfrak S_U$. For $n\in \mathbb{N}$, we also use $\mathfrak S_n$ to denote the symmetric group acting on $[n]$. We say a function $f\colon E \subseteq \Sigma^n \to \{0,1\}$ is symmetric if $x_1,\dots,x_n \in E \implies x_{\sigma(1)},\dots, x_{\sigma(n)} \in E$ and  $f(x_1,\dots,x_n) = f(x_{\sigma(1)},\dots, x_{\sigma(n)})$ for all $x\in \Sigma^n$ and for all permutations $\sigma \in \mathfrak S_n$. Note that
we do \emph{not} require a symmetric $f$ to be invariant under permutations of the alphabet $\Sigma$.

The abbreviation ``wlog'' stands for ``without loss of generality''; ``s.t.'' stands for ``such that''; ``with high probability'' stands for ``with probability at least $1-10^{-10}$'' unless stated otherwise. The symbol $\log$ stands for the base-$2$ logarithm.

For a function $f\colon D \subseteq \Sigma^n \to \Gamma$, where $D,\Sigma,\Gamma$ are alphabets and $n\in \mathbb{N}$, and $\epsilon \in (0,1/2)$, we write $Q_\epsilon(f)$ and $R_\epsilon(f)$ for the quantum and randomized query complexities for $f$ with two-sided failure probability at most $\epsilon$, respectively. We write $Q(f)$ and $R(f)$ for $Q_{1/3}(f)$ and $R_{1/3}(f)$, respectively.

\paragraph{Function definitions.}
For $m,d,n,s,k\in \mathbb{N}$, $u,v\in [n]$ with $u\neq v$, $G$ a finite abelian group with identity element $0$, we define the following functions of interest in this work. 

\begin{enumerate}
    \item $\kdist{k}_m\colon [s]^m \to \{0,1\}$ denotes the $k$-distinctness function which is defined by $\kdist{k}_m(x) = 1$ if and only if there exist $k$ distinct indices $1\leq i_1<i_2<\dots<i_k \leq m$ such that $x_{i_1} = x_{i_2} = \dots = x_{i_k}$. We also write $\ED_m$ for $\kdist{2}_m$, commonly known as the element distinctness function.
    
    \item $\ksum{k}_m\colon G^m \to \{0,1\}$ denotes the $k$-sum function which is defined by $\ksum{k}_m(x) = 1$ if and only if there exist $k$ distinct indices $1\leq i_1<i_2<\dots<i_k \leq m$ such that $x_{i_1} + x_{i_2} + \dots + x_{i_k} = 0$.
    \item $\triedge^{\{u,v\}}_m \colon D\subseteq \binom{[n]}{2}^m \to \{0,1\}$ denotes the triangle edge function which is defined by $x\in D$ if and only if $x_i\neq x_j$ for all $i\neq j$, and $\triedge_m^{\{u,v\}}(x) = 1$ if and only if there exist distinct $i,j\in [m]$ such that $x_i$, $x_j$ and $\{u,v\}$ form a triangle. 
    \item $\triedge^{\{u,v\}}_{m,d} \colon D_d\subseteq \binom{[n]}{2}^m \to \{0,1\}$ is the restriction of $\triedge^{\{u,v\}}_m$  such that $x \in D_d$ if and only if $x_i\neq x_j$ for all $i\neq j$ and $\abs{\{i\in [m] \mid \abs{x_i \cap \{u,v\}} = 1\}} = d$.
    \item $\trivertex^{u}_m \colon \binom{[n]}{2}^m \to \{0,1\}$ denotes the triangle vertex function which is defined by $\trivertex^{u}_m(x) = 1$ if and only if there exist distinct $i,j,k\in [m]$ such that $x_i$, $x_j$ and $x_k$ form a triangle containing vertex $u$. 
    \item $\trivertex^{u}_{m,d} \colon D_d\subseteq \binom{[n]}{2}^m \to \{0,1\}$ is the restriction of $\trivertex^{u}_{m,d}$ such that $x\in D_d$ if and only if $\abs{\{i\in [m] \mid \abs{x_i \cap \{u\}} = 1\}} \leq d$.  
    \item $\kcycle{k}_m \colon \binom{[n]}{2}^m \to \{0,1\}$ denotes the $k$-cycle function which is defined by $\kcycle{k}_m(x) = 1$ if and only if there exist $k$ distinct indices $1\leq i_1<i_2<\dots<i_k \leq m$ such that $x_{i_1}, \dots, x_{i_k}$ form a $k$-cycle. We also write $\tri_m$ for $\kcycle{3}_m$.
    \item $\kcycle{k}_{m,d} \colon D_d \subseteq \binom{[n]}{2}^m \to \{0,1\}$ denotes the restriction of $\kcycle{k}_m$ such that $x\in D_{d}$ if and only if for all $u\in [n]$, $\abs{\{i\in [m] \mid \abs{x_i \cap \{u\}} = 1\}} \leq d$. (That is, the maximum degree of the graph represented by $x$ is at most $d$.)  We also write $\tri_{m,d}$ for $\kcycle{3}_{m,d}$.
\end{enumerate}
For $\triedge_m^{\{u,v\}}$, $\triedge_{m,d}^{\{u,v\}}$, $\trivertex_{m}^{u}$, and  $\trivertex_{m,d}^{u}$, we refer to their superscripts as their \emph{target edge} and \emph{target vertex} respectively. We will often omit these superscripts since the (quantum) query complexity of these problems does not depend on them. 

We note that the definitions of $\triedge_m$ and $\triedge_{m,d}$ appear more restrictive than that of the other triangle problems: the input cannot have duplicate edges and the number of input edges incident (but not equal) to the target edge is \emph{exactly} $d$. These restrictions allow us to more directly relate them to the hiding transform of element distinctness.

\paragraph{Search-to-decision reduction.}
Let $k\in \mathbb{N}_{\geq 3}$ be a constant. We construct a search algorithm for finding a $k$-cycle in a length-$m$ edge list, or deciding that none exists, that is correct with high probability using $O(\log^2(m))$ calls to a bounded error decision algorithm for $\kcycle{k}_m$. The details are as follows.

Using standard probability amplification techniques, we reduce the error of the bounded-error $\kcycle{k}_m$ algorithm to $O(1/m)$ by repeating it $O(\log (m))$ times and taking the majority vote. 

Given that each positive instance of $\kcycle{k}_m$ contains a size-$k$ certificate, we can divide the input into $k+1$ parts. Among these, there exist $k$ parts that together contain the certificate. Identifying these $k$ parts requires $\binom{k}{k-1}=O(1)$ calls to the $\kcycle{k}_m$ algorithm, and in each call, we pad the input with a disjoint cycle-free graph to maintain the input size. By doing this, we reduce the search range by a factor of $k/(k+1)$. Repeating this process $O(\log(m))$ times shrinks the search range to $O(1)$, at which point the certificate can be directly identified. By applying the union bound over all calls, the probability of any error occurring remains at most $O(\log(m)/m)$, thus ensuring correctness with high probability.

In the other direction, a search algorithm can be directly used as a decision algorithm. Therefore, we typically do not distinguish between the search and decision versions of the $k$-cycle problem.

\paragraph{A convention.} Technically, all of the above functions should also be parametrized by their input alphabet size, especially since their (quantum) query 
complexity \emph{can} strictly increase as a function of this size.\footnote{As a simple example, $\ED_m \colon [s]^m \to \{0,1\}$ when $s < m$ is the constant $1$ function by the pigeonhole principle so $Q(\ED_m) = 0$; but when $s\geq m$, $Q(\ED_m) = \Theta(m^{2/3})$.} In this work, we often omit this parametrization under the \emph{convention} that the input alphabet size is sufficiently large. For example, consider $\ED_m\colon [s] \to \{0,1\}$ as defined above. We write $Q(\ED_m) = \Omega(m^{2/3})$ to mean: for all sufficiently large $s\in \mathbb{N}$, $Q(\ED_m) = \Omega(m^{2/3})$. We write $Q(\ED_m) = O(m^{2/3})$ to mean: for all $s\in \mathbb{N}$, $Q(\ED_m) = O(m^{2/3})$. Note that we do not need ``sufficiently large'' $s$ here since the (quantum) query complexity is always non-decreasing with respect to $s$. We write $Q(\ED_m) = \Theta(m^{2/3})$ to mean: $Q(\ED_m) = O(m^{2/3})$ and $Q(\ED_m) = \Omega(m^{2/3})$ in the aforementioned senses.

\section{\texorpdfstring{$\triedge$}{TriangleEdge} and Hiding}

As discussed in the introduction, $\triedge$ can be seen as an instantiation of the general hiding phenomenon. We recall the formal definition of hiding from the introduction, \cref{def:hiding}.

Let $a,b$ be integers with $b\geq a\geq 1$. Let $\tildeD,\Sigma$ be finite non-empty sets with $\tildeD\subseteq \Sigma^a$. For $f\colon \tildeD  \to \{0,1\}$, the hiding transform of $f$ is defined to be 
    \begin{equation}
    \hide_b[f] \colon D \subseteq (\Sigma \cup \{*\})^b \to \{0,1\},
    \end{equation}
where $*$ is a symbol outside of $\Sigma$, the set $D$ consists of all strings $y \in (\Sigma \cup \{*\})^b$ with exactly $a$ non-$*$ symbols such that the length-$a$ subsequence\footnote{In general, for integers $b\geq a \geq 1$, $\Sigma$ a finite non-empty set, and a string $x\in \Sigma^b$: a length-$a$ subsequence of $x$ is a string of the form $x_{i_1}x_{i_2}\dots x_{i_a}$ for some integers $1\leq i_1<i_2<\dots<i_a \leq b$.} $\tilde{y}$ formed by those non-$*$ symbols is an element of $\widetilde{D}$, and $\hide_b[f](y)$ is defined to be $f(\tilde{y})$.

By hiding the input in a more structured way, one can get a composition-like problem, for which it is easier to show lower bounds.
\begin{definition}[pSearch \cite{brassard_psearch}]
    Let $m\in \mathbb{N}$ and $\Sigma$ an alphabet. The $\pSearch_m$ function is defined by $\pSearch_m \colon D\subseteq (\{*\}\dot{\cup} \Sigma)^m \to \Sigma$, where $D$ contains all length-$m$ strings in $(\{*\}\dot{\cup} \Sigma)^m$ with exactly one non-$*$ symbol and $\pSearch(x)$ outputs the non-$*$ symbol in $x \in D$.
\end{definition}

\begin{theorem}[\textup{\cite[Theorem 9]{brassard_psearch}}]\label{thm:psearch}
    Let $n,m\in\mathbb{N}$ and $f\colon \Sigma^n \to \{0,1\}$. Then 
    \begin{equation}
        Q(f\circ \pSearch_m) = \Omega(\sqrt{m}\cdot Q(f)). 
    \end{equation}
\end{theorem}

We begin with a general result on hiding.

\begin{proposition}\label{prop:hiding_generic}
    Let $m,n\in \mathbb{N}$ with $m\geq n$ and $f\colon \Sigma^n \to \{0,1\}$. Then $Q(\hide_m[f]) \geq \Omega(\sqrt{m/n} \cdot Q(f))$. Moreover, if $\Sigma = \{0,1\}$ and $f$ is symmetric, then $Q(\hide_m[f]) \leq \tildeO(\sqrt{m/n} \cdot Q(f))$. 
\end{proposition}

\begin{proof}
    For the lower bound, we observe that $f\circ \pSearch_{m/n}$ is a subfunction of $\hide_m[f]$. Therefore, by \cref{thm:psearch}, $Q(\hide_m[f]) \geq \Omega(\sqrt{m/n}\cdot Q(f))$.

    Consider the ``moreover'' part. From \cite{polynomial_2001}, we have $Q(f) = \Theta(\sqrt{n(n- \Gamma(f))})$, where $\Gamma(f) \coloneqq \min\{\abs{2k-n+1}\mid f_k \neq f_{k+1}\}$ and $f_k \coloneqq f(x)$ for all $x$ such that $\abs{x} = k$ as $f$ is symmetric. Observe that $f_k$ must be constant for all $k\in \{\ceil{(n-\Gamma(f))/2},\dots, \ceil{(n+\Gamma(f)-2)/2}\}$. Call this constant $b\in \{0,1\}$. We also write $A \coloneqq \ceil{(n-\Gamma(f))/2}$ and $B\coloneqq \ceil{(n+\Gamma(f)-2)/2}$ for convenience.
    
    We consider the following quantum algorithm for $\hide_m[f]$.

    \begin{figure}[ht]
        \centering
        \fbox{\parbox{0.925\textwidth}{
        \underline{\textbf{Algorithm for $\hide_m(f)$}}\\[1mm]
        On input $y\in \{0,1,*\}^m$:
        \begin{enumerate}
            \item Use Grover search to collect up to $A$ distinct positions $i\in [m]$ such that $y_i = 1$, treating each $*$ as $0$. If the search fails at step $k\in \{1,\dots, A\}$, stop and output $f_{k-1}$. Otherwise, continue.
            \item Use Grover search to collect up to $(n-B)$ distinct positions $i\in [m]$ such that $y_i = 0$, treating each $*$ as $1$. If the search fails at step $k\in \{1,\dots, (n-B)\}$, stop and output $f_{n-k+1}$. Otherwise, output $b$.
        \end{enumerate}
        }\hspace{0.01\textwidth}
        }
    \end{figure}
    
    The quantum query complexity of this algorithm, accounting for error suppression, is 
    \begin{equation}
        \tildeO(\sqrt{mA} + \sqrt{m(n-B})) = \tildeO(\sqrt{m/n}\cdot \sqrt{nA}) = \tildeO(\sqrt{m/n}\cdot Q(f)),
    \end{equation}
    using the basic fact that collecting (up to) $\alpha$ marked items from a list of $\beta$ items costs $\tildeO(\sqrt{\alpha \beta})$ quantum queries by Grover search \cite{vanApeldoorn2024basicquantum}. 
    
    Correctness can be seen as follows. If $\abs{y}<A$, the first part of the algorithm determines $\abs{y}$ exactly and outputs accordingly. Otherwise, $\abs{y}\geq A$ and the algorithm continues. Then, if $\abs{y}>B$, so that $y$ contains fewer than $(n-B)$ $0$s, the second part of the algorithm determines $\abs{y}$ exactly and outputs accordingly. Otherwise, $A\leq \abs{y}\leq B$, so $\hide_m[f](y) = b$.
\end{proof}

\begin{remark}
    The ``moreover'' part of \cref{prop:hiding_generic} is false if $f$ is not symmetric. For example, if $f(x) = x_1$, then it is not hard to see that $Q(g) = \Omega(\sqrt{m-n})$, which is not $O(\sqrt{m/n})$.
\end{remark}

Now, $\triedge$ can be seen as the hiding transform of the element distinctness function. Intuitively, having no triangle that contains the target edge $\{u,v\}$ is equivalent to the neighbors of $u$ and the neighbors of $v$ being distinct.

\begin{theorem}\label{prop:hided_ED}
For $m,d\in \mathbb{N}$ with $m\geq d$, $Q(\hide_m[\ED_d]) = \tildeTheta(\sqrt{m}\cdot d^{1/6})$. 
\end{theorem}

\begin{proof}
    Since $Q(\ED_d) = \Omega(d^{2/3})$, the lower bound $Q(\hide_m[\ED]) = \Omega(\sqrt{m}\cdot d^{1/6})$ follows directly from the lower bound in \cref{prop:hiding_generic}. The upper bound, however, does not follow from the upper bound in \cref{prop:hiding_generic} since that is only stated for $\Sigma = \{0,1\}$. We can nonetheless prove the upper bound by a different argument as follows. 

    Assume $d$ is not a constant as otherwise, we can decide the problem by using Grover's algorithm to find all non-$\ast$ symbols in $O(\sqrt{m})$ queries. Also, we can assume $r\coloneqq m/d$ is an integer, or we can work with the least integer $m'>m$ such that $m'/d$ is an integer since $m'\leq m+d \leq 2m$. We proceed to describe a quantum algorithm for $\hide_m[\ED_d]$.

    Given an input $\tilde{x} \in (\Sigma\cup\{*\})^m$ to $\hide_m[\ED_d]$, we first apply a uniformly random permutation to $\tilde{x}$. Let $x$ denote the resulting string. For $b\in [d]$, let $x^{(b)} \coloneqq x[(b-1)r+1\twodots br] \in (\Sigma\cup\{*\})^r$. A standard probability argument, such as that used to bound the maximum load in a balls-into-bins experiment \cite[Theorem 1]{raab1998balls}, yields
    \begin{equation}\label{eq:block_form_whp}
        \Pr[\text{$\forall b \in [d]$, $x^{(b)}$ contains at most $\log(d)$ non-$*$ symbols}]\geq 1 - O(1/d),
    \end{equation}
    where the probability is over the random permutation. We continue assuming the event in \cref{eq:block_form_whp} occurs.

    For each given $b\in [d]$, we can Grover search for all of the non-$*$ symbols in $x^{(b)}$ using $O(\sqrt{r\log(d)})$ queries \cite{vanApeldoorn2024basicquantum} and output a string of length $\ceil{\log(d)}$, such that its first symbols are the non-$*$ symbols in $x^{(b)}$ and the rest are $*$s. Therefore, we can instantiate one query to a string $y$ of length $d \cdot \ceil{\log(d)}$ containing $*$s and precisely the same non-$*$ symbols as $x$ using $O(\sqrt{r\log(d)})$ queries to $x$. 
    
    Now, we run the optimal quantum algorithm for element distinctness on $y$ while mapping $y_i$ to some symbol $\bot_i$ outside of $\Sigma\cup \{*\}$ if $y_i=*$. (The $\bot_i$s are distinct for distinct $i$s.) This allows us to compute $\hide_m[\ED_d]$ with $O(\sqrt{r\log(d)}\cdot (d\log(d))^{2/3}) = \tildeO(\sqrt{m}\cdot d^{1/6})$ queries to the original input $x$, as required.
\end{proof}

\begin{remark}\leavevmode\label{remark:amplifying_walk}
\begin{enumerate}
    \item Another way of showing the upper bound in \cref{prop:hided_ED} is by a quantum walk that is amplified to the non-$*$ part of the input $x$ using $\tildeO(\sqrt{m/d})$ overhead, cf. proof of \cref{prop:triangle_vertex_upper_bound}.
    \item A benefit of the upper bound approach described in \cref{prop:hided_ED} is that it does not rely on the implementation details of the optimal quantum algorithm for element distinctness but rather only the structure of the problem. More specifically, this structure is used in the line ``mapping $y_i$ to some symbol $\bot_i$ outside of $\Sigma\cup \{*\}$ if $y_i=*$''.  In particular, the same argument would also show that, for every constant integer $k\geq 2$, $Q(\hide_m[\kdist{k}_d]) \leq \tildeO(\sqrt{m/d} \cdot Q(\kdist{k}_d))$ because $\kdist{k}_d$ has that same structure. Informally, using the same argument as in \cref{prop:triangle_edge}, this implies that the problem of deciding whether a length-$m$ edge list contains a $(k+1)$-clique that completes a target $k$-clique with $d$ neighbor has quantum query complexity $\tildeO(\sqrt{m/d} \cdot Q(\kdist{k}_d))$.
\end{enumerate}
\end{remark}

We now formally show that $\triedge_{m,d}$ has the same quantum query complexity as $\hide_m(\ED_d)$ by reducing these problems to each other. Our reduction also works in non-quantum query models.
\begin{corollary}\label{prop:triangle_edge}
For $m,d\in \mathbb{N}$ with $m\geq d$, $Q(\triedge_{m,d}) = \tildeTheta(\sqrt{m}d^{1/6})$. 
\end{corollary}

\begin{proof}
Let $x \in \binom{[n]}{2}^m$ be in the domain of $\triedge^{\{1,2\}}_{m,d}$. Let
$\ED_d \colon [n]^d \to \{0,1\}$. We can compute $\triedge^{\{1,2\}}_{m,d}(x)$ by computing $\hide_m[\ED_d](\tilde{x})$, where for each $i\in [m]$, $\tilde{x}_i$ is defined from $x_i \coloneqq \{u,v\}$ by
\begin{equation}
    \tilde{x}_i \coloneqq \begin{cases}
        * &\text{if  $\{u,v\} \cap \{1,2\} = \emptyset$ or $\{u,v\} = \{1,2\}$},
        \\
        v &\text{if  $u = 1$ and $v \neq 2$},
        \\
        u &\text{if  $u \neq 1$ and $v = 2$},
    \end{cases}
\end{equation}
Observe that the mapping from $x_i$ to $\tilde{x}_i$ can be performed on the fly. Therefore $Q(\triedge_{m,d}) \leq O(Q(\hide_m[\ED_d])) \leq \tildeO(\sqrt{m} d^{1/6})$ by \cref{prop:hided_ED}.

Conversely, let $\tilde{x}\in ([s]\cup \{*\})^m$ be in the domain of $\hide_m[\ED_d']$, where $\ED_d'$ is $\ED_d$ but restricted to inputs containing at most a single collision pair (i.e., two symbols that are the same). Assume that $s\geq m$.\footnote{This assumption is without loss of generality under our convention in the preliminaries of what $Q(\triedge_{m,d})$ means when the alphabet size is not specified.} Let $\triedge^{\{s+1,s+2\}}_{m,d} \colon \binom{[s+2]}{2}^{m} \to \{0,1\}$. Then we can compute $\hide_m[\ED_d'](\tilde{x})$ by computing $\triedge^{\{s+1,s+2\}}_{m,d}(x)$, where for each $i\in [m]$, $x_i$ is defined by
\begin{equation}
    x_i \coloneqq \begin{cases}
        \{s+1,\tilde{x}_i\} \text{ or } \{s+2,\tilde{x}_i\} &\text{if  $\tilde{x}_i \in [s]$},
        \\
        \{1,i\} &\text{if  $\tilde{x}_i= *$},
    \end{cases}
\end{equation}
where the choice of ``or'' in the first case is made uniformly at random. (Note that we consider $\ED_d'$ instead of $\ED_d$ and assume $s\geq m$ so that the $x$ defined above is in the domain of $\triedge^{\{s+1,s+2\}}_{m,d}$ for some choices of ``or''.)

If $\tilde{x}$ contains a collision pair, then with probability at least $1/2$, $x$ will be a yes-instance of $\triedge^{\{s+1,s+2\}}_{m,d}$. On the other hand, if $\tilde{x}$ does not contain any collision pairs, $x$ will be a no-instance of  $\triedge^{\{s+1,s+2\}}_{m,d}$. Therefore, by repeating the above reduction a constant number of times, we can distinguish between these cases with high probability. Therefore $Q(\triedge^{\{s+1,s+2\}}_{m,d}) \geq \Omega(Q(\hide_m[\ED_d'])) =\Omega(\sqrt{m/d} \cdot Q(\ED_d')) = \Omega(\sqrt{m}d^{1/6})$ by \cref{prop:hiding_generic} since $Q(\ED_d') = \Omega(d^{2/3})$.\footnote{The  $Q(\hide_m[f]) = \Omega(\sqrt{m/n}\cdot Q(f))$ part of \cref{prop:hiding_generic} also holds for partial $f$ (like $\ED_d'$) by the same argument.}
\end{proof}

\begin{remark}
The quantum query complexity of the analogue of the $\triedge$ for $\kdist{3}_m$ is $\Theta(\sqrt{m})$ and for $\ksum{3}_m$ is $\Theta(m^{2/3})$. By analogue, we mean detecting whether the input contains a full certificate given one position of the certificate.  So we see that the quantum query complexity of $\triedge$ lies between that of these analogues.
\end{remark}

The results of this section naturally raise the following conjecture.
\begin{conjecture}\label{conj:symmetric_hiding}
    Let $m,n\in \mathbb{N}$ with $m\geq n$. For every symmetric $f\colon \Sigma^n \to \{0,1\}$, $Q(\hide_m[f]) \leq O(\sqrt{m/n} \cdot Q(f))$.
\end{conjecture}

We proved this conjecture (up to logarithmic factors) when the domain of $f$ is Boolean using the structure of the optimal quantum query algorithm for such $f$. On the other hand, we proved this conjecture when $f=\ED_d$ using the structure of the function itself. We believe that attacking this conjecture for arbitrary $f$ would yield further insights into the structure of symmetric functions and their optimal quantum query algorithms.

Towards the above conjecture, we show the following proposition.
\begin{proposition}\label{prop:symmetric_partial}
     Let $m,n\in \mathbb{N}$ with $m\geq n$. For every symmetric $f\colon \Sigma^n \to \{0,1\}$, $R(\hide_m[f]) \leq O(\frac{m}{n}\cdot R(f))$. Moreover, $Q(\hide_m[f]) \leq O(\sqrt{m/n}\cdot Q(f)^3 \cdot \log^{3/2}(n))$.
\end{proposition}
\begin{proof}
    Write $k \coloneqq R_{1/100}(f) \geq 1$. Since $f$ is symmetric, \cite[Lemma 4.8 of full version]{baryoussef_2001} applies and shows that there exists a \emph{variable oblivious, uniform} $k$-query randomized algorithm for computing $f$ with bounded error $1/100$ of the following form. 
    
    {\vspace{3mm}\centering
    \noindent\fbox{\parbox{0.925\textwidth}{
        Given input $x\in \Sigma^n$:
        \vspace{1mm}
        \begin{enumerate}
            \item Choose a uniformly random subset $S\subseteq [n]$ of size $k$.
            \item Query $x_i$ for all $i\in S$.
            \item Output $1$ with probability $p$, where $p = p(X)$ is a function of the size-$k$ multiset $X\coloneqq \{x_i\mid i\in S\}$.
       \end{enumerate}}
        }
    \vspace{1mm}
    }
    
    Let $K \coloneqq \frac{m}{n}\cdot 100k$ and consider the following $K$-query randomized algorithm for computing $\hide_m[f]$:
    
    {\vspace{3mm}\centering
    \noindent\fbox{\parbox{0.925\textwidth}{
    Given input $y\in (\{*\}\cup\Sigma)^m$ in the domain of $\hide_m[f]$:
    \vspace{1mm}
        \begin{enumerate}
            \item Choose a uniformly random subset $T\subseteq [m]$ of size $K$. Let $Y$ denote the multiset $\{y_i \mid i \in T\}$.
            \item If the number of non-$*$ symbols in $Y$ is fewer than $k$, output $0$ or $1$ uniformly at random.
            \item If the number of non-$*$ symbols in $Y$ is at least $k$, choose a uniformly random size-$k$ sub-multiset $X$ of the non-$*$ symbols in $Y$ and output $1$ with probability $p(X)$, where $p(X)$ is as defined above.
        \end{enumerate}}
        }
    \vspace{1mm}
    }
    
    This algorithm behaves exactly the same as the previous algorithm provided we reach the third step. The expected number of non-$*$ symbols in $Y$ is $100k$.
    By the Paley-Zygmund inequality, the probability of $Y$ containing at least $k$ non-$*$ symbols is at least $(1-1/100)^2 K^2/(K+K^2) \geq 0.97$ since $K\geq 100$. Therefore, we reach the third step with probability at least $0.97$. Therefore the algorithm is correct with probability at least $0.97(1-1/100)\geq 2/3$. Hence $R(\hide_m[f]) \leq O(\frac{m}{n}\cdot R(f))$.

    For the ``moreover'' part, by the same permutation argument as in the proof of the upper bound in \cref{prop:hided_ED}, it follows that 
    $Q(\hide_m[f]) \leq O(\sqrt{(m/n)\cdot \log(n)} \cdot Q(\hide_{\ceil{n\log(n)}}[f]))$. Therefore, it suffices to show that $Q(\hide_{\ceil{n\log(n)}}[f]) \leq O(Q(f)^3 \cdot \log(n))$. But this is true since 
    \begin{equation}
         Q(\hide_{\ceil{n\log(n)}}[f]) \leq R(\hide_{\ceil{n\log(n)}}[f]) \leq O(\log(n)\cdot R(f)) \leq O(\log(n) \cdot Q(f)^3),
    \end{equation}
    where the first inequality uses $Q(\cdot) \leq R(\cdot)$, the second inequality uses the first part of this proposition, and the third inequality uses the result of \cite{chailloux_symmetric} for symmetric functions $f$.
\end{proof}

\section{\texorpdfstring{$\trivertex$}{TriangleVertex} and Shuffling}\label{sec:shuffle}
In this section, we consider the $\trivertex$ problem of deciding whether an input edge list of length $m$ has a triangle containing a target vertex of degree $d \leq m$.

We take this opportunity to study the shuffling transformation as well. As mentioned in the introduction, the edge list itself can be viewed as a shuffled adjacency list. In $\trivertex$, there is an additional ``block-wise'' notion of shuffling, where the blocks consist of  edges incident and not incident to the target vertex, respectively. This latter notion motivates our definition of shuffled direct sum.

\subsection{\texorpdfstring{$\trivertex$}{TriangleVertex}}
We first build intuition by directly working out some upper and lower bounds on the quantum query complexity of $\trivertex$. The upper bound uses a quantum walk algorithm, and the lower bound uses a reduction from the hiding transform of $\kdist{3}$.
\begin{proposition}\label{prop:triangle_vertex_upper_bound}
For $m,d\in \mathbb{N}$ with $m\geq d$, $Q(\trivertex_{m,d}) = \tildeO(\sqrt{m}d^{1/4})$.
\end{proposition}

\begin{proof}
     Suppose the given vertex is vertex $u$. Let $x\in \binom{[n]}{2}^m$ be the input to the $\trivertex_{m,d}$ problem. Then $\abs{\{i\in [m] \mid \abs{x_i \cap \{u\}} = 1\}} \leq d$ and we may wlog assume the inequality is saturated. Observe that $x$ can be viewed as having two parts: one part $A\subseteq[m]$ such that $i\in A$ if and only if $x_i$ is an edge incident to vertex $u$, and the other part $B \coloneqq [m]-A$ containing the remaining edges. Note $\abs{A} = d$.
    
    We perform a quantum walk based on the uniform random walk on the Hamming graph $G$ labeled by $r$-tuples of indices from $A$, where $r\in \mathbb{N}$ will be optimized later. We say a vertex $(i_1,\dots,i_r) \in A^r$ of $G$ is marked if there exists a $b\in B$ such that $x_b$ forms a triangle with $x_{i_j}$ and $x_{i_k}$ for some $j,k\in [r]$. At each vertex, we also store the data $(x_{i_1},\dots,x_{i_r})$.
    
    We now analyze the query complexity of this quantum walk following \cite{mnrs_2011}.
    
    {\vspace{2mm}\centering
    \noindent\fbox{\parbox{0.925\textwidth}{
        \begin{enumerate}
            \item The setup cost $S$ of the quantum walk is $S = \tildeO(r\sqrt{m/d} + \sqrt{m})$ by amplitude amplification and then Grover search, where the $\tildeO$ accounts for error reduction. In more detail, amplitude amplification gives the state
            \begin{equation}
                \Bigl(\frac{1}{\sqrt{\abs{A}}} \sum_{i\in A} \ket{i} \ket{x_i} \Bigr)^{\otimes r}  = \frac{1}{\sqrt{\abs{A}^r}} \sum_{i_1,\dots, i_r \in A} \ket{i_1,\dots, i_r} \ket{x_{i_1},\dots, x_{i_r}},
            \end{equation}
            and costs $\tildeO(r\sqrt{m/d})$ queries since $\abs{A} = d$. Then, Grover search is used to compute into a separate register whether the $(i_1,\dots,i_r)$ is marked. Assume wlog that measuring that register gives ``unmarked'', then the state becomes
            \begin{equation}
                \sum_{i_1,\dots, i_r \in A \, \colon \,  (i_1,\dots i_r) \text{ unmarked}} \ket{i_1,\dots, i_r} \ket{x_{i_1},\dots, x_{i_r}}.
            \end{equation}
            \item  The update cost of the quantum walk is $U = \tildeO(\sqrt{m/d})$ by amplitude amplification, where the $\tildeO$ accounts for error reduction, which crucially uses the uniformity of the underlying random walk.
            \item The checking cost of the quantum walk is $C = O(\sqrt{m})$ by Grover search.
            \item The fraction of marked vertices is $\epsilon = \Omega((r/d)^2)$ if $x$ is a yes-instance, else $\epsilon =0$.
            \item The spectral gap of the transition matrix of the random walk is $\delta = \Omega(1/r)$.
        \end{enumerate}}
    }
    \vspace{1mm}
    }

    Therefore, \cite{mnrs_2011} shows that the query complexity of the quantum walk is
    \begin{equation}
        O\Bigl(S + \frac{1}{\sqrt{\epsilon}}\bigl(\frac{1}{\sqrt{\delta}}\cdot U + C\bigr)\Bigr) = \tildeO\Bigl(r \sqrt{\frac{m}{d}} + \sqrt{m} + \frac{1}{\sqrt{r}}\sqrt{md} +\frac{d}{r}\sqrt{m}\Bigr) = \tildeO(\sqrt{m}\cdot d^{1/4}),
    \end{equation}
    by optimally setting $r$ to be $\ceil{d^{3/4}}$. 
\end{proof}

\begin{proposition}\label{prop:triangle_vertex_lower_bound}
   For $m,d\in \mathbb{N}$ with $m\geq d$, $Q(\trivertex_{m,d})\geq \Omega(\sqrt{m/d}\cdot Q(\kdist{3}_d))$.
\end{proposition}

\begin{proof}
    We describe a reduction from $\hide_m[\kdist{3}_d]$ to $\trivertex_{m,d}$. Assume wlog that $m$ is a multiple of $3$. Let $x \in ([s]\cup \{*\})^m$ be an input to $\hide_m[\kdist{3}_d]$. Choose a uniformly random partition of $[m]$ into three subsets $X,Y,Z$ of the same size $m/3$. Let $x^X,x^Y,x^Z$ denote the subsequence of $x$ formed by indices in $X,Y,Z$ respectively. If $x$ has a $3$-collision, the probability of the event that exactly one of the three colliding indices is in each of $X$, $Y$, and $Z$ is at least $2/9$.

    We can compute $\hide_m[\kdist{3}_d](x)$ by computing $\trivertex_{m,d}$ on an edge list $y$ with vertex set 
    \begin{equation}
        \{v\} \ \dot{\cup} \ \{u\} \ \dot{\cup} \ \{a^X\mid a \in [s]\} \ \dot{\cup} \ \{a^Y\mid a\in [s]\} \ \dot{\cup} \ \{*_i \mid i \in [m]\}
    \end{equation}
    and target vertex $v$. For each $i\in [m]$, we map $x_i$ to an edge $y_i$ as follows.
    \begin{enumerate}
        \item Case $x_i \in [s]$: Map $x_i$ to edge $\{x_i^X,v\}$ if $i\in X$, edge $\{x_i^Y,v\}$ if $i\in Y$, and edge $\{x_i^X,x_i^Y\}$ if $i\in Z$.
        \item Case $x_i = *$: Map $x_i$ to $\{u,*_i\}$.
    \end{enumerate}
    Note that $y$ defined this way is of length $m$, and the number of edges in $y$ incident to $v$ is exactly $d$. Therefore, $y$ is a valid input to $\trivertex_{m,d}$.
     
    If $x$ has a $3$-collision, then $y$ has a triangle containing vertex $v$ with probability at least $2/9$. If $x$ does not have a $3$ collision, then $y$ never has a triangle containing vertex $v$. Therefore, repeating the reduction a constant number of times allows us to compute $\hide_m[\kdist{3}_d](x)$ with high probability. Therefore, $Q(\trivertex_{m,d})$ is lower bounded by $Q(\hide_m[\kdist{3}_d])$, which is $\Omega(\sqrt{m/d}\cdot Q(\kdist{3}_d))$ by \cref{prop:hiding_generic}, as required.
\end{proof}

Observe that in the proof of \cref{prop:triangle_vertex_upper_bound} above, the algorithm does not a priori know which part of the input $x$ is in $A$ and which part is in $B$. It appears that the shuffling of these parts adds to the complexity of the problem. To investigate how generic this phenomenon is, we consider two types of questions motivated by shuffling.
\begin{enumerate}
    \item How does shuffling a function change its complexity?
    \item How does shuffling affect direct sum theorems?
\end{enumerate}
We formalize these questions in the next two sections and prove some first results.

\subsection{Shuffled functions}\label{sec:shuffled_functions}

We recall \cref{def:shuffling} of shuffling from the introduction. Let $n\in \mathbb{N}$. Let $\tildeD,\Sigma$ be finite non-empty sets with $\tildeD\subseteq \Sigma^n$. For a function $f\colon \widetilde{D} \to \{0,1\}$, the shuffling transform of $f$ is defined to be
    \begin{equation}
        \shuffle[f] \colon D \subseteq (\Sigma\times [n])^n  \to \{0,1\}, \quad \text{where}
    \end{equation}
    \begin{enumerate}
        \item the set $D$ consists of all elements $x\in (\Sigma\times [n])^n$ satisfying $x = ((v_1,\pi(1)),\dots,(v_n,\pi(n)))$ for some bijection $\pi\colon [n]\to [n]$ such that $(v_{\pi^{-1}(1)},v_{\pi^{-1}(2)},\dots,v_{\pi^{-1}(n)})\in \widetilde{D}$.
    \item $\shuffle[f](x)\coloneqq f(v_{\pi^{-1}(1)},v_{\pi^{-1}(2)},\dots,v_{\pi^{-1}(n)})$.
    \end{enumerate}

Next, we collect some simple facts about $\shuffle[f]$. 
\begin{fact}\label{fact:shuffle}
For every $f\colon E\subseteq\Sigma^n\to\{0,1\}$, the following holds:
\begin{enumerate}
    \item\label{enu:shuffle-vs-origin} $Q(\shuffle[f])=\Omega(Q(f)).$
    \item\label{enu:shuffle-sym} If $f$ is symmetric, then $Q(\shuffle[f])=Q(f).$
    \item\label{enu:shuffle-sep} $Q(\shuffle[f])=O(\sqrt{n}\cdot Q(f)).$ Moreover, there exists $f$ such that $Q(\shuffle[f])=\Omega(\sqrt{n}\cdot Q(f))$.
\end{enumerate}
\end{fact}
\begin{proof}
    \cref{enu:shuffle-vs-origin}: Every query complexity measure on $\shuffle[f]$ is at least that on the $f$ since $\shuffle[f]$ contains $f$ as a subfunction. \cref{enu:shuffle-sym}: If $f$ is symmetric, then every query complexity measure on $f$ is the same as that measure on $\shuffle[f]$ since the algorithm computing $\shuffle[f]$ could simply ignore the second coordinate of each input symbol. \cref{enu:shuffle-sep}: Let $\calA$ be an arbitrary algorithm for $f$. For each $z\in \Sigma^n$, we can simulate the quantum query oracle in the algorithm $\calA$ by exact Grover search through $x\in (\Sigma\times [n])^n$ over the second coordinate for the target index, and exact Grover search incurs a multiplicative overhead of $O(\sqrt{n})$. The ``moreover'' part follows by simply taking $f(x) = x_1$ to be the dictator function. This $f$ witnesses a $1$ vs $\Omega(\sqrt{n})$ separation between $Q(f)$ and $Q(\shuffle[f])$. 
\end{proof}
We record one more fact that relates the quantum and randomized query complexities for $\shuffle[f]$.
\begin{fact}[cf. \cite{chailloux_symmetric}]\label{fact:shuffle-classical-quantum}
For every $f\colon E\subseteq\Sigma^n\to\{0,1\}$,
    \begin{equation*}
        Q(\shuffle[f])) = \Omega(R(\shuffle[f])^{1/3}).
    \end{equation*}  
\end{fact}
The importance of the above fact is that it provides a generic lower bound method for the quantum query complexity for problems in the edge list model, since analyzing the randomized query complexity is often easier. The proof of this fact is based on the simple observation that $\shuffle[f]$ transforms every function $f$ into a symmetric function. The beautiful work of Chailloux \cite{chailloux_symmetric} shows that the quantum query complexity of every symmetric function is at least big-$\Omega$ of its randomized query complexity raised to the power $1/3$.

As we have observed, for general $f$, $Q(\shuffle[f])$ can be unboundedly larger than $Q(f)$ while for symmetric $f$, $Q(\shuffle[f]) = Q(f)$. Therefore, the interesting question becomes what happens when $f$ is partially symmetric, for example, by being defined via graph properties? This question can be used to understand how the complexity of a graph property changes in the edge list model versus the adjacency list and adjacency matrix models. 

\begin{quote}
\emph{Can there be a large separation between the quantum query complexity of a graph property in the edge list model versus the adjacency list model?}
\end{quote}

We show there can be an exponential separation.
\begin{theorem}\label{thm:shuffling-adj-list}
    There is (partial) graph property $\calP$ on graphs of maximum degree $5$, such that its quantum query complexity in the adjacency list model is $O(\polylog(n))$, but its quantum query complexity in the edge list model is $\Omega(n^{1/48}).$
\end{theorem}
\begin{proof}
We know from \cite{symmetries_bcgkpw} that there is a graph property $\calP$ 
in the adjacency list model on graphs of maximum degree $5$ with quantum query complexity $O(\polylog(n))$ and randomized query complexity $\Omega(n^{1/16})$, where the input adjacency list has size $n$. Now, by~\cref{fact:shuffle-classical-quantum}, when we try to compute $\calP$ in the edge list model, the quantum query complexity is at least $\Omega(n^{1/48})$.
\end{proof}

As the edge list model is the shuffled version of the adjacency list model, we find it natural to also consider the shuffled version of the adjacency matrix model and ask the following question.

\begin{quote}\label{quote:adjacency_matrix_shuffled}
\emph{Can there be a large separation between the quantum query complexity of a graph property in the adjacency matrix model versus the shuffled adjacency matrix model?}
\end{quote}

We can no longer appeal to the previous argument because \cite{symmetries_bcgkpw} shows that the quantum query complexity of every graph property in the adjacency matrix model is at least that of its randomized query complexity raised to the $1/6$ power.
Nonetheless, we can answer the question with ``yes'' by considering a function $\SMAJ_n$ defined in~\cref{def:smaj}. $\SMAJ_n$ can be seen as a (directed) graph property computed in the adjacency matrix model.

Recall that $\SMAJ_n\colon D_0\dot{\cup} D_1 \subseteq \{0,1\}^{n^2} \to \{0,1\}$ is defined to be the following restriction of $\MAJ_n\circ \MAJ_n$ on $n^2$ bits.
\begin{equation*}\label{eq:maj_circ_maj}
     \SMAJ_n : x \mapsto \begin{cases}
         0, &x\in D_0,\\
         1, &x\in D_1,
     \end{cases}
\end{equation*}
where
 $D_0$ consists of all $x = (x_{1,1},x_{1,2},\dots, x_{n,n}) \in \{0,1\}^{n^2}$  such that there exists a subset $S\subseteq[n]$ of size at least $2n/3$ such that for all $i\in S$, $x^i \coloneqq (x_{i,1},\dots,x_{i,n})$ has Hamming weight $\abs{x^i} \ge 2n/3$ and for all $i\in [n]-S$,  $\abs{x^i} \le n/3$; while 
 $x = (x_{1,1},x_{1,2},\dots, x_{n,n}) \in \{0,1\}^{n^2}$ is in $D_1$ if and only if there exists a subset $S\subseteq[n]$ of size at most $n/3$ such that for all $i\in S$, $x^i \coloneqq (x_{i,1},\dots,x_{i,n})$ has Hamming weight $\abs{x^i} \ge 2n/3$ and for all $i\in [n]-S$, $\abs{x^i} \le n/3$.

\begin{theorem}\label{prop:shuffling_adj_matrix}
$R(\SMAJ_n) = O(1)$ but $R(\shuffle[\SMAJ_n]) = \Omega(\sqrt{n})$. In addition, $Q(\SMAJ_n) = O(1)$ but $Q(\shuffle[\SMAJ_n]) = \Omega(n^{1/6})$.
\end{theorem}

\begin{proof}
To see that $R(\SMAJ_n) = O(1)$, observe that for a given $x^i$ we can test with high probability whether $\abs{x^i} \geq 2n/3$ (call such an $x^i$ \emph{dense}) or $\abs{x^i} \leq n/3$ (call such an $x^i$ \emph{sparse}) by querying $x^i$ repeatedly at uniformly random locations a constant number of times. Then we choose a subset of $T\subseteq [n]$ of constant size uniformly at random and for each $i\in T$ test whether $x^i$ is dense or sparse. If more than half are dense, we output $0$. If at most half are dense, we output $1$. It is not hard to see that this algorithm computes $\SMAJ_n$ with at most $1/3$ probability of error.\footnote{An alternative way to see this is simply that $\SMAJ_n$ is the composition of two gapped $\MAJ_n$ functions, each of which has constant randomized query complexity. Therefore, the randomized query complexity of the composition is also constant using the fact that $R(f\circ g) = O(R(f)\cdot R(g)\cdot \log(R(f)))$.} Therefore, $Q(\SMAJ_n) \leq R(\SMAJ_n)  = O(1)$. 

For the lower bounds, by~\cref{fact:shuffle-classical-quantum}, it suffices to show $R(\shuffle[\SMAJ_n]) = \Omega(\sqrt{n})$ to conclude the proof.  
Assume wlog that $n$ is a multiple of $12$. Then by Yao's lemma, it suffices to show that the following two distributions $\calD_0$ and $\calD_1$ on inputs to $\shuffle[\SMAJ_n]$ are hard to distinguish by a deterministic algorithm of query complexity $O(\sqrt{n})$, where we define
\begin{description}
    \item $\calD_b$: sample $x\leftarrow H_b$ and $\pi \leftarrow \mathfrak S_{[n]\times [n]}$ (the set of all permutations on $[n]\times [n]$), then output 
    \begin{equation}
        (x_{\pi(1,1)}, \pi(1,1)), (x_{\pi(1,2)}, \pi(1,2)), \dots, (x_{\pi(n,n)}, \pi(n,n)).
    \end{equation}
\end{description}
Here, $H_0\subseteq D_0$ and $H_1\subseteq D_1$ are defined by
\begin{equation}
\begin{aligned}
    H_0 \coloneqq \{(x_{1,1},x_{1,2},&\dots, x_{n,n}) \in \{0,1\}^{n^2} \mid \exists S\subseteq[n] ,  \text{ s.t.}
    \\
    & \text{(i) } |S|=2n/3,\, \text{(ii) }\forall i\in S, \,  |x^i| = 3n/4, \text{ and (iii) }\forall i\not\in S, \, |x^i|=0\},
    \\
    H_1 \coloneqq \{(x_{1,1},x_{1,2},&\dots, x_{n,n}) \in \{0,1\}^{n^2} \mid \exists S\subseteq[n] ,  \text{ s.t.}
        \\
    & \text{(i) } |S|=n/3,\, \text{(ii) }\forall i\in S, \,  |x^i| = n, \text{ and (iii) }\forall i\not\in S, \, |x^i|=n/4\}.
\end{aligned}
\end{equation}

Let $t\in \mathbb{N}$. For $0 \leq \beta < 1/2 < \alpha\leq 1$ and $\gamma \in [0,1]$, define the random process $P(\alpha,\beta,\gamma)$ as in \cref{fig:random_process}. We refer to $P(3/4,0,2/3)$ as $P_0$ and $P(1,1/4,1/3)$ as $P_1$. The $P_i$s are defined in such a way that the distribution on query outcomes induced by every $t$-query deterministic decision tree (that we assume wlog never queries an index it has already queried and is a balanced tree) running on an input chosen sampled according to $\calD_i$ is the same as the distribution $P_i$.

\begin{figure}[ht]
\centering
\fbox{\begin{minipage}[t]{0.9\columnwidth}
\textbf{\underline{Process $P(\alpha,\beta,\gamma)$.}}
\vspace{2mm}

Set $\SEEN =\emptyset$, $\DENSE =\emptyset$, $\SPARSE =\emptyset$, and $\forall i\in [n]$, $M_i =  N_i = 0$.

Repeat the following $t$ times:
\begin{enumerate}
    \item Randomly select $(i,j) \leftarrow [n]\times [n] - \SEEN$.
    \item If $i\in \DENSE$, select ``dense''. If $i\in \SPARSE$, select ``sparse''. Else randomly select ``dense'' or ``sparse'' such that ``dense'' is selected with probability
    \begin{equation}
        p_\dense \coloneqq 
        \frac{\gamma \cdot n - \abs{\DENSE}}{n - \abs{\DENSE} - \abs{\SPARSE}}.
    \end{equation}
    \begin{enumerate}
    \item If ``dense'' is selected, randomly set a bit $b$ to be $1$ or $0$ such that $1$ is selected with probability
    \begin{equation}
        p_1 \coloneqq \frac{(\alpha n - M_i)}{n - N_i}.
    \end{equation}
     If $b=1$, set $M_i$ to $M_i + 1$. Set $\DENSE$ to $\DENSE\cup\{i\}$. 
     \item If ``sparse'' is selected, randomly set a bit $b$ to be $1$ or $0$ such that $1$ is selected with probability
     \begin{equation}
         q_1 \coloneqq \frac{(\beta n - M_i)}{n - N_i}.
     \end{equation}
      If $b=1$, set $M_i$ to $M_i + 1$. Set $\SPARSE$ to $\SPARSE\cup\{i\}$. 
    \end{enumerate}
    \item Set $\SEEN$ to $\SEEN \cup (i,j)$. Set $N_i$ to $N_i+1$.
    \item Output $(i,j,b) \in [n]\times [n]\times \{0,1\}$.\vspace{2mm}
\end{enumerate}
\end{minipage}
\hspace{2mm}}
\caption{The random process $P(\alpha,\beta,\gamma)$.}\label{fig:random_process}
\end{figure}

Our goal is to show that the total variation distance between the output distributions of $P_0$ and $P_1$ scales as $O(t^2/n)$. We write $\TVD$ for this particular total variation distance. In the following, we write $\Pr_i[\cdot]$ for taking the probability of an event over the randomness in the process $P_i$. We also assume $t-1\leq n^2/2$ to help simplify calculations.

Let $\Bad$ denote the subset of $([n]\times[n]\times\{0,1\})^t$ such that a $t$-tuple $(i_k,j_k,b_k)$, $k=1,2,\dots,t$ belongs to $\Bad$ if and only if there exists $k_1\neq k_2$ such that $i_{k_1} = i_{k_2}$. Then it is easy to see that
\begin{equation}
\begin{aligned}
    \Prob_0[\Bad] = \Prob_1[\Bad] =& 1 - 1\cdot \frac{n^2-n}{n^2-1}\cdot \frac{n^2-2n}{n^2-2} \cdots \frac{n^2-(t-1)n}{n^2-(t-1)} 
    \\
    =&1 - 1\cdot \Bigl(1 - \frac{n-1}{n^2-1}\Bigr)\cdot \Bigl(1 - \frac{2n-2}{n^2-2}\Bigr)\cdots \Bigl(1 - \frac{(t-1)(n-1)}{n^2-(t-1)}\Bigr) \leq \frac{t^2}{n},
\end{aligned}
\end{equation}
where the last inequality uses $\forall a,b\geq 0, (1-a)(1-b) \geq 1-a-b$ and $t-1 \leq n^2/2$.

Let $X,Y\in ([n]\times[n]\times \{0,1\})^t$ denote the output of process $P_0$ and $P_1$ respectively. Then
\begin{equation}\label{eq:tvd_two_parts}
    \TVD = \sum_{x}\abs{\Prob_0[X = x] - \Prob_1[Y = x]} \leq \sum_{x \notin \Bad}\abs{\Prob_0[X = x] - \Prob_1[Y = x]} + \frac{2t^2}{n}.
\end{equation}

We proceed to analyze the first term. Write $X^1 \in ([n]\times [n])^t$ for the non-bit-part of $X$. Write $X^2 \in \{0,1\}^t$ for the bit-part of $X$. Similarly write $Y^i,x^i$. Write $\Bad^1 \subseteq ([n]\times [n])^t$ for the non-bit-part of $\Bad$. Then
\begin{align}
    \sum_{x \notin \Bad}\bigl|&\Prob_0[X = x] - \Prob_1[Y = x]\bigr| \nonumber
    \\
    =& \sum_{x \notin \Bad}\bigl|\Prob_0[X^1 = x^1, X^2 = x^2] - \Prob_1[Y^1 = x^1, Y^2 = x^2]\bigr| \nonumber
    \\
    =&\sum_{x \notin \Bad}\bigl|\Prob_0[X^2 = x^2\mid X^1 = x^1]\Pr[X^1 = x^1] - \Prob_1[Y^2 = x^2 \mid Y^1 = x^1 ]\Pr[Y^1 = x^1]\bigr| \nonumber
    \\
    =&\sum_{x \notin \Bad}\bigl|\Prob_0[X^2 = x^2\mid X^1 = x^1] - \Prob_1[Y^2 = x^2 \mid Y^1 = x^1 ]\bigr| \cdot \Pr[X^1 = x^1],
    \nonumber\\
    =&\sum_{x^1 \notin \Bad^1}\biggl(\sum_{x^2}\bigl|\Prob_0[X^2 = x^2\mid X^1 = x^1] - \Prob_1[Y^2 = x^2 \mid Y^1 = x^1 ]\bigr| \biggr)\cdot \Pr[X^1 = x^1],
    \label{eq:tvd_no_collision}
\end{align}
where the second-to-last equality uses $\Pr[X^1 = x^1] = \Pr[Y^1 = x^1]$, which holds because the first steps defining $P_0$ and $P_1$ are identical.

Now, write $U,V \in \{\dense,\sparse\}^t$ for the sequence of ``dense'' or ``sparse'' choices during the process $P_0$ and $P_1$ respectively. Given $u\in \{\dense,\sparse\}^t$, write $\bar{u}$ for $u$ but switching ``dense'' and ``sparse'' component-wise.

For a fixed $x\notin \Bad$,  write $\BarPr_0[\cdot]$ for $\Pr_0[\cdot \mid X^1 = x^1]$ and $\BarPr_1[\cdot]$ for $\Pr_1[\cdot \mid Y^1 = x^1]$. Then,
\begin{align}
    \Prob_0[X^2 &= x^2\mid X^1 = x^1] - \Prob_1[Y^2 = x^2 \mid Y^1 = x^1]\nonumber\\
    =&\sum_u \BarPr_0[X^2 = x^2\mid U = u]\BarPr_0[U=u] - \BarPr_1[Y^2 = x^2 \mid V = u]\BarPr_1[V=u]\nonumber
    \\
    =&\sum_u \Bigl(\BarPr_0[X^2 = x^2\mid U = u] - \BarPr_1[Y^2 = x^2 \mid V = \bar{u}]\Bigr)\BarPr_0[U=u],\label{eq:tvd_to_simplify}
\end{align}
where the last equality uses $\Prob_0[U=u \mid X^1 = x^1] = \Prob_1[V=\bar{u} \mid Y^1 = x^1]$, which holds because when $x\notin \Bad$, ``dense'' and ``sparse'' choices are made randomly in both $P_0$ and $P_1$; the important bar over $u$ on the right-hand side arises from the fact that the $\gamma$ defining $P_1$ is $1$ minus that defining $P_0$.

Now let $W \in \{\dense,\sparse\}^t$ denote a random variable such that each $W_i \in \{\dense,\sparse\}$ is independently distributed and equal to ``dense'' with probability $2/3$. Consider a bag of $n$ balls such that $2n/3$ are labeled ``dense'' and $n/3$ are labeled ``sparse''. Observe that the distribution on $u \in \{\dense,\sparse\}^t$ defined by $\BarPr_0[U=u]$ is the same as the distribution of labels when drawing $t$ balls from the bag \emph{without} replacement. On the other hand, the distribution of $W$ is the same as the distribution of labels when drawing $t$ balls from the bag \emph{with} replacement. Therefore, we can appeal to \cite{sampling_freedman} to deduce that
\begin{equation}\label{eq:tvd_replacement_sampling}
    \sum_u \bigl|\Pr[W=u] - \BarPr_0[U=u]\bigr| \leq \frac{t(t-1)}{2n}.
\end{equation}
Therefore, summing \cref{eq:tvd_to_simplify} over $x^2\in \{0,1\}^t$ gives
\begin{align}
     \sum_{x^2}\Bigl|\Prob_0[&X^2 = x^2\mid X^1 = x^1] - \Prob_1[Y^2 = x^2 \mid Y^1 = x^1]\Bigr|\nonumber
     \\
     =&\sum_{x^2}\Bigl|\sum_u \Bigl(\BarPr_0[X^2 = x^2\mid U = u]\BarPr_0[U=u] - \BarPr_1[Y^2 = x^2 \mid V = \bar{u}]\BarPr_0[U=u]\Bigr)\Bigr|\nonumber
     \\
     \leq&\sum_{x^2}\Bigl|
     \sum_u \Bigl(\BarPr_0[X^2 = x^2\mid U = u]\BarPr_0[U=u] - \BarPr_0[X^2 = x^2\mid U = u]\Prob[W=u]\Bigr)\Bigr|\nonumber
     \\
    &\hspace{8mm} +
    \sum_{x^2}\Bigl|\sum_u\Bigl(\BarPr_0[X^2 = x^2\mid U = u]\Prob[W=u] - \BarPr_1[Y^2 = x^2 \mid V = \bar{u}]\Prob[W=u]\Bigr)\Bigr|\nonumber
     \\
    &\hspace{8mm} +
    \sum_{x^2} \Bigl| \sum_u\Bigl(\BarPr_1[Y^2 = x^2 \mid V = \bar{u}]\Prob[W=u] -
     \BarPr_1[Y^2 = x^2 \mid V = \bar{u}]\BarPr_0[U=u]
     \Bigr) 
     \Bigr|\nonumber
     \\
     \leq&\sum_{x^2}\Bigl|\sum_u \Bigl(\BarPr_0[X^2 = x^2\mid U = u] - \BarPr_1[Y^2 = x^2 \mid V = \bar{u}]\Bigr)\Prob[W=u]\Bigr|+ \frac{2t^2}{n} = \frac{2t^2}{n},\label{eq:iid_conversion}
\end{align}
where the first inequality is the triangle inequality, the second inequality uses \cref{eq:tvd_replacement_sampling} twice, and the last equality follows from the observations that 
\begin{enumerate}
\item  the distribution over $x^2 \in \{0,1\}^t$ defined by $\sum_u\BarPr_0[X^2 = x^2\mid U = u] \cdot \Prob[W=u]$ is the same as binomial $B \coloneqq \Bin(t, \frac{2}{3}\cdot \frac{3}{4}) = \Bin(t,\frac{1}{2})$.
\item the distribution over $x^2 \in \{0,1\}^t$ defined by $\sum_u\BarPr_1[Y^2 = x^2 \mid V = \bar{u}]\cdot\Pr[W=u]$ is the same as binomial $\Bin(t, \frac{2}{3}\cdot\frac{1}{4} + \frac{1}{3} \cdot 1) = \Bin(t,\frac{1}{2}) = B$.
\end{enumerate}
Combining \cref{eq:tvd_no_collision,eq:iid_conversion,eq:tvd_two_parts}, we see
\begin{align}
    \TVD &~\leq~ \frac{2t^2}{n} + \sum_{x \notin \Bad}\bigl|\Prob_0[X = x] - \Prob_1[Y = x]\bigr|
    \nonumber \\
    &~\leq~ \frac{2t^2}{n} + \sum_{x^1 \notin \Bad^1}\frac{2t^2}{n}\Pr[X^1 = x^1] \leq \frac{4t^2}{n},\nonumber
\end{align}
as required, where the first step uses \cref{eq:tvd_two_parts}, and the second step uses~\cref{eq:tvd_no_collision} and~(\ref{eq:iid_conversion}).
\end{proof}

\begin{remark}
    The lower bound for $R(\shuffle[\SMAJ_n])=\Omega(\sqrt n)$ is tight: for large enough constant $C$, when we query $C\sqrt n$ random positions of the input, there will be $\Omega(1)$ blocks in each of which we have queried at least $2$ entries. Then we can determine $\shuffle[\SMAJ_n]$ with an $\Omega(1)$ advantage.
    We conjecture that $Q(\shuffle[\SMAJ_n]) = \Omega(n^{1/3})$.
\end{remark}

\subsection{Shuffled direct sum}
In the $\trivertex$ problem, the shuffling can be seen as arising from a two-step process. The first step shuffles \emph{among} edges incident to the target vertex $v$ (or not incident to $v$), and the second step shuffles \emph{between} edges incident or not incident to $v$. The latter type of shuffling motivates the notion of shuffled direct sum in~\cref{def:shuffled-direct-sum}, which we now recall.

Let $n,k\in \mathbb{N}$. Let $\Sigma$ be a finite non-empty set. For a function $f \colon \Sigma^n \to \{0,1\}$, the $k$-shuffled direct sum of $f$ is defined to be 
    \begin{equation}
        \shuffle^k[f]\colon D \subseteq  (\Sigma\times [k])^{kn} \to \{0,1\}^k, \quad \text{where}
    \end{equation}
\begin{enumerate}
    \item the set $D$ consists of all elements $x\in (\Sigma\times [k])^{kn}$ satisfying
    $x = ((v_1,c_1),\dots, (v_{kn},c_{kn}))$ such that, for all $j\in [k]$, there are exactly $n$ indices $i\in [kn]$ with $c_i  = j$.
    \item $\shuffle^k[f](x)$ is defined to equal $(f(v^{(1)}),\dots,f(v^{(k)}))$, where $v^{(j)}$ is the subsequence of $v\coloneqq v_1\dots v_{kn}$ indexed by those $i\in [kn]$ such that $c_i = j$.
\end{enumerate}

Naturally, the question is how $Q(\shuffle^k(f))$ relates to  $Q(\shuffle(f))$.  Since $\shuffle^k(f)$ restricts to the $k$-fold direct sum of $f$, the direct sum theorem for the quantum query complexity \cite{ambainis_direct_sum,reichardt_direct_sum} immediately gives the following.
\begin{fact}
    For every $f\colon\Sigma^n \to \{0,1\}$, $Q(\shuffle^k[f])=\Omega(k\cdot Q(f))$.
\end{fact}

The following result shows that the above lower bound is tight when the domain of $f$ is $\{0,1\}^n$.
\begin{proposition}\label{prop:shufflek}
     For every symmetric $f\colon \{0,1\}^n \to \{0,1\}$, $Q(\shuffle^k[f]) = \tildeTheta(k \cdot Q(f))$.
\end{proposition}

\begin{proof}
    As in the proof of \cref{prop:hiding_generic}, we use the fact from \cite{polynomial_2001} that $Q(f) = \Theta(\sqrt{n(n- \Gamma(f))})$, where $\Gamma(f) \coloneqq \min\{\abs{2k-n+1}\mid f_k \neq f_{k+1}\}$ and $f_k \coloneqq f(x)$ for all $x$ such that $\abs{x} = k$. Observe that $f_k$ must be constant for all $k\in \{(n-\Gamma(f))/2,\dots, (n+\Gamma(f)-2)/2\}$. Call this constant $b\in \{0,1\}$. We also write $A \coloneqq (n-\Gamma(f))/2$ and $B\coloneqq (n+\Gamma(f)-2)/2$ for convenience.
    
    We consider the quantum algorithm for $\shuffle^k[f]$ defined in \cref{fig:algo_shuffle}. The quantum query complexity of this algorithm, accounting for error suppression, is 
    \begin{equation}
        \tildeO(\sqrt{kA\cdot (kn)} + \sqrt{k(n-B)\cdot (kn)}) = \tildeO(k\sqrt{n(n-\Gamma(f)})) = \tildeO(k\cdot Q(f)),
    \end{equation}
    using the basic fact that collecting (up to) $\alpha$ marked items from a list of $\beta$ items costs $\tildeO(\sqrt{\alpha \beta})$ quantum queries by Grover search \cite{vanApeldoorn2024basicquantum}. 
    
    The correctness of the algorithm can be argued similarly to \cref{prop:hiding_generic} so we omit it.
\end{proof}

We conjecture that \cref{prop:shufflek} still holds for symmetric $f\colon \Sigma^n\to \{0,1\}$ even if $\Sigma$ is non-Boolean.

\begin{conjecture}\label{conj:shuffling_composition}
    For every symmetric $f\colon \Sigma^n \to \{0,1\}$, 
    \begin{equation}
        \Omega(k\cdot Q(f)) \leq Q(\shuffle^k[f]) \leq O(k\log(k) \cdot Q(f)).
    \end{equation}
\end{conjecture}
We remark that this conjecture holds for randomized query complexity, i.e., with $Q(\cdot)$ replaced by $R(\cdot)$: the upper bound follows by a similar argument to the proof of \cref{prop:symmetric_partial}, while the lower bound follows by \cite[Theorem 2]{jain_rand_directsum_10}.

\begin{figure}[H]
        \centering
        \fbox{\parbox{0.925\textwidth}{
        \underline{\textbf{Algorithm for $\shuffle^k[f]$}}\\[1mm]
        On input $x \coloneqq ((v_1,c_1),\dots,(v_{kn},c_{kn}))\in (\{0,1\} \times [k])^{kn}$:
        \begin{enumerate}
            \item[I.] Use Grover search to collect up to $kA$ distinct indices $i\in [kn]$ such that $v_i = 1$, but as soon as we collect at least $A$ such indices $i$ from a fixed copy $c\in [k]$ of $f$, i.e., $c_i = c$, we \emph{stop} collecting all further indices $i$ from copy $c$ and record that ``copy $c$ is above lower threshold''. (This stopping can be done on the fly since $c_i$ tells which copy of $f$ it is associated with.) When no more indices $i$ can be collected, for each copy of $f$ not ``above lower threshold'', record the number of indices $i$ collected that are associated with it.
            \item[II.] Use Grover search to collect up to $k(n-B)$ distinct indices $i\in [kn]$ such that $v_i = 0$, but as soon as we collect at least $(n-B)$ such indices $i$ from a fixed copy $c\in [k]$ of $f$, i.e., $c_i = c$, we \emph{stop} collecting all further indices $i$ from copy $c$ and record that ``copy $c$ is below upper threshold''. When no more indices $i$ can be collected, for each copy of $f$ not ``below upper threshold'', record the number of indices $i$ collected that are associated with it.
        \end{enumerate}
    
        Then:
        \begin{enumerate}
            \item  For copies of $f$ that are ``above lower threshold'' and ``below upper threshold'', output $b$. 
            \item   For copies of $f$ that are not ``above lower threshold'', output $f_k$ where $k$ is the number of indices $i$ recorded that are associated with that copy in step I. 
            \item   For copies of $f$ that are not ``below upper threshold'', output $f_{n-k}$ where $k$ is the number of indices $i$ recorded that are associated with that copy in step II. 
        \end{enumerate}
        }
        \hspace{0pt}
        }
    \caption{Algorithm for $\shuffle^k[f]$.}\label{fig:algo_shuffle}
    \end{figure}

\section{Triangle Finding}\label{sec:tri-finding}
In this section, we study the triangle problem, $\tri$. We first show that this problem bridges $\kdist{3}$ and $\ksum{3}$. Then, we give a nearly tight characterization of the problem's quantum query complexity when the input has low maximum degree, for example, if it is a random sparse graph.

\subsection{Bridging \texorpdfstring{$\kdist{3}$ and $\ksum{3}$}{3-DIST and 3-SUM}}
We formalize the aforementioned connection of $\tri$ to $\kdist{3}$ and $\ksum{3}$ in the following proposition.
\begin{proposition}\label{prop:two_way_reduct}
The following relationships hold between $\kdist{3}_m$, $\tri_m$, and $\ksum{3}_m$.
\begin{center}
\begin{enumerate*}
    \item \label{itm:3-dist_tri} $Q(\kdist{3}_m) \leq O(Q(\tri_m))$ \quad \text{and} \quad \quad
    \item \label{itm:tri_3-sum} $Q(\tri_m) \leq O(Q(\ksum{3}_m))$.
\end{enumerate*}
\end{center}
In particular, if $Q(\kdist{3}_m) \geq \Omega(m^\alpha)$, then $\Omega(m^{\alpha})\leq Q(\tri_m) \leq O(m^{3/4})$.
\end{proposition}

For comparison, we note that the ``trivial'' lower bound on $Q(\tri_m)$ by reducing from $\OR_{m-2}$ is $Q(\tri_m)=\Omega(\sqrt{m})$ while the above proposition gives $Q(\tri_m)=\Omega(m^{2/3})$ as the latter is the best-known lower bound on $\kdist{3}_m$~\cite{aaronson_shi}.

\begin{proof}
We prove each item as follows.

\cref{itm:3-dist_tri} We show $Q(\kdist{3}_m) \leq O(Q(\tri_m))$ by the following reduction. Let $x$ be an input to $\kdist{3}_m$. Assume wlog that $m$ is a multiple of $3$. Uniformly randomly partition $[m]$ into $3$ sets $P_1, P_2, P_3$ of length $m/3$ each. We construct a length-$m$ edge list $y$ of a graph on vertex set $V \coloneqq \{(x_i,v):i \in [m], v \in [3]\}$ as follows. For each $j \in [m]$, $y_j$ is the edge $\{(x_j,\ell), (x_j,\ell+1 \bmod 3)\}$ where $\ell \in [3]$ is such that $j \in P_\ell$. 
Note that for each $j \in [n]$, querying $y_j$ only requires querying $x_j$. Compute $\tri(y)$ and output the result. Repeat the entire procedure a large but constant number of times, if the output of some repeat is $1$, output $1$, otherwise, output $0$.

We now show correctness. Suppose $x$ contains a $3$-collision in $x$. That is, there exist indices $i_1, i_2, i_3$ and an element $z$ such that $x_{i_1} = x_{i_2} = x_{i_3} = z$. Then, each of the sets $P_1, P_2, P_3$ will contain exactly one index from $\{i_1, i_2, i_3\}$ with probability at least $2/9$. In that case, $y$ will contain the edges $\{(z,1), (z,2)\}$, $\{(z,2), (z,3)\}$, $\{(z,3), (z,1)\}$, which form a triangle. 

Conversely, suppose $y$ contains a triangle. Then there must exist an element $z$ such that $y$ contains the edges $\{(z,1), (z,2)\}$, $\{(z,2), (z,3)\}$, $\{(z,3), (z,1)\}$. Therefore, there exist $j_1\in P_1$, $j_2\in P_2$, $j_3\in P_3$ such that $x_{j_1}=x_{j_2} = x_{j_3} = z$, which implies that $z$ occurs (at least) $3$ times in $x$.

\cref{itm:tri_3-sum} We show $Q(\tri_m) \leq O(Q(\ksum{3}_m))$ by the following reduction. Assume wlog that $m$ is a multiple of $3$. Let $\tilde{x}$ be an input to $\tri_m$. Apply a uniformly random permutation to $\tilde{x}$ and denote the resulting string as $x$. Choose a uniformly random element $s\in (\{-1,1\}^2)^m$. For a given $i\in [m]$, we map the edge $x_i = \{a,b\}$, where $a<b$, to the $4$-dimensional vector $(0,s_i(1)a,s_i(2)b,-1)$ if $i\in  [1,m/3]$, $(s_i(1)a,0,s_i(2)b,-2)$ if $i\in [m/3+1,2m/3]$, and $(s_i(1)a,s_i(2)b,0,3)$ if $i\in[2m/3+1,m]$, where $s_i = (s_i(1), s_i(2))\in \{-1,1\}^2$. Denote the string after the mapping by $y\in (\mathbb{Z}^4)^m$. Compute $\ksum{3}(y)$ and output the result. Repeat the entire procedure a large but constant number of times, if the output of some repeat is $1$, output $1$, otherwise, output $0$.

We now show correctness. Write $y^{i}\coloneqq y[(i-1)m/3+1\twodots im/3]$ for $i\in \{1,2,3\}$. Suppose $\tilde{x}$ contains the triangle $\{a,b\}$, $\{b,c\}$, $\{a,c\}$ where $a<b<c$. Then with probability at least a constant, $y$ will contain symbols $(0,-b,-c,-1)$ in $y^1$, $(-a,0,c,-2)$ in $y^2$, and $(a,b,0,3)$ in $y^3$, which sum to $0$. 

Conversely, suppose $y$ is a yes-instance for $\ksum{3}$, then $y$ must contain symbol $(0,a_1,b_1,-1)$ in $y^1$, $(a_2,0,b_2,-2)$ in $y^2$, and $(a_3,b_3,0,3)$ in $y^3$ that sum to $0$ for some $a_i,b_i\in \mathbb{Z}$. (Note that the last coordinate ensures that one symbol must come from each of $y^1,y^2,y^3$ in the yes-certificate.) Therefore $\abs{a_2}=\abs{a_3}$ (call $a$), $\abs{a_1} = \abs{b_3}$ (call $b$), and $\abs{b_1} = \abs{b_2}$ (call $c$), so $x$ contains edges $\{b,c\}$, $\{a,c\}$, and $\{a,b\}$, which form a triangle.

The ``in particular'' part follows from $Q(\ksum{3}_m) \leq O(m^{3/4})$  \cite{childs_eisenberg,quantum_walk_ed}.  
\end{proof}

\begin{remark}

Notice that the lower bound of $\tri_m$ in \Cref{prop:two_way_reduct} applies when $n = \Omega(m)$. On the other hand, when $n = O(1)$, we have $Q(\tri_m) = O(\sqrt{m})$ since we can search whether one of the size-$3$ subsets of $[n]$ form a triangle. What is the complexity of $\tri_m$ when $n = \omega(1)$ and $n = o(m)$? In particular, can we show a better lower bound when $n = \Theta(\sqrt{m})$, which also captures dense graphs? The answer is yes, and we can reduce $\ED_{n^2}$ to $\tri_{\Theta(n^2)}$ (with $\Theta(n)$ vertices) as follows.

Let $x \in ([n] \times [n])^{n^2}$ be an input to $\ED_{n^2}$. We create an edge list $y$ of size $m = n^2 + n$ as follows. Let $V = \{j^A, j^B, j^C: j \in [n]\}$. For each $j \in [n]$, add an edge $(j^A, j^B)$ in $y$. Partition the set $[n^2]$ into two sets $A$ and $B$ of equal size. For each $i \in [n^2]$, denote $x_i = (j_i,k_i)\in ([n]\times [n])$, if $i \in A$ (respectively $B$), add the edge $(j_i^A, k_i^C)$ (respectively $(j_i^B, k_i^C)$) in $y$.    
We see the correctness of this reduction as follows. Suppose $x$ is a positive instance of $\ED_{n^2}$. Then, with probability $1/2$, we will have indices $i_1$ and $i_2$ such that $i_1 \in A$, $i_2 \in B$ and $x_{i_1} = x_{i_2} = (j,k)$ for some $j, k \in [n]$. Thus, we will have $(j^A, j^B), (j^A, k^C), (j^B, k^C) \in y$ so $y$ will be a positive instance of triangle. On the other hand, suppose that $y$ is a positive instance of triangle. Then, for some $j, k \in [n]$, we will have $(j^A, j^B), (j^A, k^C), (j^B, k^C) \in y$ since this is the only way there could be a triangle in $y$. It follows that $(j,k)$ must have appeared at least twice in $x$ so $x$ is a positive instance of $\ED_{n^2}$. 
\end{remark}

\subsection{The uniformly random edge list}
We now turn to the average-case analysis of $\tri$. In particular, we consider the random graph model with $m$ (multi-)edges, such that the $m$ edges are each sampled independently and uniformly from $\binom{[n]}{2}$. In other words, we consider the sample-with-replacement model. In this model, the input edge list $x$ can be written succinctly as $x\leftarrow \binom{[n]}{2}^m$.

Denote the graph on $n$ vertices formed by the edges in $x$ by $\calG(x)$. We use $\maxdeg(\calG(x))$ or simply $\maxdeg(x)$ to denote the max vertex degree of $\calG(x)$. A wedge in $\calG(x)$ is a length-two path. Since $x$ could contain repeated edges, $\calG(x)$ could be non-simple. In particular, the degree and wedge counts of $\calG(x)$ account for the multiplicities of repeated edges. Later, we will also consider strings $x\in (\binom{[n]}{2}\cup \{\bot\})^m$, where $\bot$ means the input edge has not yet been revealed.

Our focus will be the $m=\Theta(n)$ regime of sparse graphs. In this regime, we obtain near-matching upper and lower bounds for the quantum query complexity of triangle finding. The following basic facts about sparse random graphs will be useful. We defer the proof to \cref{app:random_graphs}.
\begin{fact}[Sparse Random Graph]
\label{fact:sparse-random-graph}
Let $x\leftarrow \binom{[n]}{2}^m$.
\begin{enumerate}
    \item \label{enu:random-graph-maxdeg} (Low max degree) For $m  \leq O(n)$, 
    \begin{equation}
        \Pr[\maxdeg(x) \leq 2\log (n) /\log\log (n)] \geq  1-o(1).
    \end{equation}
    \item \label{enu:random-graph-triangle} (Existence of a triangle) For $m \geq \Omega(n)$, 
    \begin{equation}
        \Pr[x \textup{ contains a triangle}] \geq  \Omega(1).
    \end{equation}
\end{enumerate} 
In particular, for $m = \Theta(n)$, 
\begin{equation}
    \Pr[\maxdeg(x) \leq 2\log(n) /\log\log(n) \textup{ and $x$ contains a triangle}] \geq \Omega(1).
\end{equation}
\end{fact}

For a technical reason, we also need the following simple fact.
\begin{fact}[Vertex Avoidance]
\label{fact:vertex-avoid}
    Let $x\in\binom{[n]}{2}^m$ and let $d\coloneqq \maxdeg(x)$. 
    Let  $T\subseteq [n]$ and let $t$ denote its size. Let $s\in \{0,1,\dots,m\}$. Then,
    \begin{equation}\label{eq:vertex_avoidance}
        \Pr\Bigl[\bigl\{v\in [n]\mid \textup{$v\in x_i$ for some $i\in S$}\bigr\} \cap T = \varnothing \Bigm| S \leftarrow {\scalebox{1.25}{$\binom{[m]}{s}$}}\Bigr] \ge 1-\frac{tds}{m-td+1}.
    \end{equation}
\end{fact}
Note that $\{v\in [n]\mid \textup{$v\in x_i$ for some $i\in S$}\}$ is formal notation for the set of vertices involved in the edges in $x$ at positions in $S$. 
\begin{proof}
    By direct calculation, the left-hand side of \cref{eq:vertex_avoidance} is at least
    \begin{align}
        \binom{m-td}{s}\binom{m}{s}^{-1} &= \frac{(m-td)!(m-s)!}{m!(m-td-s)!}=\frac{(m-s)\cdot(m-s-1)\cdots (m-s-td+1)}{m\cdot(m-1)\cdots(m-td+1)}
        \nonumber\\
        &=\Bigl(1-\frac{s}{m}\Bigr)\cdots\Bigl(1-\frac{s}{m-td+1}\Bigr) \geq 1- \frac{tds}{m-td+1}.
        \nonumber\qedhere
    \end{align}
\end{proof}

We call a graph $x \in \binom{[n]}{2}^m$ \emph{good} if $\maxdeg(\calG(x)) = O( \log(n) / \log\log(n) )$ and $x$ contains a triangle. Our quantum lower and upper bounds will hold for good graphs.  By~\cref{fact:sparse-random-graph}, a random graph $x\leftarrow\binom{[n]}{2}^m$ is good with probability $\Omega(1)$ for $m = \Theta(n)$. 

\cref{fact:vertex-avoid} states that a small subgraph of a good graph $x$ does not touch any vertex from some small subset $T$ with high probability. The vertex set $T$ that the subgraph wants to avoid will be a triangle, and $s$ will be of size $o(m)$. Thus, the avoidance property holds almost surely.

\subsection{Triangle finding lower bound}

Our lower bound result reads:
\begin{theorem}\label{prop:triangle_search}
    Let $n\in \mathbb{N}$. Suppose $m, T\in \mathbb{N}$ are functions of $n$ such that  $m=\Theta(n)$ and $T \leq o(m^{5/7}/\log^{4/7}(m))$. Then, for every $T$-query quantum query algorithm $\calA_T$, we have
    \begin{equation}
        \Pr[\textup{$\calA_T(x) = (i,j,k)$  such that $x_i,x_j,x_k$ form a triangle}] \leq o(1),
    \end{equation}
    where $\calA_T(x)$ denotes the output of $\calA_T$ when its queries are made to $x$, and the probability is over $x \leftarrow \binom{[n]}{2}^m$ and the randomness of $\calA_T$.
\end{theorem}

A simple corollary of \cref{prop:triangle_search} is that  $Q(\tri_m)\geq \widetilde{\Omega}(m^{5/7})$. In fact, the lower bound holds even if we are promised the maximum degree of the input graph is low.
\begin{corollary}\label{cor:worst-case-search}
    For all $m\in \mathbb{N}$ and $d = 2\log (m)/\log\log(m)$,
    \begin{equation}
        Q(\tri_{m}) \geq Q(\tri_{m, d}) \geq  \widetilde{\Omega}(m^{5/7}).
    \end{equation}
\end{corollary}
\begin{proof}
    The first inequality follows by restriction so it suffices to prove the second inequality. By the search-to-decision reduction (see preliminaries section), if $Q(\tri_{m,d}) \leq o(m^{5/7}/\log^{4/7+2}(m))$ then there exists a quantum algorithm for triangle search with query complexity $\leq o(m^{5/7}/\log^{4/7}(m))$ that is correct with at least constant probability. Now let $n\in \mathbb{N}$ and consider an input $x \leftarrow \binom{[n]}{2}^m$.
    By~\cref{fact:sparse-random-graph}, for $n=m$,
    $x$ contains a triangle and has degree at most $2\log (m)/\log\log(m)$ with at least constant probability, so a triangle can be found by the search algorithm using $o(m^{5/7}/\log^{4/7}(m))$ queries with at least constant probability, contradicting \cref{prop:triangle_search}.
\end{proof}

For the rest of this section, we prove \cref{prop:triangle_search}. 

\subsubsection{Setting up}
\paragraph{Recording query framework.}
We start by reviewing the recording query framework~\cite{zhandry_compressed,Hamoudi_2023}.
For convenience of notation, we write $N \coloneqq \binom{n}{2}$ and identify $[N]$ with $\binom{[n]}{2}$ for the rest of this section. A projector refers to a complex square matrix $P$ of some context-appropriate dimension such that $P^\dagger = P$ and $P^2 = P$. 

Let $m, T\in \mathbb{N}$ and $\Sigma$ be an alphabet.
Given a relation $R\subseteq [N]^m \times \Sigma$, a $T$-query quantum query algorithm for computing $R$ is specified by a sequence of unitary matrices $U_0,\dots,U_T$ acting on $\mathbb{C}^m\otimes \mathbb{C}^N \otimes \mathbb{C}^K$ for some $K\in \mathbb{N}$. $\mathbb{C}^K$ is referred to as a workspace register. 

For $t\in \{0,1,\dots,T\}$, the state of the algorithm at the $t$-th step given input $x\in [N]^m$ is defined by
\begin{equation}
    \ket{\psi_t^x} \coloneqq U_t\mathcal{O}_xU_{t-1}\mathcal{O}_x\cdots U_1\mathcal{O}_xU_0\ket{0}.
\end{equation}
Here, $\mathcal{O}_x$ is the quantum oracle of $x$, i.e., the unitary matrix acting on $\mathbb{C}^m\otimes \mathbb{C}^N \otimes \mathbb{C}^K$ defined by
\begin{equation}
    \forall (i,u,w)\in [m]\times [N]\times [K]\colon \mathcal{O}_x\ket{i,u,w} \coloneqq \omega_N^{ux_i}\ket{i,u,w},
\end{equation}
where the $\ket{i,u,w}$s form the computational basis of $\mathbb{C}^m\otimes \mathbb{C}^N \otimes \mathbb{C}^K$ and $\omega_N \coloneqq \exp(2\pi\mathrm{i}/N)$ is the $N$-th root of unity.

The algorithm finishes by measuring $\mathbb{C}^K$ in the computational basis, parses the output $w\in [K]$ as a string, and outputs a substring of $w$ at a fixed location.\footnote{Formally, $\mathbb{C}^K$ is identified with a tensor product space $\bigotimes_{i=1}^k\mathbb{C}^{\Sigma_i}$ for some alphabets $\Sigma_i$ and the output on measuring the string $w = w_1 \dots w_k$, where $w_i\in \Sigma_i$, is the substring  $w' \coloneqq w_1 \dots w_{k'}$ for some $k'\leq k$.} 
For $x\in [N]^m$, let $W_x \coloneqq \{s\in \Sigma \mid (x,s)\in R\}$ denote the set of all desired outputs associated with $x$, and let $\Pi_{\success}^x$ be the projector onto all basis states $\ket{i,u,w}$ such that the substring $w'$ of $w$ at the fixed location satisfies $w' \in W_x$. The success probability of the quantum algorithm on input $x\in [N]^m$ is then defined by $\norm{\Pi_{\success}^x\ket{\phi_T}}^2$.

Following \cite{A00}, a quantum query algorithm can be viewed as acting on basis states $\ket{i,u,w}\ket{x}$, where $x\in [N]^m$ and $\ket{x} \coloneqq \bigotimes_{i\in[m]} \ket{x_i}$ is the state on an additional \emph{input register}. Suppose $\calD$ is a probability distribution over $[N]^m$, the state of the algorithm at the $t$-th step given an input sampled according to $\calD$ can be represented by
\begin{equation}\label{eq:standard_model_state}
    \ket{\psi_t^{\calD}} \coloneqq (U_t\otimes\id)\mathcal{O}(U_{t-1}\otimes\id)\mathcal{O}\cdots (U_1\otimes\id)\mathcal{O}(U_0\otimes\id)(\ket{0}\ket{\calD}),
\end{equation}
where  $\mathcal{O}$ is the standard query operator that maps each basis state $\ket{i,u,w}\ket{x}$ to $(\mathcal{O}_x\ket{i,u,w})\ket{x}$, and
\begin{equation}
    \ket{\calD} \coloneqq \sum_{x\in[N]^m}\sqrt{\Prob[x\leftarrow\calD]}\ket{x}.
\end{equation}

Let $\Pi_{\success} \coloneqq \sum_{x\in[N]^m}\Pi_{\success}^x\otimes\ketbra{x}{x}$. The success probability of the algorithm on $\calD$ is defined to be $\norm{\Pi_{\success}\ket{\psi^{\calD}_T}}$. It is not hard to see that
\begin{equation}
    \norm{\Pi_{\success}\ket{\psi_T^{\calD}}}^2 = \sum_{x\in[N]^m} \Prob[x\leftarrow\calD]\norm{\Pi_{\success}^x\ket{\psi^x_T}}^2.
\end{equation}
which shows that $\norm{\Pi_{\success}\ket{\psi_T^{\calD}}}^2$ equals the expected success probability of the quantum query algorithm when the input is sampled from the distribution $\calD$.

The recording query framework works with an alternative view of \cref{eq:standard_model_state}. We follow the prescription of this method described by Hamoudi and Magniez in \cite{Hamoudi_2023} that specializes to the case where 
\begin{equation}
    \calD = \calD_1 \times \calD_2 \times \dots \times \calD_m
\end{equation}
is a product distribution on $[N]^m$. To reach this alternative view, we first augment the alphabet $[N]$ by an additional character $\bot$, and extend the action of $\mathcal{O}$ by defining
\begin{equation}
    \mathcal{O}\ket{i,u,w}\ket{x} \coloneqq \ket{i,u,w}\ket{x},
\end{equation}
whenever $x_i = \bot$. (The lack of change in the above equation explains why $x_i=\bot$ can be thought of as the algorithm's lack of knowledge of the value of $x_i$.)

Now, define $\ket{\calD_i} \coloneqq \sum_{y\in[N]}\sqrt{\Prob[y\leftarrow \calD_i]}\ket{y}$ so that the initial state of the input register in the standard query model can be written as $\ket{\calD} = \bigotimes_{i\in[m]}\ket{\calD_i}$. In the recording query model, the corresponding initial state of the input register in the recording query model is defined to be $\bigotimes_{i\in[m]}\ket{\bot} = \ket{\bot^m}$. 

Then, we define the recording query operator, which is $\calO$ under a change of basis defined using $\cal{D}$.
\begin{definition}[{\cite[Definition 3.1]{Hamoudi_2023}}]\label{def:recording_query_operator}
    For all $i\in[m]$, define the unitary matrix $\calS_i$ acting on $\mathbb{C}^{N+1}$ by
    \begin{equation}
        \calS_i \colon
        \left\{
        \begin{alignedat}{2}
           &\ket{\bot} &&\mapsto \ket{\mathcal{D}_i},\\ 
           &\ket{\mathcal{D}_i} &&\mapsto \ket{\bot},\\ 
           &\ket{\phi} &&\mapsto \ket{\phi}, \quad \text{if } \langle \phi | \bot \rangle = \langle \phi | \mathcal{D}_i \rangle = 0.
        \end{alignedat}
        \right.
    \end{equation}
    The registers of $\ket{i,u,w}\ket{x} \in \calH\coloneqq  \mathbb{C}^m\otimes \mathbb{C}^N\otimes \mathbb{C}^K \otimes \mathbb{C}^{N^m}$ are labeled as $\ket{i}_Q\ket{u}_P\ket{w}_W\ket{x}_I$. The register labels $Q,P,W,I$ stand for ``query'', ``phase'', ``work'', and ``input'' respectively. Then, define the following unitary operators acting on $\calH$.
    \begin{align}
        \mathcal{T}_{\calD} &\coloneqq \id_{QPW}\otimes\bigotimes_{i\in[m]}\calS_i,
        \\
        \calS_{\calD} &\coloneqq \sum_{i\in[m]}\ketbra{i}{i}_Q\otimes\id_{PW}\otimes\bigotimes_{j=1}^{i-1}\id_{I_j}\otimes \calS_i\otimes\bigotimes_{j=i+1}^{m}\id_{I_j},
        \\
        \calR_{\calD} &\coloneqq \calS_{\calD}^\dagger\mathcal{O}\calS_{\calD}.\label{eq:recording_query_operator}
    \end{align}
    $\calR_{\calD}$ is referred to as the \emph{recording query operator}. 
\end{definition}
In the recording query model, the state of the quantum query algorithm after the $t$-th recording query is defined by 
\begin{equation}\label{eq:recording_model_state}
    \ket{\phi_t^{\calD}} \coloneqq (U_t\otimes\id_I)\calR_{\calD}(U_{t-1}\otimes\id_I)\calR_{\calD}\cdots (U_1\otimes\id_I)\calR_{\calD}(U_0\otimes\id_I)(\ket{0}\ket{\bot^m}).
\end{equation}
For convenience, we will often abuse notation and write $U_i$ for the $U_i \otimes \id_I$ in \cref{eq:recording_model_state}. 

The recording query model is particularly useful due to the following facts established in the previous works. First, the recording query model is essentially the same as the standard average-case query model up to a rotation on the input register at the end of the algorithm. Consequently, a measurement on the other registers sees no difference at all. Formally, we have
\begin{theorem}[{\cite[Theorem 3.3]{Hamoudi_2023}}]
\label{thm:recording-vs-standard}
For all $t\in \mathbb{Z}_{[0,T]}$, 
    \begin{equation}
        \ket{\phi^{\calD}_t} = \mathcal{T}_{\calD}\ket{\psi^{\calD}_t},
    \end{equation}
    where $\ket{\phi^{\calD}_t}$ and $\ket{\psi^{\calD}_t}$ are defined in \cref{eq:recording_model_state,eq:standard_model_state} respectively.
\end{theorem}

A crucial but easy-to-see fact in the recording query framework is\footnote{For intuition only, if we view each non-$\bot$ symbol in $x$ as contributing one unit of ``degree'', then this fact is akin to the polynomial method \cite{polynomial_2001}. We remark that (essentially) this fact \emph{alone} can be used to prove non-trivial lower bounds, see, e.g., \cite{beame_matrix_24}.}
\begin{fact}[{\cite[Fact 3.2]{Hamoudi_2023}}]\label{fact:t_non_bots}
    The state $\ket{\phi_t^{\calD}}$ is a linear combination of basis states $\ket{i,u,w}\ket{x}$ where $x$ contains at most $t$ entries different from $\bot$.
\end{fact}

Moreover, the effect of the recording query operator can be exactly calculated as
\begin{lemma}[{\cite[Lemma 4.1]{Hamoudi_2023}}]\label{lemma:recording_operator_effect}
    Let $i\in [m]$. Suppose $\calD_i$ is uniform over $[N]$. 
    Suppose the recording query operator $\calR_{\calD}$ is applied to a basis state $\ket{i,u,w}\ket{x}$. 
    If $u\neq 0$, the register $\ket{x_i}_{I_i}$ is mapped to
    \begin{equation}
    \begin{dcases}
        \sum_{y\in[N]}\frac{\omega_N^{u y}}{\sqrt{N}}\ket{y}, & \text{if }x_i=\bot,\\[1.5mm]
        \frac{\omega^{u x_i}_N}{\sqrt{N}}\ket{\bot}+\frac{1+\omega^{u x_i}_N(N-2)}{N}\ket{x_i}+\sum_{y\in[N]\backslash\{x_i\}}\frac{1-\omega^{u y}_N-\omega^{u x_i}_N}{N}\ket{y}, & \text{if } x_i\in[N],
    \end{dcases}
    \end{equation}
    and the other registers are unchanged. If $u = 0$ then none of the registers are changed.
\end{lemma}

\cref{fact:t_non_bots} is due to the ``reveal-on-demand'' or ``lazy-sampling'' nature of the recording query operator. \cref{lemma:recording_operator_effect} formally grounds the intuition that once a value of $x_i$ has been recorded, later queries should not significantly alter it (note the coefficient on $\ket{x_i}$ has a large norm).

As we are proving \cref{prop:triangle_search}, we will henceforth specialize $\calD$ to be the uniform (product) distribution on $[N]^m = \binom{[n]}{2}^m$. 
We abbreviate
\begin{equation}
    \ket{\phi_t},\ket{\psi_t}, \calT,\calS,\calR \quad \text{to} \quad \ket{\phi_t^{\calD}}, \ket{\psi_t^{\calD}}, \calT_{\calD}, \calS_{\calD}, \calR_{\calD},
\end{equation}
respectively. These states and operators depend on the $m$ and $n$ that parametrize $\calD$, but we leave these parameters implicit for notational convenience.

\paragraph{Roadmap for the proof of~\cref{prop:triangle_search}.}

Recall that $x\in \binom{[n]}{2}^m$, we write $\calG(x)$ for the graph on $n$ vertices formed by the edges in $x$. The notation trivially generalizes to $x\in (\binom{[n]}{2}\cup \{\bot\})^m$. 
We say $x$ (or $\ket{x}$) \emph{records} a triangle if $\calG(x)$ contains a triangle; we say $x$ (or $\ket{x}$) \emph{records} $k$ wedges if $\calG(x)$ contains $k$ wedges; and so on. Since there is a trivial bijection between $N= \binom{n}{2}$ and $\binom{[n]}{2}$, we do not distinguish $[N]$ and $\binom{[n]}{2}$, as for the recording query framework, the notation $[N]^m$ is arguably more natural. We will assume that $m = n$. 
This is without loss of generality since the only property of the regime $m = \Theta(n)$ that we use in the proof is $\maxdeg(x) \leq O(\log(n))$, which is promised with high probability by~\cref{fact:sparse-random-graph}. Nonetheless, we will keep the notation $m$ and $n$ separate until we need to use $m=n$. This is for clarity as most parts of the proof do not use $m=n$. Indeed, the only places that do are  \cref{cor_bad_event_bound} and the few results invoking it, which includes \cref{prop:triangle_search}.

Our analysis revolves around bounding the norm of the projection of $\ket{\phi_t}$ onto computational basis states $\ket{i,u,w}\ket{x}$ where $x$ records properties key to finding triangles, such as: $\calG(x)$ contains a large number of high-degree vertices; $\calG(x)$ contains a large number of wedges; and, of course, $\calG(x)$ contains a triangle. \cref{lemma:recording_operator_effect} allows us to carry out our analysis on a query-by-query basis. This discussion motivates the following definitions.

\begin{definition}[Recording projectors]
\label{def:recording_projectors} 
    For  \(R\subseteq\mathbb{R}\) and \(y\in \binom{[n]}{2}\cup \{\bot\}\), define the following projectors by giving the basis states onto which they project.\\ \vspace{-2mm}

    \begin{tabular}{@{}l@{\ }l}
    $\Pi^\triangle$            &\,: every basis state $\ket{i,u,w}\ket{x}$ such that $\calG(x)$ contains a triangle.
    \\[4pt]
    $\Pi^{\trianglewedge}_{R}$           &\,: every basis state $\ket{i,u,w}\ket{x}$ such that the number of wedges in $\calG(x)$ is in $R$.
    \\[4pt]
    $\Pi_y$                    &\,: every basis state $\ket{i,u,w}\ket{x}$ such that $u\neq 0$ and $x_i = y$.
    \end{tabular}
\end{definition}
\begin{definition}[Progress measures for wedges and triangle] For $t\in \mathbb Z_{[0,T]}$ and $r\geq 0$, define
    \begin{align*}
        \Lambda_{t,r}:=\norm{\Pi^{\trianglewedge}_{[r,\infty)}\ket{\phi_t}}
        \quad\text{and} \quad
        \Delta_{t}\coloneqq\norm{\Pi^\triangle\ket{\phi_t}}.
    \end{align*}
\end{definition}
Thus, $\Lambda_{t,r}$ measures the progress of recording $r$ wedges in $t$ queries; while $\Delta_t$ measures the progress of recording a triangle in $t$ queries. Our strategy is to first show that a large investment of queries is needed to record many wedges. That is, $\Lambda_{t,r}$ is negligible when $t$ is small and $r$ large. Then, we show that, to record a triangle, a large number of wedges must have been recorded during the execution of the algorithm. That is, $\Delta_t$ will be negligible as long as $\Lambda_{t,r}$ remains negligible for a somewhat large $r$.

A naive implementation of the above strategy falls short of giving the tight bound. One technical challenge is that the $t$-th query can, in principle,  lead to as many as $\Omega(t)$ extra wedges being recorded. This result in an overestimate on the wedge progress in each step and therefore an underestimate on the final query lower bound. 
However, as the input is drawn from the distribution $\calD$, the underlying graph has a maximum degree of $O(\log(n))$ with high probability. By excluding graphs with a maximum degree of $\Omega(\log^2(n))$, we ensure that each query records at most $\tildeO(1)$ wedges.
This motivates
\begin{definition}[Excluding projectors]
\label{def:excluding-projectors} 
    For \(d \in \mathbb{Z}_{\geq 0}\), vertex \(v\in [n]\),  we define the following projectors by giving the basis states onto which they project.\\ \vspace{-2mm}
    
    \begin{tabular}{@{}l@{\ }l}
    $\Pi^{\deg}_{v,\geq d}$    &\,: every basis state $\ket{i,u,w}\ket{x}$ such that the degree of $v$ in $\calG(x)$ is at least $d$.
    \\[4pt]
    $\Pi^{\deg}_{\geq d}$      &\,: every basis state $\ket{i,u,w}\ket{x}$ such that $\calG(x)$ contains a vertex with degree at least $d$.
    \\[4pt]
    $\Pi_{\Bad}$ &\,:  every basis state $\ket{i,u,w}\ket{x}$ such that $\calG(x)$ contains a vertex with degree at least $\ceil{\log(n)}^2$. 
    \\[4pt]
    \end{tabular}
\end{definition}
In the next three subsections, we discuss the above three pieces in detail, starting from controlling the high-degree vertices, and then analyzing $\Lambda_{t,r}$ and $\Delta_t$.

\subsubsection{Excluding high-degree graphs}

A random sparse graph is rarely a high-degree graph. We show how we can exclude these graphs in the recording query framework so that the wedge count does not jump by a large amount at each query.  While some aspects of our presentation are tailored to our specific problem, we have made an effort to generalize our proofs so that they can be applied in other contexts.

We start with a generic treatment stated in~\cref{lem:exclusion_lemma} that enables us to exclude unfavorable events from consideration during the analysis of recording progress, deferring their accounting to the final recording probability. \cref{lem:mirror} establishes that when the input is sampled from $\calD$, if a certain event occurs with negligible probability, then no quantum algorithm can reliably record that event. (As mentioned above, we have dropped $\calD$-superscripts, $\calD$-subscripts, and the identities tensored to the $U_i$s for notational convenience.)

\begin{lemma}[Exclusion Lemma]\label{lem:exclusion_lemma}
    Let $t\in \mathbb{Z}_{[1,T-1]}$. Let $\Pi_1,\Pi_2,\dots,\Pi_{t-1}$; $\Pi_1',\Pi_2',\dots,\Pi_{t-1}'$  be projectors. For $i\in [t]$, let $\ket{\alpha_i}$ and $\ket{\alpha_i'}$ denote the (possibly unnormalized) states
    \begin{align*}
        \ket{\alpha_i} \coloneqq& U_i\calR(\id-\Pi_{i-1})U_{i-1}\calR(\id-\Pi_{i-2})U_{i-2}\calR\cdots (\id-\Pi_1)U_1\calR U_0(\ket{0}\ket{\bot^m}),
        \\
        \ket{\alpha_i'} \coloneqq& U_i\calR(\id-\Pi_{i-1}')(\id-\Pi_{i-1})U_{i-1}\calR(\id-\Pi_{i-2}')(\id-\Pi_{i-2})U_{i-2}\calR\cdots (\id-\Pi_1')(\id-\Pi_1)U_1\calR U_0(\ket{0}\ket{\bot^m}).
    \end{align*}
    Then,
    \begin{equation}
        \norm{\ket{\alpha_t} - \ket{\alpha_t'}} \leq \sum_{i=1}^{t-1}\norm{\Pi_i'(\id-\Pi_i)\ket{\alpha_i}}.
    \end{equation}
    In particular, for every projector $\Pi_{\record}$,
    \begin{equation}
        \norm{\Pi_{\record}\ket{\alpha_t}}\leq \norm{\Pi_{\record}\ket{\alpha'_t}}+\sum_{i=1}^{t-1}\norm{\Pi_i'(\id-\Pi_i)\ket{\alpha_i}}.
    \end{equation}
\end{lemma}
\begin{proof}
    We adopt the product notation $\prod$ for the non-commutative product  in right-to-left order, i.e., for $a,b\in \mathbb{N}$,
    \begin{equation}
        \prod_{i=a}^b X_i \coloneqq
        \begin{cases}
            X_b X_{b-1}\cdots X_{a+1}X_a, & \text{if $a\leq b$},\\
            \id, & \text{if $a>b$}.
        \end{cases}
    \end{equation}
    For $i\in\mathbb{Z}_{[1,t-1]}$, define $M_i \coloneqq (\id-\Pi_i)U_i$. Then, 
    \begin{align}
        \ket{\alpha_t}
        &=  U_t\calR\Bigl(\prod_{i=1}^{t-1}(\Pi_i'+(\id-\Pi_{i}'))M_i\calR\Bigr)U_0(\ket{0}\ket{\bot^m}) \notag \\
        &= U_t\calR\biggl(\Bigl(\prod_{i=1}^{t-1}(\id-\Pi_{i}')M_i\calR\Bigr)+\sum_{i=1}^{t-1}\Bigl(\prod_{j=i+1}^{t-1}(\id-\Pi_{j}')M_j\calR\Bigr)\Pi_i'\Bigl(\prod_{k=1}^{i}M_k\calR\Bigr)\biggr)U_0(\ket{0}\ket{\bot^m}) \notag \\
        &= \ket{\alpha'_t}+\sum_{i=1}^{t-1}U_t\calR\Bigl(\prod_{j=i+1}^{t-1}(\id-\Pi_{j}')M_j\calR\Bigr)\Pi_i'(\id-\Pi_i)\ket{\alpha_i}, \label{eq:exclusion_proof}  
    \end{align}
    where the second equality can be interpreted as the hybrid argument \cite{bbbv97}.
    
    Observe that
    \begin{equation}
        \biggl\|U_t\calR\Bigl(\prod_{j=i+1}^{t-1}(\id-\Pi_{j}')M_j\calR\Bigr)\Pi_i'(\id-\Pi_i)\ket{\alpha_i}\biggr\| \leq \norm{\Pi_i'(\id-\Pi_i)\ket{\alpha_i}},
    \end{equation}
    because the $\Pi_j$, $\Pi_j'$s are projectors, and $\calR$ and the $U_j$s are unitaries. 

    Therefore, applying the triangle inequality on \cref{eq:exclusion_proof} gives
    \begin{equation}
        \norm{\ket{\alpha_t}-\ket{\alpha'_t}} \leq \sum_{i=1}^{t-1}\norm{\Pi_i'(\id-\Pi_i)\ket{\alpha_i}},
    \end{equation}
    The ``in particular'' part of the lemma follows from $\norm{\Pi_{\record}\ket{\alpha_t} - \Pi_{\record}\ket{\alpha'_t}} \leq  \norm{\ket{\alpha_t} - \ket{\alpha'_t}}$ and the reverse triangle inequality.
\end{proof}

\paragraph{The Hamming events.}
Now we study a rather general situation using the Exclusion Lemma. Consider the product distribution $\calD =  \calD_1 \times \dots \times \calD_m$, where each $\calD_i$ is over the finite domain $[N]$. 

For each $i\in [m]$, fix some partition 
\begin{equation}
[N] = \Sigma_0^{(i)}\dot{\cup}\Sigma_1^{(i)}.
\end{equation}
Symbols from $\Sigma_1^{(i)}$ can be interpreted as a ``logical 1'' and correspond to the symbols of interest for the $i$-th index. We abuse the Hamming weight notation $|\cdot|$  to apply to $x\in \Sigma_1^{(1)}\times\cdots\times\Sigma^{(m)}_1$, such that
\begin{equation}
    \abs{x} \coloneqq \Bigl|\bigl\{i\in[m] \mid x_i \in \Sigma_1^{(i)}\bigr\}\Bigr|.
\end{equation}
Then we have the natural probability distribution in the form of the product of independent Bernoulli distributions, where the probability $p_i$ of the $i$-th symbol being a logical 1 is
\begin{equation}
    p_i\coloneqq \Prob[e\in\Sigma_1^{(i)}\mid e\leftarrow\calD_i].
\end{equation}
We also write
\begin{equation}
    p\coloneqq\max_{i\in[m]}p_i.
\end{equation}
The above definition hints at the events that we are going to focus on: $\abs{x}$ being large. Define the binary-valued logical 0-1 measurement on the $i$-th index as
\begin{equation}
    \Bigl\{\Pi_0^{(i)} = \textstyle\sum_{c\in\Sigma_0^{(i)}}\ketbra{c}{c}, \ \Pi_1^{(i)} = \textstyle\sum_{c\in\Sigma_1^{(i)}}\ketbra{c}{c} \Bigr\}.
\end{equation}

We first show the operator $\calS_i$ causes little ``leakage" from the logical-$0$ subspace to the logical-$1$ subspace by defining the following \emph{leakage operator} and bounding its norm.
\begin{definition}[Leakage operator]\label{def:leakage_op}
For $i\in [m]$, the leakage operator on the $i$-th index is
    \begin{equation}
    L_i \coloneqq \Pi_0^{(i)}\calS_i^\dagger \Pi_1^{(i)}\calS_i\Pi_0^{(i)}\in  \mathbb{C}^{([N] \cup \{\bot\}) \times ([N] \cup \{\bot\}) }.
\end{equation}
\end{definition}
\begin{lemma}[Bound on leakage]\label{lemma_leakage_bound}
    For all $i\in [m]$,
    \begin{equation}
        \norm{L_i} = p_i(1-p_i).
    \end{equation}
\end{lemma}
\begin{proof}
    When $p_i=0$ or $p_i=1$, the proof is trivial as one of the projectors $\Pi_0^{(i)}$ or $\Pi_1^{(i)}$ is 0, so we assume $p_i\in (0,1)$. Notice $L_i$ is sandwiched by projectors $\Pi_0^{(i)}$, so $L_i$ is a block matrix that is all zero except on the $\Sigma_0^{(i)} \times \Sigma_0^{(i)}$ block. Therefore, the  norm of $L_i$ equals its norm restricted to that block. For $x\in [N]$, write $p_{i,x}\coloneqq \Prob[z = x\mid z\leftarrow\calD_i]$. Then, for all $k,l \in \Sigma_0^{(i)}$,
    
    \begin{align}
        \bra{k}L_i\ket{l} 
        =& \bra{k}\Pi_0^{(i)}\calS_i^\dagger \Pi_1^{(i)}\calS_i\Pi_0^{(i)}\ket{l} \notag
        \\
        =&\bra{k}\calS_i^\dagger \Pi_1^{(i)}\calS_i\ket{l} \notag
        \\
        =&\Bigl(\braket{k|\calD_i}\bra{\calD_i} + (\bra{k} - \braket{k|\calD_i}\bra{\calD_i})\Bigr) \, \calS_i^\dagger \Pi_1^{(i)}\calS_i \, \Bigl(\braket{l|\calD_i}\ket{\calD_i} + (\ket{l} - \braket{l|\calD_i}\ket{\calD_i})\Bigr) \notag
        \\
        =&\Bigl(\sqrt{p_{i,k}}\bra{\bot} + (\bra{k} - \sqrt{p_{i,k}}\bra{\calD_i})\Bigr) \, \Pi_1^{(i)} \, \Bigl(\sqrt{p_{i,l}}\ket{\bot} + (\ket{l} - \sqrt{p_{i,l}}\ket{\calD_i})\Bigr) \notag \\
        =&\sqrt{p_{i,k}\, p_{i,l}}\bra{\calD_i}\Pi_1^{(i)}\ket{\calD_i}=\sqrt{p_{i,k}\, p_{i,l}} \ {\textstyle \sum_{x\in\Sigma_1^{(i)}}} \, p_{i,x}=\sqrt{p_{i,k}\, p_{i,l}} \, p_i,\label{eq_leakage_op_0_block}
    \end{align}
    where the fourth equality uses \cref{def:recording_query_operator} of $\calS_i$ and the fact that $\ket{x} - \braket{x|\calD_i} \ket{\calD_i}$ is orthogonal to $\ket{\calD_i}$ and $\ket{\bot}$ for all $x\in [N]$. 
    
    Now, define $\ket{\varphi}\in\mathbb{C}^{\Sigma_0^{(i)}}$ by
    \begin{equation}
        \ket{\varphi}\coloneqq \frac{1}{\sqrt{1-p_i}}{\textstyle \sum_{x\in\Sigma_0^{(i)}}}\sqrt{p_{i,x}} \, \ket{x},
    \end{equation}
    which is normalized since
    \begin{equation}
        \braket{\varphi|\varphi}= \frac{1}{(\sqrt{1-p_i})^2}{\textstyle \sum_{x\in\Sigma_0^{(i)}}}(\sqrt{p_{i,x}})^2 = \frac{1}{1-p_i}{\textstyle \sum_{x\in\Sigma_0^{(i)}}} \, p_{i,x} = 1.
    \end{equation}
    
    From \cref{eq_leakage_op_0_block}, we see that the $\Sigma_0^{(i)} \times \Sigma_0^{(i)}$ block of $L_i$ is exactly $p_i(1-p_i)\, \ketbra{\varphi}{\varphi}$ since
    \begin{equation}
        \bra{k}(\, p_i(1-p_i) \, \ketbrasame{\varphi} \,)\ket{l} = p_i\sqrt{p_{i,k}\, p_{i,l}}.
    \end{equation}
    
    Therefore,
    \begin{equation}
        \norm{L_i} = p_i(1-p_i)\, \norm{\ketbra{\varphi}{\varphi}} = p_i(1-p_i),
    \end{equation}
    as required.
\end{proof}

The Chernoff bound tells us that the probability that $\abs{x}$ is large decays exponentially. We will use the following version of Chernoff bound specialized to i.i.d. Bernoulli random variables.
\begin{lemma}[Chernoff bound]\label{lemma_chernoff}
    Let $m\in \mathbb{N}$, $p\in [0,1]$, and $\mu \coloneqq pm$. Let $X_1,\dots, X_m$ be independent and identically distributed Bernoulli random variables such that $\Pr[X_1 = 1] = p$. Then, for every $c > 0$,
    \begin{equation}
        \Pr[X_1+\dots + X_m \geq c] \leq \Bigl(\frac{e \cdot \mu}{c}\Bigr)^c.
    \end{equation}
\end{lemma}
\begin{proof}
    We may assume $p \neq 0$, else the lemma trivially holds. If $p>0$, so that $\mu>0$, there are two cases:
    \begin{enumerate}
        \item $0<c\leq \mu$: in this case, $\mu/c\geq 1$ so $(e\mu/c)^c \geq e^c > 1$ so the lemma trivially holds.
        \item $c>\mu$: in this case, the lemma follows from \cite[Theorem 4.4 (part 1)]{MitzenmacherUpfal_17}. \qedhere
    \end{enumerate}
\end{proof}

Now we instantiate the Exclusion Lemma to handle the kind of rare events arising from the large deviation property of i.i.d. random variables. To formally state the result, define the following projectors:
    \begin{itemize}
        \item For $r\geq0$, let $\Pi_{\geq r}$ be the projector  onto those $\ket{g}$s such that $g\in([N]\cup\{\bot\})^m$ and $\abs{g}\geq r$.
        \item For $r\geq 0$, let $P_{\geq r}$ be the projector  onto those $\ket{f}$s such that $f\in[N]^m$ and $\abs{f}\geq r$; let $P_{< r}$ be the projector that projects onto those $\ket{f}$s such that $f\in[N]^m$ and $\abs{f}< r$.
    \end{itemize}\vspace{1mm}
These projectors look similar but will serve different purposes. The projector $\Pi_{\geq r}$ will be applied to the input register of the recording query model; while projectors $P_{\geq r}$ and $P_{<r}$ will be applied to the input register of the standard query model.

With the above preparation, we are now ready to state the key technical result of this part that allows us to exclude high Hamming-weight inputs (which form the ``Hamming event'') during the quantum query algorithm's execution:
\begin{lemma}[Mirroring Lemma]\label{lem:mirror}
    For every $\Gamma, \gamma \geq 0$ and $t\in \mathbb{Z}_{[0,T]}$, we have 
    \begin{align}
        \norm{(\id\otimes \Pi_{\geq \Gamma})  \ket{\phi_t}}^2 
         \leq 2\pr[\abs{x} \geq \gamma \mid x \leftarrow \calD] + 2\norm{(\id\otimes \Pi_{\geq \Gamma})   \calT (\id\otimes P_{< \gamma})}^2.
    \end{align}
    Moreover, for all integers  $\Gamma \geq \gamma \geq 0$, 
    \begin{equation}
        \norm{(\id \otimes \Pi_{\geq \Gamma})   \calT (\id \otimes P_{< \gamma})}^2 \leq \binom{m}{\Gamma}{\binom{m}{<\gamma}} p^{\Gamma - \gamma}.
    \end{equation}
\end{lemma}
\begin{proof}
        Let $\ket{\psi_t}$ be the state after $t$ queries of the quantum algorithm in the standard oracle model (with additional input register) as defined in \cref{eq:standard_model_state}. Note that there exist unit vectors $\ket{\chi_x}$s such that
    \begin{equation}
        \ket{\psi_t} = \sum_{x\in[N]^m} \sqrt{\pr[x\leftarrow \calD]}\ket{\chi_x}\ket x.
    \end{equation}

    By~\cref{thm:recording-vs-standard}, we can write $(\id\otimes \Pi_{\geq \Gamma})  \ket{\phi_t}$ as
    \begin{equation}
    \begin{aligned}
        (\id\otimes \Pi_{\geq \Gamma})  \ket{\phi_t} =& 
         (\id\otimes \Pi_{\geq \Gamma}) \calT \ket{\psi_t}.
    \end{aligned}
    \end{equation}
    Then,
    \begin{align}
        \|(\id\otimes \Pi_{\geq \Gamma})  \ket{\phi_t}\|^2
        \notag 
        =&~ \norm{(\id\otimes \Pi_{\geq \Gamma})   \calT  (\id\otimes (P_{\geq \gamma} + P_{< \gamma}))\ket{\psi_t}}^2 
        \notag \\
        \le&~ \left(\norm{(\id\otimes\Pi_{\geq \Gamma}) \calT  (\id\otimes P_{\geq \gamma})\ket{\psi_t}}
        +\norm{(\id\otimes\Pi_{\geq \Gamma}) \calT  (\id\otimes P_{< \gamma})\ket{\psi_t}}\right) ^2
        \notag \\
        \le&~ \left(\norm{ (\id\otimes P_{\geq \gamma})\ket{\psi_t}}
        +\norm{(\id\otimes\Pi_{\geq \Gamma}) \calT  (\id\otimes P_{< \gamma})\ket{\psi_t}}\right) ^2
        \notag \\
        \le&~ 2\norm{ (\id\otimes P_{\geq \gamma})\ket{\psi_t}}^2
        +2\norm{(\id\otimes\Pi_{\geq \Gamma}) \calT  (\id\otimes P_{< \gamma})\ket{\psi_t}} ^2
        \notag \\
        \leq&~ 2\pr[\abs{x} \geq \gamma \mid x\leftarrow \calD] + 2\norm{(\id\otimes \Pi_{\geq \Gamma})   \calT (\id\otimes P_{< \gamma})}^2,
        \notag
    \end{align}
    where the first inequality uses the triangle inequality, the second inequality holds as $(\id\otimes \Pi_{\geq \Gamma})\calT$ is a contraction, and the third inequality uses the Cauchy-Schwarz inequality.
    
    Now consider the ``moreover'' part. Write 
    \begin{equation}
        \calT' \coloneqq \bigotimes_{i\in [m]} \calS_i,
    \end{equation}
    so that $\calT = \id \otimes \calT'$. Then, the ``moreover'' part is equivalent to: every unit vector $\ket{v} \in \mathbb{C}^{([N]\cup\{\bot\})^m}$ satisfies
    \begin{equation}
         \norm{\Pi_{\geq \Gamma}\calT' P_{< \gamma} \ket{v}}^2 \leq \binom{m}{\Gamma} \binom{m}{<\gamma}p^{\Gamma - \gamma}.
    \end{equation}
    
    To show this, we first write
    \begin{equation}\label{eq:Pgammav}
        P_{< \gamma} \ket{v} = \sum_{S\subseteq [m], \abs{S} <\gamma} \beta_S\ket{v_S},
    \end{equation}
    for some $\beta_S \in \mathbb{C}$ satisfying $\sum_S \abs{\beta_S}^2 \leq 1$, and $\ket{v_S}$ is a normalized quantum state supported only on basis states in
    \begin{equation}\label{eq:condition_betag}
        \Bigl\{\ket{g} \bigm| g\in[N]^m \ \text{such that} \ g_i\in\Sigma_1^{(i)} \Longleftrightarrow i\in S \ \text{for all} \ i\in [m] \Bigr\}.
    \end{equation}
    
    For any subset of indices $R\subseteq [m]$, if we measure the state $\calT'\ket{v_S}$ at indices in $R$ using the logical 0-1 measurements $\{\Pi_0^{(i)}, \Pi_1^{(i)}\}_{i\in R}$, the probability of getting all-ones is equal to
    \begin{align}
        &\Bigl\|\Bigl(\bigotimes_{i\in R}\Pi_1^{(i)}\otimes \mathbb{I}\Bigr)\calT'\ket{v_S}\Bigr\|^2 \notag
        \\
        =&\bra{v_S}(\calT')^\dagger\Bigl(\bigotimes_{i\in R}\Pi_1^{(i)}\otimes \mathbb{I}\Bigr)\calT'\ket{v_S}\notag
        \\
        =&\bra{v_S}\Bigl(\bigotimes_{i\in S\cap R}\Pi_1^{(i)}\otimes \bigotimes_{i\in R\backslash S}\Pi_0^{(i)}\otimes \mathbb{I}\Bigr)^\dagger(\calT')^\dagger\Bigl(\bigotimes_{i\in R}\Pi_1^{(i)}\otimes \mathbb{I}\Bigr)\calT'\Bigl(\bigotimes_{i\in S\cap R}\Pi_1^{(i)}\otimes \bigotimes_{i\in R\backslash S}\Pi_0^{(i)}\otimes \mathbb{I}\Bigr)\ket{v_S}
        \notag
        \\
        =&\bra{v_S}\Bigl(\bigotimes_{i\in S\cap R}\Pi_1^{(i)}\calS_i^\dagger \Pi_1^{(i)}\calS_i\Pi_1^{(i)} \otimes \bigotimes_{i\in R\backslash S}\Pi_0^{(i)}\calS_i^\dagger \Pi_1^{(i)}\calS_i\otimes\Pi_0^{(i)}\otimes \mathbb{I}\Bigr)\ket{v_S}
        \notag
        \\
        \leq & \Bigl(\prod_{i\in S\cap R}\norm{\Pi_1^{(i)}\calS_i^\dagger \Pi_1^{(i)}\calS_i\Pi_1^{(i)}}\Bigr)\cdot \Bigl(\prod_{i\in R\backslash S}\norm{\Pi_0^{(i)}\calS_i^\dagger \Pi_1^{(i)}\calS_i\Pi_0^{(i)}}\Bigr) \leq \prod_{i\in R\backslash S}\norm{\Pi_0^{(i)}\calS_i^\dagger \Pi_1^{(i)}\calS_i\Pi_0^{(i)}}. \label{eq:mirror_SR_connection}
    \end{align}
    Each factor in \cref{eq:mirror_SR_connection} is exactly the leakage operator in \cref{def:leakage_op}. By \cref{lemma_leakage_bound}, we have
    \begin{equation}\label{eq:bound_fixedS}
        \Bigl\|\Bigl(\bigotimes_{i\in R}\Pi_1^{(i)}\otimes\id\Bigr)\calT'\ket{v_S}\Bigr\|^2\leq \prod_{i\in R\backslash S}p_i(1-p_i) \leq  p^{\abs{R\backslash S}}. 
    \end{equation}
    
    Now, observe the matrix inequality
    \begin{equation}\label{eq:matrix_ineq}
        \Pi_{\geq \Gamma} \leq \sum_{R \subseteq [m], \abs{R} = \Gamma} \Bigl(\bigotimes_{i\in R}\Pi_1^{(i)}\otimes \id\Bigr).
    \end{equation}
    Therefore,
    \begin{align*}
        &\norm{\Pi_{\geq \Gamma}\calT' P_{< \gamma} \ket{v}}^2\\
        =&\bra{v} P_{< \gamma}^\dagger (\calT')^\dagger \Pi_{\geq \Gamma}\calT' P_{< \gamma} \ket{v}\\
        \leq &\sum_{R \subseteq [m], \abs{R} = \Gamma} \Bigl\|\Bigl(\bigotimes_{i\in R}\Pi_1^{(i)}\otimes  \mathbb{I}\Bigr)\calT'\Bigl(\sum_{S\subseteq [m], \abs{S} <\gamma} \beta_S\ket{v_S}\Bigr)\Bigr\|^2&&\text{\cref{eq:Pgammav,eq:matrix_ineq}}
        \\
        \leq &\sum_{R \subseteq [m], \abs{R} = \Gamma} \Bigl(\sum_{S\subseteq [m], \abs{S}<\gamma}\abs{\beta_S}^2\Bigr)\Bigl(\sum_{S\subseteq [m], \abs{S}<\gamma} p^{\abs{R\backslash S}}\Bigr)&&\text{Cauchy-Schwarz and \cref{eq:bound_fixedS}}\\
        \leq&\sum_{R \subseteq [m], \abs{R} = \Gamma} \sum_{S\subseteq [m], \abs{S}<\gamma} p^{\abs{R\backslash S}}\\
        \leq&\binom{m}{\Gamma}\binom{m}{<\gamma}p^{\Gamma - \gamma},
    \end{align*}
    as required.
\end{proof}

The Mirroring Lemma is intuitive in view of~\cref{thm:recording-vs-standard} and~\cref{lemma:recording_operator_effect}: $\ket{\phi_t}$ and $\ket{\psi_t}$ are ``mirror images'' of each other up to the change of $\calT$ whose effect is controllable, and $\ket{\psi_t}$ clearly cannot record a rare event with a probability larger than that of the event happening. Now we instantiate the lemma into the following quantitative statement that is useful in our setting.

\begin{lemma}[Excluding high-degree graphs]\label{lem:cycle-degree-event-bound}
    Let $\Delta\in\mathbb{N}$. If $n\geq 2$, we have
    \begin{equation}\label{eq:exclusion_high_degree}
        \norm{\Pi^{\deg}_{\geq \Delta} \ket{\phi_t}}^2
         \leq 2n  \Bigl(\frac{2em}{n\ceil{\log(n)}}\Bigr)^{\ceil{\log(n)}}+2n\left(\frac{2em}{n\Delta}\right)^{\Delta}\Bigl(\frac{emn}{2\ceil{\log(n)}}\Bigr)^{\ceil{\log(n)}}.
    \end{equation}
\end{lemma}
\begin{proof}
    Fix $v\in [n]$. For $x\in \binom{[n]}{2}^m$, we write $|x|$ for the degree of $v$ in $\mathcal{G}(x)$ in this proof. When $\Delta\geq \ceil{\log(n)}$, applying \cref{lem:mirror} with $\Gamma = \Delta$ and $\gamma = \ceil{\log(n)}$ gives
    \begin{align}\label{eq:instantiated_combine}
        \norm{\Pi^{\deg}_{v,\geq \Delta}  \ket{\phi_t}}^2
         &\leq 2 \pr[\abs{x} \geq \gamma \mid x \leftarrow \calD] + 2\binom{m}{\Delta}\binom{m}{<\ceil{\log(n)}}\Bigl(\frac{n-1}{\binom{n}{2}}\Bigr)^{\Delta - \ceil{\log(n)}}
         \\
         &\leq 2 \pr[\abs{x} \geq \gamma \mid x \leftarrow \calD]+2\left(\frac{em}{\Delta}\right)^{\Delta}\Bigl(\frac{em}{\ceil{\log(n)}}\Bigr)^{\ceil{\log(n)}} \cdot \Bigl(\frac{2}{n}\Bigr)^{\Delta - \ceil{\log(n)}} 
         \notag
         \\
         &\leq 2  \Bigl(\frac{2em}{n\ceil{\log(n)}}\Bigr)^{\ceil{\log(n)}}+2\left(\frac{2em}{n\Delta}\right)^{\Delta}\Bigl(\frac{emn}{2\ceil{\log(n)}}\Bigr)^{\ceil{\log(n)}},
         \notag
    \end{align}
    where the last step holds due to the Chernoff bound in \cref{lemma_chernoff}.
    
    Now unfixing $v$ and summing over all $v\in[n]$, we obtain
    \begin{equation}
         \norm{\Pi^{\deg}_{\geq \Delta}  \ket{\phi_t}}^2 \leq \sum_{v\in [n]}  \norm{\Pi^{\deg}_{v,\geq \Delta}  \ket{\phi_t}}^2 \leq 2n \Bigl(\frac{2em}{n\ceil{\log(n)}}\Bigr)^{\ceil{\log(n)}}+2n\left(\frac{2em}{n\Delta}\right)^{\Delta}\Bigl(\frac{emn}{2\ceil{\log(n)}}\Bigr)^{\ceil{\log(n)}},
    \end{equation}
    as required.

    When $\Delta < \ceil{\log(n)}$, the bound also holds as the second term on the right hand side of \cref{eq:instantiated_combine} is strictly greater than 1, while the left hand side is less than or equal to 1.
\end{proof}

\paragraph{Excluding high-degree graphs.}
Recall that $\ket{\phi_t}$ depends on the $m$ and $n$ that specify the distribution $\calD$. When $m=n$, the previous lemma implies that only a tiny fraction of basis states $\ket{i,u,w}\ket{x}$ in the support $\ket{\phi_t}$ can have $x$ recording a high-degree vertex, \emph{irrespective} of the value of $t$. More formally, we have the following corollary.
\begin{corollary}
    \label{cor_bad_event_bound}
    Suppose $m = n$ is sufficiently large. Then for all $d\in \mathbb{Z}$ such that $d\geq \ceil{\log(n)}^2$, and for all $t\in \mathbb{Z}_{[0,T]}$, we have
    \begin{equation}
        \norm{\Pi_{\geq d}^{\deg}\ket{\phi_t}} \leq \frac{1}{n^2}.
    \end{equation}
    In particular,
    \begin{equation}
        \sum_{i=1}^{t}\norm{\Pi_{\Bad}\ket{\phi_i}} \leq \frac{t}{n^2}.
    \end{equation}
\end{corollary}
\begin{proof}
    By \cref{lem:cycle-degree-event-bound} (with its $\Delta$ set to $\ceil{\log(n)}^2$),
    \begin{equation}
    \begin{aligned}
        \norm{\Pi_{\geq d}^{\deg}\ket{\phi_t}}\leq& 2n\Bigl(\frac{2e}{\ceil{\log(n)}}\Bigr)^{\ceil{\log(n)}} + 2n\left(\frac{2e}{\ceil{\log(n)}^2}\right)^{\ceil{\log(n)}^2}\Bigl(\frac{en^2}{2\ceil{\log(n)}}\Bigr)^{\ceil{\log(n)}}\\
        =& 2n\Bigl(\frac{2e}{\ceil{\log(n)}}\Bigr)^{\ceil{\log(n)}} + 2n\Bigl(\frac{2^{\ceil{\log(n)}}e^{\ceil{\log(n)}+1}n^2}{2\ceil{\log(n)}^{2\ceil{\log(n)} + 1}}\Bigr)^{\ceil{\log(n)}}\leq\frac{1}{n^2},
    \end{aligned}
    \end{equation}
    where the first inequality uses $m=n$, and second inequality uses $\ceil{\log(n)}^{\ceil{\log(n)}} = \omega(n^k)$ for any constant $k$. The ``in particular'' part follows immediately from recalling that $\Pi_{\Bad}$ is defined to be $\Pi_{\geq \ceil{\log(n)}^2}^{\deg}$.
\end{proof}

\subsubsection{Progress in recording wedges}
To reliably record a triangle, there must be many wedges recorded along the way. In this part, we bound the progress of a quantum query algorithm in recording wedges. Recall the measure of progress of recording at least $r$ wedges in $t$ queries is defined as
\begin{equation}
    \Lambda_{t,r} \coloneqq \norm{\Pi^{\trianglewedge}_{[r,\infty)}\ket{\phi_t}}.
\end{equation}

As mentioned previously, high-degree vertices can potentially cause an overestimate of the power of quantum algorithms to record wedges. We apply the Exclusion Lemma to suppress these high-degree vertices.
In particular, invoking~\cref{lem:exclusion_lemma} with $\Pi_i\coloneqq 0$, $\Pi_i'\coloneqq \Pi_{\Bad}$ for all $i\in [t-1]$ and $\Pi_{\record} \coloneqq \Pi^{\trianglewedge}_{[r,\infty)}$  gives
\begin{equation}\label{eq:vtr_good_bad}
    \Lambda_{t,r} \leq \norm{\Pi^{\trianglewedge}_{[r,\infty)}\ket{\phi'_t}} + \sum_{i=1}^{t-1}\norm{\Pi_{\Bad}\ket{\phi_i}},
\end{equation}
where $\ket{\phi'_t}$ is the (possibly unnormalized) state defined by
\begin{equation}
    \ket{\phi'_t} \coloneqq U_t\calR(\id-\Pi_{t-1}')U_{t-1}\calR(\id-\Pi_{t-2}')U_{t-2}\calR\cdots (\id-\Pi_1')U_1\calR U_0(\ket{0}\ket{\bot^m}).
\end{equation}
The term  $\sum_{i=1}^{t-1}\norm{\Pi_{\Bad}\ket{\phi_i}}$ in \cref{eq:vtr_good_bad} can be bounded using \cref{cor_bad_event_bound}. Therefore, to bound $\Lambda_{t,r}$, it suffices to bound
\begin{equation}
    \Lambda'_{t,r} \coloneqq \norm{\Pi^{\trianglewedge}_{[r,\infty)}\ket{\phi'_t}}.
\end{equation}

\begin{lemma}[Wedges progress recurrence]
    \label{lemma_v_prime_recurrence}
    For all $t\in \mathbb{Z}_{[0,T]}$ and $r\geq 0$, we have
    \begin{equation}
        \Lambda'_{t+1,r} \leq \Lambda'_{t,r} + 8\sqrt{\frac{t}{n}}\Lambda'_{t,r-2\ceil{\log(n)}^2}.\label{eq:v-prime-recurrence}
    \end{equation}
    In addition, $\Lambda'_{0,0} = 1$, and $\Lambda'_{0,r} =0$ for all $r\geq 1$.
\end{lemma}
\begin{proof}
    It is clear that the boundary condition $\Lambda'_{0,0} = 1$, and $\Lambda'_{0,r} =0$ for all $r\geq 1$ since $\ket{\phi_0'}=\ket{\phi_0}=\ket{0}\ket{\bot^m}$. So in the remainder of the proof, we focus on deriving the recurrence~\cref{eq:v-prime-recurrence}.

    For integer $t$ with $t+1\in\mathbb Z_{[0,T]}$, since $U_{t+1}$ acts as identity on the input register,
    \begin{equation}
        \Lambda'_{t+1,r} = \norm{\Pi^{\trianglewedge}_{[r,\infty)}U_{t+1}\calR(\id-\Pi_{\Bad})\ket{\phi'_t}}=\norm{\Pi^{\trianglewedge}_{[r,\infty)}\calR(\id-\Pi_{\Bad})\ket{\phi'_t}}.
    \end{equation}

    Consider a basis state $\ket{i,u,w}\ket{x}$. If $x$ does not record at least $r$ wedges, then $x$ must record at least $r-2\maxdeg(\calG(x))$ wedges in order for $\Pi^{\trianglewedge}_{[r,\infty)}\calR \ket{i,u,w}\ket{x} \neq 0$. Recall $\maxdeg(\calG(x))$ denotes the maximum vertex degree of $\calG(x)$. This is because $\calR$ can introduce at most one new edge in register $I_i$, and that new edge can introduce at most $2\maxdeg(\calG(x))$ new wedges. (This bound is saturated when the new edge connects two disconnected vertices in $\calG(x)$ of degree $\maxdeg(\calG(x))$ each.) If $\ket{i,u,w}\ket{x}$ is in the support of $(\id-\Pi_{\Bad})\ket{\phi_t'}$, then $\maxdeg(\calG(x)) < \ceil{\log(n)}^2$ by the definition of $\Pi_{\Bad}$. Therefore,
    \begin{align}
        \Lambda'_{t+1,r}~\leq~ &\norm{\Pi^{\trianglewedge}_{[r,\infty)}\calR(\id-\Pi_{\Bad})\Pi^{\trianglewedge}_{[r,\infty)}\ket{\phi'_t}} + \norm{\Pi^{\trianglewedge}_{[r,\infty)}\calR(\id-\Pi_{\Bad})\Pi^{\trianglewedge}_{[r-2\ceil{\log(n)}^2,r)}\ket{\phi'_t}}
        \notag 
        \\     
        ~=~&\norm{\Pi^{\trianglewedge}_{[r,\infty)}\calR(\id-\Pi_{\Bad})\Pi^{\trianglewedge}_{[r,\infty)}\ket{\phi'_t}} 
        +\norm{\Pi^{\trianglewedge}_{[r,\infty)}\calR (\I -\ketbra{0}{0}_{P})(\id-\Pi_{\Bad})\Pi^{\trianglewedge}_{[r-2\ceil{\log(n)}^2,r)}\ket{\phi'_t}}   
        \nonumber\\ ~\leq~&\norm{\Pi^{\trianglewedge}_{[r,\infty)}\ket{\phi'_t}} + \sum_{y\in [N]\cup\{\bot\}}\norm{\Pi^{\trianglewedge}_{[r,\infty)}\calR\Pi_y(\id-\Pi_{\Bad})\Pi^{\trianglewedge}_{[r-2\ceil{\log(n)}^2,r)}\ket{\phi'_t}}
        \nonumber\\        
        ~=~&\Lambda'_{t,r} + \sum_{y\in [N]\cup\{\bot\}}\norm{\Pi^{\trianglewedge}_{[r,\infty)}\calR\Pi_y(\id-\Pi_{\Bad})\Pi^{\trianglewedge}_{[r-2\ceil{\log(n)}^2,r)}\ket{\phi'_t}}.\label{eq:wedge_recurrence_terms}       
    \end{align}
    To go from the second term in the first line to that in the second line, we used the fact that $\calR$ acts as identity on each basis state $\ket{i,u,w}\ket{x}$ when $u=0$, so no more wedges will be recorded. In the third step, note that $\I_{P}-\ketbrasame{0}_{P} = \sum_{y\in[N]\cup \{\bot\}} \Pi_y$, then it follows from triangle inequality. 

    Abbreviate
    \[
        \ket\rho\coloneqq(\id-\Pi_{\Bad})\Pi^{\trianglewedge}_{[r-2\ceil{\log(n)}^2,r)}\ket{\phi'_t},
    \]
    we bound the second term in~\cref{eq:wedge_recurrence_terms} based on whether $y=\bot$ or not. We will show
    \begin{align}
        \norm{\Pi^{\trianglewedge}_{[r,\infty)}\calR\Pi_\bot\ket\rho} &\le 2\sqrt{\frac{t}{n}}\Lambda'_{t,r-2\ceil{\log(n)}^2},
        \label{eq:wedge-y-bot}
        \\
        \sum_{y\in[N]}   \norm{\Pi^{\trianglewedge}_{[r,\infty)}\calR\Pi_y\ket\rho} &\le 6\sqrt{\frac{t}{n}}\Lambda'_{t,r-2\ceil{\log(n)}^2}.
        \label{eq:wedge-y-vertex}
    \end{align}
    Plugging~\cref{eq:wedge-y-bot,eq:wedge-y-vertex} into~\cref{eq:wedge_recurrence_terms}, we obtain our main recurrence~\cref{eq:v-prime-recurrence} stated in the lemma, concluding the proof.

    \paragraph{Bound~\cref{eq:wedge-y-bot} for $y=\bot$.} 
        When $y=\bot$, consider each basis state $\ket{i,u,w}\ket{x}$ in the support of
        \begin{equation}
            \Pi_\bot(\id-\Pi_{\Bad})\Pi^{\trianglewedge}_{[r-2\ceil{\log(n)}^2,r)}\ket{\phi'_t}.
        \end{equation}
         Due to $\Pi_\bot$, we have $x_i = \bot$ and $u\neq 0$, so \cref{lemma:recording_operator_effect} states
        \begin{equation}
            \calR\ket{i,u,w}\ket{x} = \ket{i,u,w}\biggl(\sum_{y\in [N]}\frac{\omega_N^{u y}}{\sqrt{N}}\ket{y}_{I_i}\biggr)\otimes\bigotimes_{j\neq i}\ket{x_{j}}_{I_{j}}.
        \end{equation}
        Since $\ket{i,u,w}\ket{x}$ is also in the support of $\ket{\phi_t'}$, $x$ records at most $t$ edges by \cref{fact:t_non_bots}. Due to $\Pi^{\trianglewedge}_{[r-2\ceil{\log(n)}^2,r)}$, $x$ records fewer than $r$ wedges.  Therefore, in order for $\ket{y}_{I_i} \otimes \bigotimes_{j\neq i}\ket{x_{j}}_{I_{j}}$ to record at least $r$ wedges, $y$ must be an edge adjacent to some edge recorded in $x$, of which there are at most $t \cdot 2(n-1)$ possibilities. Therefore,
        \begin{equation}
        \norm{\Pi^{\trianglewedge}_{[r,\infty)}\calR\ket{i,u,w}\ket{x}}\leq\sqrt{\frac{t \cdot 2(n-1)}{N}}=2\sqrt{\frac{t}{n}}.
        \end{equation}
        Since any two distinct basis states in the support of $\Pi_\bot(\id-\Pi_{\Bad})\Pi^{\trianglewedge}_{[r-2\ceil{\log(n)}^2,r)}\ket{\phi'_t}$ remain orthogonal after $\Pi^{\trianglewedge}_{[r,\infty)}\calR$ is applied,\footnote{Suppose basis state $\ket{i',u',w'}\ket{x'}$ is orthogonal to $\ket{i,u,w}\ket{x}$. If $(i',u',w')\neq (i,u,w)$, then it is clear that the two states remain orthogonal after applying $\Pi^{\trianglewedge}_{[r,\infty)}\calR$ since the operator can only act non-trivially on the input register. If $(i',u',w') = (i,u,w)$ and $x_j\neq x'_{j}$ for some $j\neq i$, then the two states still remain orthogonal since  $\calR$ does not act on the $I_j$ register and $\Pi^{\trianglewedge}_{[r,\infty)}$ is diagonal in the $\{\ket{z} \mid z\in ([N]\cup \{\bot\})^m\}$ basis. The case $(i',u',w') = (i,u,w)$ and $x_i\neq x'_i$ is forbidden by the assumption that the basis states lie in the support of $\Pi_\perp$. We will later reuse similar arguments without comment.} we have
        \begin{equation}
        \begin{aligned}            
            &\norm{\Pi^{\trianglewedge}_{[r,\infty)}\calR\Pi_\bot(\id-\Pi_{\Bad})\Pi^{\trianglewedge}_{[r-2\ceil{\log(n)}^2,r)}\ket{\phi'_t}}
            \\
            &\qquad \le 2\sqrt{\frac{t}{n}}\norm{\Pi_\bot(\id-\Pi_{\Bad})\Pi^{\trianglewedge}_{[r-2\ceil{\log(n)}^2,r)}\ket{\phi'_t}} \leq 2\sqrt{\frac{t}{n}}\Lambda'_{t,r-2\ceil{\log(n)}^2}.
        \end{aligned}
        \end{equation}
        \paragraph{Bound~\cref{eq:wedge-y-vertex} for $y\in [N]$.}
        When $y\in [N]$, consider each basis state $\ket{i,u,w}\ket{x}$ with non-zero amplitude in
        \begin{equation}
            \Pi_y(\id-\Pi_{\Bad})\Pi^{\trianglewedge}_{[r-2\ceil{\log(n)}^2,r)}\ket{\phi'_t}.
        \end{equation}
        Due to $\Pi_y$, we have $x_i = y$ and $u\neq 0$, so \cref{lemma:recording_operator_effect} states
        \begin{equation}\label{eq:wedge_recurrence_nonperp}
        \begin{aligned}
            \calR\ket{i,u,w}\ket{x}
            =\ket{i,u,w}&\biggl(\frac{\omega^{u y}_N}{\sqrt{N}}\ket{\bot}_{I_i}+\frac{1+\omega^{u y}_N(N-2)}{N}\ket{y}_{I_i}\\
            &\qquad\qquad+\sum_{y'\in [N]\backslash\{y\}}\frac{1-\omega^{u y'}_N-\omega^{u y}_N}{N}\ket{y'}_{I_i}\biggr)\otimes\bigotimes_{j\neq i}\ket{x_{j}}_{I_{j}}.
        \end{aligned}
        \end{equation}
        Applying similar reasoning as in the previous case, we deduce
        \begin{equation}
            \norm{\Pi^{\trianglewedge}_{[r,\infty)}\calR\ket{i,u,w}\ket{x}}\leq 3\frac{\sqrt{t\cdot 2(n-1)}}{N}.
        \end{equation}
        Since any two distinct basis states in the support of $\Pi_y(\id-\Pi_{\Bad})\Pi^{\trianglewedge}_{[r-2\ceil{\log(n)}^2,r)}\ket{\phi'_t}$ remain orthogonal after $\Pi^{\trianglewedge}_{[r,\infty)}\calR$ is applied, we have
        \begin{equation}
        \begin{aligned}
            &\norm{\Pi^{\trianglewedge}_{[r,\infty)}\calR\Pi_y(\id-\Pi_{\Bad})\Pi^{\trianglewedge}_{[r-2\ceil{\log(n)}^2,r)}\ket{\phi'_t}} 
            \\
            &\qquad\leq 3\frac{\sqrt{t\cdot 2(n-1)}}{N}\norm{\Pi_y(\id-\Pi_{\Bad})\Pi^{\trianglewedge}_{[r-2\ceil{\log(n)}^2,r)}\ket{\phi'_t}}.
        \end{aligned}
        \end{equation}
        Therefore, applying the Cauchy-Schwarz inequality gives
        \begin{align}
            &\sum_{y\in [N]}\norm{\Pi^{\trianglewedge}_{[r,\infty)}\calR\Pi_y(\id-\Pi_{\Bad})\Pi^{\trianglewedge}_{[r-2\ceil{\log(n)}^2,r)}\ket{\phi'_t}}
            \nonumber\\
            &\qquad \le 3\sqrt{\frac{t\cdot 2(n-1)}{N}}\sqrt{\sum_{y\in [N]}\norm{\Pi_y(\id-\Pi_{\Bad})\Pi^{\trianglewedge}_{[r-2\ceil{\log(n)}^2,r)}\ket{\phi'_t}}^2} 
            \nonumber\\
            &\qquad \leq  6\sqrt{\frac{t}{n}}\Lambda'_{t,r-2\ceil{\log(n)}^2}.\qedhere
        \end{align}
\end{proof}

Solving the wedges progress recurrence, we obtain
\begin{lemma}[Wedges progress]\label{lemma_v_prime_bound}
    For all $t \in \mathbb{Z}_{[0,T]}$ and $r \geq 2\ceil{\log(n)}^2$, we have
    \begin{equation}
        \Lambda'_{t,r} \leq \biggl(\frac{8et^{3/2}}{\lfloor \frac{r}{2\ceil{\log(n)}^2}\rfloor\sqrt{n}}\biggr)^{\lfloor \frac{r}{2\ceil{\log(n)}^2}\rfloor}.
    \end{equation}
\end{lemma}
\begin{proof}
    Let $k \coloneqq \lfloor \frac{r}{2\ceil{\log(n)}^2}\rfloor \geq 1$. By definition, $\Lambda'_{t,r} \leq \Lambda'_{t,k\cdot 2\ceil{\log(n)}^2}$. Therefore,
    \begin{equation}
        \Lambda'_{t,r}\leq \Lambda'_{t,k\cdot 2\ceil{\log(n)}^2} \leq \binom{t}{k}\biggl(8\sqrt{\frac{t}{n}}\biggr)^{k} \leq \biggl(\frac{8et^{3/2}}{k\sqrt{n}}\biggr)^{k},
    \end{equation}
    where we solved the recurrence from \cref{lemma_v_prime_recurrence} in the second inequality.
\end{proof}

\subsubsection{Progress in recording a triangle}
We move on to bounding the progress of a quantum query algorithm in recording a triangle. Recall that we defined the progress measure of finding a triangle within $t$ queries
\begin{equation}
       \Delta_t\coloneqq \norm{\Pi^\triangle\ket{\phi_t}}.
\end{equation}

Let 
\begin{equation}
    r^* \coloneqq \biggl\lceil \frac{n^{4/7}}{\log^{6/7}(n)}\biggr\rceil \cdot 2\ceil{\log(n)}^2.
\end{equation}
and for $t\in \mathbb{Z}_{[0,T]}$, define
\begin{equation}\label{eq:phit_star}
    \ket{\phi_t^*}\coloneqq U_t\calR(\id-\Pi^*)U_{t-1}\calR(\id-\Pi^*)U_{t-2}\calR\cdots (\id-\Pi^*)U_1\calR U_0(\ket{0}\ket{\bot^m}),
\end{equation}
where $\Pi^*$ denotes the projector onto every basis state $\ket{i,u,w}\ket{x}$ such that $\calG(x)$ contains at least $r^*$ wedges. Then, define
\begin{equation}\label{eq:triangleprogress_star}
    \Delta_t^*\coloneqq \norm{\Pi^\triangle\ket{\phi_t^*}}.
\end{equation}

We can bound $\Delta_t$ by bounding $\Delta_t^*$ due to the next lemma.
\begin{lemma}\label{lem:delta_deltastar}
    For all sufficiently large $n$ and $T\leq o(n^{5/7}/\log^{4/7}(n))$, we have
    \begin{equation}
        \Delta_T \leq \Delta_T^* + o(1).
    \end{equation}
\end{lemma}
\begin{proof}
    We have
    \begin{align*}
        \norm{\ket{\phi_T^*} - \ket{\phi_T}} \leq& \sum_{t=1}^T \norm{\Pi^*\ket{\phi_t}}  &&\text{\cref{lem:exclusion_lemma}}
        \\
        =& \sum_{t=1}^T\Lambda_{t,r^*} &&\text{by definition}
        \\
        \leq&\sum_{t=1}^T \biggl(\Lambda_{t,r^*}' + \sum_{i=1}^{t-1}\norm{\Pi_{\Bad}\ket{\phi_i}}\biggr) &&\text{\cref{eq:vtr_good_bad}}
        \\
        \leq& \sum_{t=1}^T \biggl(\Lambda_{t,r^*}' + \frac{t}{n^2}\biggr) &&\text{\cref{cor_bad_event_bound}}
        \\
        \leq& (T/n)^2 + \sum_{t=1}^T \Lambda_{t,r^*}'.
    \end{align*}
    But if $t\in \mathbb{N}$ is such that $t\leq T \leq o(n^{5/7}/\log^{4/7}(n))$, then \cref{lemma_v_prime_bound} gives
    \begin{equation}
        \Lambda'_{t,r^*} \leq  \biggl(\frac{8et^{3/2}}{\lfloor \frac{r^*}{2\ceil{\log(n)}^2}\rfloor\sqrt{n}}\biggr)^{\lfloor \frac{r^*}{2\ceil{\log(n)}^2}\rfloor} \leq 2^{-{\lfloor \frac{r^*}{2\ceil{\log(n)}^2}\rfloor}} \leq 2^{-\sqrt{n}}
    \end{equation}
    for sufficiently large $n$.

    Therefore, $\norm{\ket{\phi_T^*} - \ket{\phi_T}} \leq (T/n)^2 + T/2^{\sqrt{n}}$, which is at most $o(1)$ for $T\leq o(n^{5/7}/\log^{4/7}(n))$.
\end{proof}

\begin{lemma}[Triangle progress recurrence]\label{lemma_Delta_upper_bound_improved}
   For all $t\in \mathbb{Z}_{[0,T]}$,
   \begin{equation}\label{eq:triangle_recurrence}
       \Delta_{t+1}^* \leq \Delta_t^* + 4\sqrt{\frac{r^*}{N}},
   \end{equation}
   and $\Delta_0^* = 0$
\end{lemma}
\begin{proof}[Proof of \cref{lemma_Delta_upper_bound_improved}]
    The boundary condition $\Delta_0^* = 0$ in \cref{eq:triangle_recurrence} is immediate since $\ket{\phi_0^*} = U_0 (\ket{0}\ket{\bot^m})$ does not record any wedges. To prove the recurrence, observe that
    \begin{align}
        \Delta^*_{t+1} =&~ \norm{\Pi^{\triangle} \ket{\phi_{t+1}^*}} &&\text{definition \cref{eq:triangleprogress_star}}
        \notag
        \\
        =&~ \norm{\Pi^{\triangle} U_{t+1} \calR \bigl(\I - \Pi^*\bigr)\ket{\phi_t^{*}}} &&\text{definition \cref{eq:phit_star}}
        \notag
        \\
        =&~\norm{\Pi^{\triangle} \calR \bigl(\I - \Pi^*\bigr)\ket{\phi_t^{*}}}  &&\text{$\Pi^{\triangle}$, $U_{t+1}$ commute}
        \notag
        \\
        \leq&~ \norm{\Pi^{\triangle} \ket{\phi_t^{*}}} + \norm{\Pi^{\triangle} \calR \bigl(\I - \Pi^* \bigr)\bigl(\I-\Pi^{\triangle}\bigr)\ket{\phi_t^{*}}}  &&\text{triangle inequality}
        \notag
        \\
        =&~  \Delta^*_{t} + \norm{\Pi^{\triangle} \calR \bigl(\I - \Pi^*\bigr)\bigl(\I-\Pi^{\triangle}\bigr)\ket{\phi_t^{*}}} &&\text{definition \cref{eq:triangleprogress_star}}\label{eq:triangle_recurrence_terms}
    \end{align}

    Abbreviate
    \begin{equation}
        \ket{\rho} \coloneqq \bigl(\I - \Pi^{*} \bigr)\bigl(\I-\Pi^{\triangle}\bigr)\ket{\phi_t^{*}}.
    \end{equation}
    Then, we can write the second term in \cref{eq:triangle_recurrence_terms} as
    \begin{equation}
        \norm{\Pi^{\triangle} \calR \ket{\rho}} \leq \sum_{y\in [N]\cup \{\bot\}} \norm{\Pi^{\triangle} \calR \Pi_y \ket{\rho}},
    \end{equation}
    using $\Pi^{\triangle} \calR \ketbrasame{0}_P \ket{\rho} = 0$ because $\ket{\rho}$ has recorded no triangle and $\calR\ketbrasame{0}_P = \ketbrasame{0}_P$.
    
    We will show \begin{equation}\label{eq:triangle_progress_bound}
    \norm{\Pi^{\triangle} \calR \Pi_{\bot}\ket{\rho}} \leq \sqrt{\frac{r^*}{N}} \quad \text{and} \quad
    \sum_{y\in [N]}\norm{\Pi^{\triangle} \calR \Pi_y \ket{\rho}} \leq 3\sqrt{\frac{r^*}{N}}.
    \end{equation}

    \paragraph{First bound in \cref{eq:triangle_progress_bound}.} 
        When $y=\bot$, consider each basis state $\ket{i,u,w}\ket{x}$ in the support of
        \begin{equation}
            \Pi_\bot \bigl(\I - \Pi^*\bigr)\bigl(\I-\Pi^{\triangle}\bigr)\ket{\phi_t^{*}}. 
        \end{equation}
         Due to $\Pi_\bot$, we have $x_i = \bot$ and $u\neq 0$, so \cref{lemma:recording_operator_effect}  states
        \begin{equation}
            \calR\ket{i,u,w}\ket{x} = \ket{i,u,w}\biggl(\sum_{y\in [N]}\frac{\omega_N^{u y}}{\sqrt{N}}\ket{y}_{I_i}\biggr)\otimes\bigotimes_{j\neq i}\ket{x_{j}}_{I_{j}}.
        \end{equation}
        Since $\ket{i,u,w}\ket{x}$ is also in the support of $(\I-\Pi^{*})$, $x$ records at most $r^*$ wedges. Due to $\bigl(\I-\Pi^{\triangle}\bigr)$, $x$ does not record a triangle. Therefore, in order for $\ket{y}_{I_i} \otimes \bigotimes_{j\neq i}\ket{x_{j}}_{I_{j}}$ to record a triangle, edge $y$ must connect the two ends of a wedge. The number of such $y$s is at most $r^*$. Therefore,
        \begin{equation}
        \norm{\Pi^{\triangle}\calR\ket{i,u,w}\ket{x}}\leq\sqrt{\frac{r^*}{N}}.
        \end{equation}
        Since any two distinct basis states in the support of $\Pi_\bot \bigl(\I - \Pi^{*}\bigr)\bigl(\I-\Pi^{\triangle}\bigr)\ket{\phi_t^{*}}$ remain orthogonal after $\Pi^{\triangle}\calR$ is applied, we have
        \begin{equation}
        \begin{aligned}            
             \norm{\Pi^{\triangle} \calR \Pi_{\bot}\ket{\rho}} &= \norm{\Pi^{\triangle} \calR \Pi_{\bot}\bigl(\I-\Pi^{*}\bigr)\bigl(\I-\Pi^{\triangle}\bigr)\ket{\phi_t^{*}}} \leq\sqrt{\frac{r^*}{N}}. 
        \end{aligned}
        \end{equation}
        
        \paragraph{Second bound in \cref{eq:triangle_progress_bound}.} 

        When $y\in [N]$, consider each basis state $\ket{i,u,w}\ket{x}$ in the support of
        \begin{equation}
            \Pi_y \bigl(\I - \Pi^{*}\bigr)\bigl(\I-\Pi^{\triangle}\bigr)\ket{\phi_t^{*}}. 
        \end{equation}
        Due to $\Pi_y$, we have $x_i = y$ and $u\neq 0$, so \cref{lemma:recording_operator_effect} states
        \begin{equation}
        \begin{aligned}
            \calR\ket{i,u,w}\ket{x}
            =\ket{i,u,w}&\biggl(\frac{\omega^{u y}_N}{\sqrt{N}}\ket{\bot}_{I_i}+\frac{1+\omega^{u y}_N(N-2)}{N}\ket{y}_{I_i}\\
            &\qquad\qquad+\sum_{y'\in [N]\backslash\{y\}}\frac{1-\omega^{u y'}_N-\omega^{u y}_N}{N}\ket{y'}_{I_i}\biggr)\otimes\bigotimes_{j\neq i}\ket{x_{j}}_{I_{j}}.
        \end{aligned}
        \end{equation}
        Applying similar reasoning as in the proof of the first bound in \cref{eq:triangle_progress_bound}, we deduce
        \begin{equation}
            \norm{\Pi^{\triangle}\calR\ket{i,u,w}\ket{x}}\leq 3\frac{\sqrt{r^*}}{N}.
        \end{equation}
        Since any two distinct basis states in the support of $\Pi_y \bigl(\I - \Pi^*\bigr)\bigl(\I-\Pi^{\triangle}\bigr)\ket{\phi_t^{*}}$ remain orthogonal after $\Pi^{\triangle}\calR$ is applied, we have
        \begin{equation}
            \norm{\Pi^{\triangle} \calR \Pi_{y}\bigl(\I-\Pi^*\bigr)\bigl(1-\Pi^{\triangle}\bigr)\ket{\phi_t^{*}}} \leq 3\frac{\sqrt{r^*}}{N}\norm{\Pi_{y}\bigl(\I-\Pi^*\bigr)\bigl(1-\Pi^{\triangle}\bigr)\ket{\phi_t^*}}.
        \end{equation}
        Therefore, the Cauchy-Schwarz inequality gives
        \begin{align}
            &\sum_{y\in [N]} \norm{\Pi^{\triangle} \calR \Pi_{y}\bigl(\I-\Pi^*\bigr)\bigl(1-\Pi^{\triangle}\bigr)\ket{\phi_t^{*}}} 
            \nonumber\\
            &\qquad \le 3\sqrt{\frac{r^*}{N}}\sqrt{\sum_{y\in [N]}\norm{\Pi_{y}\bigl(\I-\Pi^*\bigr)\bigl(1-\Pi^{\triangle}\bigr)\ket{\phi_t^{*}}}^2}
            \leq 3\sqrt{\frac{r^*}{N}}.\qedhere
        \end{align}
\end{proof}

\begin{lemma}[Hardness of recording a triangle]\label{lem:recording_probability} For all sufficiently large $n$, $T\leq o(n^{5/7}/\log^{4/7}(n))$ and $m = n$, we have
\begin{equation}
    \Delta_T \leq o(1) .
\end{equation}
\end{lemma}
\begin{proof}
Solving the recurrence in \cref{lemma_Delta_upper_bound_improved} gives
\begin{equation}
    \Delta_T^* \leq 4T \sqrt{\frac{r^*}{N}} \leq O\biggl(T\sqrt{\frac{n^{4/7} \log^{8/7}(n)}{n^2}}\biggr) \leq O\biggl(T\frac{\log^{4/7}(n)}{n^{5/7}}\biggr),
\end{equation}
where we recalled the definition of $r^*$ in the second inequality.

Therefore, $\Delta_T^*\leq o(1)$ if $T\leq o(n^{5/7}/\log^{4/7}(n))$. But  $\Delta_T \leq \Delta_T^* + o(1)$ if $T\leq o(n^{5/7}/\log^{4/7}(n))$ by \cref{lem:delta_deltastar}. Hence the lemma.
\end{proof}

\subsubsection{Completing the proof of \texorpdfstring{\cref{prop:triangle_search}}{Theorem 5.5}}
To finish the proof, we use the next lemma, which bounds the probability of the algorithm succeeding even if it does \emph{not} record a triangle. Intuitively, if the algorithm does not record a triangle, it should not be significantly better at finding a triangle than  guessing at random. The lemma is standard in the recording query framework --- indeed, our proof is based on \cite[Proof of Proposition 4.4]{Hamoudi_2023} --- but we include it for completeness.

Let $\Pi_{\success}$ denote the projector onto basis states $\ket{i,u,w}\ket{x}$ such that $w$ contains an output substring (i.e., a substring located at some fixed output register) of the form $(a_1,a_2,a_3,b_1,b_2,b_3) \in [m]^3 \times \binom{[n]}{2}^3$ where the $a_i$s are distinct, the $b_i$s form a triangle, and   $x_{a_i} = b_i$ for all $i$.
\begin{lemma}[Hardness of guessing a triangle]\label{lem:guessing_probability}
    For every state $\ket{\phi}$, $\norm{\Pi_{\success}\calT (\id-\Pi^\triangle)\ket{\phi}} \leq O(\frac{1}{n})$. 
\end{lemma}

With the above lemma, \cref{prop:triangle_search} becomes a simple corollary. If $T = o(n^{5/7}/\log^{4/7}(n))$, we have
\begin{equation}\label{eq:final_step}
    \norm{\Pi_{\success}\ket{\psi_T}} =  \norm{\Pi_{\success}\calT\ket{\phi_T}}\leq \norm{\Pi^\triangle \ket{\phi_T}} + \norm{\Pi_{\success}\calT(\id-\Pi^\triangle)\ket{\phi_T}} = o(1), 
\end{equation}
where we used \cref{lem:recording_probability} to bound the first term and \cref{lem:guessing_probability} to bound the second term.

\begin{proof}[Proof of~\cref{lem:guessing_probability}]
For $k,l\in \mathbb{Z}_{\geq 0}$ with $k+l\leq 3$, we define the projector $P_{k,l}$ to be onto basis states $\ket{i,u,w}\ket{x}$, where
\begin{enumerate}
    \item $w$ contains the output substring $(a_1,a_2,a_3,b_1,b_2,b_3) \in [m]^3 \times \binom{[n]}{2}^3$ where the $a_i$s are distinct and the $b_i$s form a triangle;
    \item there are exactly $k$ indices $i\in [3]$ such that $x_{a_i} = \bot$;
    \item there are exactly $l$ indices $i \in [3]$ such that $x_{a_i} \neq \bot$ and $x_{a_i} \neq b_i$.
\end{enumerate}
Observe that if $\ket{i,u,w}\ket{x}$ is in the support of $\id-\Pi^\triangle$, then $P_{0,0} \ket{i,u,w}{\ket{x}} = 0$. This is because if $\ket{i,u,w}\ket{x}$ is in the support of $P_{0,0}$, then $x_{a_i} = b_i$ for all $i$ and the $b_i$s form a triangle, which contradicts the assumption.

For a state $\ket{i,u,w}\ket{x}$ in the support of $P_{k,l}$, we have 
\begin{equation}\label{eq:succ_single_basis_kl}
    \norm{\Pi_{\success} \calT \ket{i,u,w}\ket{x}} \leq (1/\sqrt{N})^k(1/N)^l,
\end{equation}
using the definition of $\calT$.

For $\ket{i,u,w}\ket{x}$ in the support of $P_{k,l}$, we write $w_{a}\coloneqq \{a_1,a_2,a_3\}$ for the set containing the first three elements of the output substring. Observe that $\ket{i,u,w}\ket{x}$ and $\ket{i',u',w'}\ket{x'}$ remain orthogonal after applying $\Pi_{\success}\calT$ unless $i=i',u=u',w=w'$ and $x_s = x'_s$ for all $s\in [m]-\{a_1,a_2,a_3\}$.

Therefore, for a state $\ket{\chi} \coloneqq \sum_{i,u,w,x}\alpha_{i,u,w,x} \ket{i,u,w}\ket{x}$ in the support of $P_{k,l}$, we have
\begin{equation}\label{eq:succ_kl}
\begin{aligned}
    \norm{\Pi_{\success}\calT \ket{\chi}}^2 &= \sum_{i,u,w,(x_{a'})_{a'\notin w_a}}\biggl\|\sum_{(x_{a'})_{a'\in w_a}} \alpha_{i,u,w,x} \Pi_{\success} \calT \ket{i,u,w}\ket{x}\biggr\|^2
    \\
    &\leq  \sum_{i,u,w,(x_{a'})_{a'\notin w_a}}\Bigl(\sum_{(x_{a'})_{a'\in w_a}} \abs{\alpha_{i,u,w,x}}^2 \Bigr)\Bigl(\sum_{(x_{a'})_{a'\in w_a}}\1[\alpha_{i,u,w,x}\neq 0]\cdot \norm{\Pi_{\success} \calT \ket{i,u,w}\ket{x}}^2\Bigr)
    \\
    &\leq   \binom{3}{k,l,3-(k+l)}(N-1)^l \cdot \Bigl(\frac{1}{N}\Bigr)^k\Bigl(\frac{1}{N^2}\Bigr)^l \leq \frac{6}{N^{k+l}},
\end{aligned}
\end{equation}
where the first inequality is Cauchy-Schwarz and the second inequality is \cref{eq:succ_single_basis_kl}.

Finally, write $P$ for the projector onto basis states $\ket{i,u,w}\ket{x}$ not satisfying the first condition defining $P_{k,l}$ so that $\id = P + \sum_{k,l\in \mathbb{Z}_{\geq 0}\colon k+l \leq 3}P_{k,l}$. Therefore
\begin{align}
    \norm{\Pi_{\success}\calT (\id-\Pi^\triangle)\ket{\phi}}^2 &= \biggl\|\Pi_{\success}\calT (\id-\Pi^{\triangle})\Bigl(P + \sum_{k,l\in \mathbb{Z}_{\geq 0}\colon k+l \leq 3} P_{k,l}\Bigr)\ket{\phi}\biggr\|^2
    \notag
    \\
    &=\biggl\|\Pi_{\success}\calT (\id-\Pi^{\triangle})\Bigl(\sum_{k,l\in \mathbb{Z}_{\geq 0}\colon 0< k+l \leq 3} P_{k,l}\Bigr)\ket{\phi}\biggr\|^2
    \notag
    \\
    &= \biggl\|\Pi_{\success}\calT \Bigl(\sum_{k,l\in \mathbb{Z}_{\geq 0}\colon 1\leq k+l \leq 3} P_{k,l}\Bigr) (\id-\Pi^{\triangle})\ket{\phi}\biggr\|^2 \leq \frac{54}{N},\label{eq:guessing_prob}
\end{align}
where the second equality uses $\Pi_{\success}P = 0$ and $(\id-\Pi^{\triangle}) P_{0,0} = 0$ as observed previously, and the last inequality uses \cref{eq:succ_kl} and the restriction on the sum of $P_{k,l}$s to $1 \leq k+l$. 

The lemma follows from \cref{eq:guessing_prob} after taking square roots on both sides and recalling $N = \binom{n}{2}$.
\end{proof}

\subsection{Triangle finding upper bound}
In this section, we prove the following theorem.
\begin{theorem}\label{thm:worst-case-decision}
    For all $m,d \in \mathbb{N}$, $Q(\tri_{m,d}) \leq O(2^{(12/7)d}d^{3/7}m^{5/7})$.
\end{theorem}
\begin{corollary}\label{cor:worst-case-decision}
    For all $m\in \mathbb{N}$ and $d\leq O(\log(m)/\log\log(m))$, $Q(\tri_{m,d}) \leq O(m^{5/7 + o(1)})$. In view of \cref{fact:sparse-random-graph}, if $n\geq \Omega(m)$ there exists a quantum query algorithm using $O(m^{5/7 + o(1)})$ queries that finds a triangle in a random sparse graph $x \leftarrow \binom{[n]}{2}^m$, or decides it does not exist, with success probability $\geq 1 - o(1)$.
\end{corollary}

We prove \cref{thm:worst-case-decision} by adapting Belovs’s learning graph algorithm for $3$-distinctness \cite{kdist_learning_graphs_12}. A key challenge in Belovs’s work is the presence of \emph{faults}, as intuitively described in the introduction. At a technical level, faults are violations of the feasibility conditions in the semi-definite program characterizing quantum query complexity. Belovs found that directly using a learning graph algorithm of the type in his earlier work \cite{belovs_constant_sized_12} led to faults. In \cite{kdist_learning_graphs_12}, for $k$-distinctness, Belovs devised a technique based on the inclusion-exclusion principle to ``error-correct'' these faults at the cost of an exponential-in-$k$ factor, which is constant for constant $k$.

Technically speaking, in $3$-distinctness, at most one fault can occur because negative instances are limited to having $2$-collisions, and once one index of a $2$-collision is fixed, only its second index could possibly contribute a fault. In our problem, however, a single edge can belong to up to $2(d - 1)$ wedges, where $d$ is the maximum degree of the graph, introducing the possibility of multiple faults. Fortunately, Belovs’s learning graph algorithm for $k$-distinctness for general $k$ shows how inclusion-exclusion can be used to handle multiple faults. We adapt this technique to our problem and it represents the main change we make to his $3$-distinctness algorithm. The additional changes we make, such as explicitly defining the graph, replacing the concept of “arcs taken” with “active arcs,” and elaborating on somewhat opaque definitions, are primarily for improved clarity and accessibility. These adjustments do not alter the underlying algorithm and are not strictly necessary.

In summary, our algorithm for triangle finding is closely aligned with Belovs’s original learning graph algorithm for $3$-distinctness, which highlights the deep structural similarities between $3$-distinctness and triangle finding.

\subsubsection{Setting up proof of \cref{thm:worst-case-decision}}
When $d\geq \log(m)$, the trivial bound $Q(\tri_m)\leq m$ applies, ensuring \cref{thm:worst-case-decision} is satisfied. When $d\leq \log(m)$, we employ the learning graph method along with \cref{fact:vertex-avoid} to confirm the validity of \cref{thm:worst-case-decision}.
\paragraph{Adversary bound and learning graphs.}
In \cite{kdist_learning_graphs_12}, Belovs used a learning graph approach to construct matrices for the adversary bound. For every function $g:\mathcal{D}\subseteq \Sigma^m\to\{0,1\}$, the adversary bound $\Adv^{\pm}(g)$ satisfies the following theorem.
\begin{theorem}[\cite{negative_weights_adv},\cite{LMRSS11}]
    $Q(g) = \Theta(\Adv^{\pm}(g))$.
\end{theorem}
$\Adv^{\pm}(g)$ can be formulated as the following semi-definite program \cite[Theorem 6.2]{reichardt_adv_tight}:
\begin{alignat}{2}
&\text{minimize}\qquad && \max_{x \in \mathcal{D}} \sum_{j \in [m]} X_j[x,x],\label{complexity_of_learning_graph}\\
&\text{subject to}\qquad && \sum_{j \in [m] \colon x_j \neq y_j} X_j[x,y] = 1\qquad \text{if } g(x) \neq g(y);\label{feasibility_of_learning_graph}\\
& &&0 \leq X_j \in \mathbb{C}^{\calD \times \calD} \quad \text{for all $j \in [m]$}.
\end{alignat}

Every valid construction of $\{X_j\}_{j\in[m]}$ provides an upper bound for $\Adv^\pm(g)$ and, consequently, an upper bound for $Q(g)$ up to constant factors. We construct $\{X_j\}_{j\in[m]}$ using the learning graph framework, following a similar approach to that in \cite{kdist_learning_graphs_12}. Specifically, for each arc $A_j^{R,S}$ in the learning graph from vertex $R$ and to vertex $S$ where $R$ and $S$ uniquely determines $j$ (to be shown in detail later), we associate a positive semi-definite matrix $X_j^{R,S}$, and define
\begin{equation}
    X_j \coloneqq \sum_{R,S} X_j^{R,S}.
\end{equation}

\subsubsection{Learning graph construction}
We first describe the construction of the learning graph. For simplicity, we use $f$ to denote $\tri_{m,d}$. The vertices of the learning graph are
\begin{equation}
    \dot{\bigcup}_{i\in\{0,1,2,3,4\}} V^{(i)},
\end{equation}
where a vertex $S^{(i)}$ in $V^{(i)}$ is (labelled by) an array of length $(2^{2d} - 1) + 2d$ containing pairwise-disjoint subsets of $[m]$; we refer to entries of $S^{(i)}$ by
\begin{equation}
    S^{(i)} = (S_1^{(i)}(\Gamma),S_2^{(i)}(\gamma))_{\emptyset\neq \Gamma \subseteq [2d],\, \gamma\in [2d]} = (S_1^{(i)}(\{1\}),\dots, S_1^{(i)}([2d]),  S_2^{(i)}(1), \dots, S_2^{(i)}(2d)).
\end{equation}
We further impose conditions on the sizes of these sets using some $r_1,r_2\in \mathbb{N}$ that will be chosen later, at the end of \cref{subsubsection:learning_graph_complexity}. To describe these conditions, it is convenient to write $\mu(A)$ for the minimum element of a finite non-empty subset $A \subset \mathbb{Z}$. 

First, $V^{(1)}$ consists of all vertices $S^{(1)}$ such that
\begin{equation}
    \forall \emptyset\neq \Gamma\subseteq[2d], \gamma\in[2d], \quad \abs{S_1^{(1)}(\Gamma)} = r_1 \quad \text{and} \quad \abs{S_2^{(1)}(\gamma)} = r_2.
\end{equation}
Then, for $i\in \{2,3,4\}$, we impose that the set $V^{(i)}$ can be partitioned as $V^{(i)} = \dot{\cup}_{\emptyset \neq \Gamma \subseteq[2d]} V^{(i)}(\Gamma)$ such that 
\begin{itemize}
    \item $V^{(2)}(\Gamma)$ consists of all vertices $S^{(2)}$ with
\begin{equation}
    \abs{S_1^{(2)}(\Gamma)} = r_1+1, \qquad \forall \Gamma'\neq \Gamma,\abs{S_1^{(2)}(\Gamma')} = r_1,\qquad \forall \gamma\in[2d], \abs{S_2^{(2)}(\gamma)} = r_2.
\end{equation}
\item $V^{(3)}(\Gamma)$ consists of all vertices $S^{(3)}$ such that
\begin{equation}
    \abs{S_1^{(3)}(\Gamma)} = r_1+1, \qquad \forall \Gamma'\neq \Gamma,\abs{S_1^{(3)}(\Gamma')} = r_1,
\end{equation}
and
\begin{equation}
    \abs{S_2^{(3)}(\mu(\Gamma))} = r_2+1,\qquad \forall \gamma\neq \mu(\Gamma),\abs{S_2^{(3)}(\gamma)} = r_2.
\end{equation}
\item 
$V^{(4)}(\Gamma)$ consists of all vertices $S^{(4)}$ such that
\begin{equation}
    \abs{S_1^{(4)}(\Gamma)} = r_1+1, \qquad \forall \Gamma'\neq \Gamma,\abs{S_1^{(4)}(\Gamma')} = r_1,
\end{equation}
and
\begin{equation}
    \abs{S_2^{(4)}(\mu(\Gamma))} = r_2+2,\qquad \forall \gamma\neq \mu(\Gamma),\abs{S_2^{(4)}(\gamma)} = r_2.
\end{equation}
\end{itemize}

For each vertex $R$ in the learning graph, define
\begin{equation}
    \bigcup R_1\coloneqq \bigcup_{\emptyset\neq \Gamma\subseteq[2d]}R_1(\Gamma),
    \quad
    \bigcup R_2\coloneqq \bigcup_{\gamma\in[2d]}R_2(\gamma),
    \quad \text{and} \quad
    \bigcup R\coloneqq \bigcup R_1 \cup \bigcup R_2.
\end{equation}
We refer to indices contained in the first, second, and third sets as having been loaded in $R_1$, $R_2$, and $R$, respectively. If there is an arc from vertex $R$ to vertex $S$ in the learning graph, by the construction in the following paragraphs, it must satisfy $\bigcup S = \bigcup R \cup \{j\}$. We denote the arc as $A^{R,S}_j$, and accordingly, we say $A^{R,S}_j$ is associated with the loading of $j$.

There exists a unique vertex $R^{(0)}\in V^{(0)}$ such that $\bigcup R^{(0)} = \emptyset$. We refer to this vertex as the source, denoted by $\emptyset$. Each vertex $R^{(1)}\in V^{(1)}$ loads exactly
$r \coloneqq r_1\cdot (2^{2d}-1) + r_2\cdot(2d)$ indices. We fix an arbitrary ordering $t_1,\dots, t_r$ of indices in $\bigcup R^{(1)}$ such that all indices in $\bigcup R^{(1)}_1$ precede those in $\bigcup R^{(1)}_2$. There exists a length-$r$ path from the source $\emptyset$ to $R^{(1)}$, with all intermediate vertices lying in $V^{(0)}$. Along this path, the element $t_i$ is loaded on the $i$-th arc for $i\in[r]$ in stages I.1 and I.2. Specifically, 
\begin{itemize}
    \renewcommand{\labelitemi}{\tiny$\blacksquare$}
    \item 
    If $t_i\in \bigcup R_1^{(1)}$, it is loaded in stage I.1,
    \item 
    If $t_i\in \bigcup R_2^{(1)}$, it is loaded in stage I.2.
\end{itemize}
All such paths are disjoint except at the source $\emptyset$. Therefore, there are $r\bigl|V^{(1)}\bigr|$ arcs in stages I.1 and I.2 in total.

In stage II.1, for each $R^{(1)}\in V^{(1)}$, $j\not\in\bigcup R^{(1)}$ and $\emptyset\neq \Gamma\subseteq[2d]$, there is an arc, denoted  $A^{R^{(1)}, R^{(2)}}_j$, from $R^{(1)}$ to $R^{(2)} \in V^{(2)}(\Gamma)$ defined by 
\begin{equation}
\left\{\begin{alignedat}{3}
    &R^{(2)}_1(\Gamma) &&= R^{(1)}_1(\Gamma)\cup\{j\},&&\\
    &R^{(2)}_1(\Gamma') &&= R^{(1)}_1(\Gamma'),&& \text{for all $\emptyset \neq \Gamma'\subseteq [2d]$ s.t. $\Gamma'\neq \Gamma$},\\
    &R^{(2)}_2(\gamma) &&= R^{(1)}_2(\gamma),&&\text{for all $\gamma\in[2d]$.}
\end{alignedat}\right.
\end{equation}

In stage II.2, for each $\emptyset\neq \Gamma\subseteq[2d]$, $R^{(2)}\in V^{(2)}(\Gamma)$, and $j\not\in\bigcup R^{(2)}$, there is an arc, denoted $A^{R^{(2)},R^{(3)}}_j$, from $R^{(2)}$ to $R^{(3)} \in V^{(3)}(\Gamma)$ defined by
\begin{equation}
\left\{\begin{alignedat}{3}
    & R^{(3)}_1(\Gamma') &&= R^{(2)}_1(\Gamma'),\qquad &&\text{for all $\emptyset\neq \Gamma'\subseteq[2d]$},\\
    &R^{(3)}_2(\gamma) &&= R^{(2)}_2(\gamma),&& \text{for all $\gamma \neq \mu(\Gamma)$},\\
    &R^{(3)}_2(\mu(\Gamma)) &&= R^{(2)}_2(\mu(\Gamma))\cup\{j\}.&&
\end{alignedat}\right.
\end{equation}

In stage II.3, for each $\emptyset\neq \Gamma\subseteq[2d]$, $R^{(3)}\in V^{(3)}(\Gamma)$, and $j\not\in\bigcup R^{(3)}$, there is an arc, denoted $A^{R^{(3)},R^{(4)}}_j$, from $R^{(3)}$ to $R^{(4)} \in V^{(4)}(\Gamma)$ defined by
\begin{equation}
\left\{\begin{alignedat}{3}
    & R^{(4)}_1(\Gamma') &&= R^{(3)}_1(\Gamma'),\qquad &&\text{for all $\emptyset\neq \Gamma'\subseteq[2d]$},\\
    &R^{(4)}_2(\gamma) &&= R^{(3)}_2(\gamma),&& \text{for all $\gamma \neq \mu(\Gamma)$},\\
    &R^{(4)}_2(\mu(\Gamma)) &&= R^{(3)}_2(\mu(\Gamma))\cup\{j\}.&&
\end{alignedat}\right.
\end{equation}

\subsubsection{Active arcs}
For every $x\in f^{-1}(1)$, let $\mathcal{C}(x) = \{a_1, a_2, a_3\}\subseteq[m]$ be a specific certificate for $x$ such that $a_1<a_2<a_3$ and $x_{a_1}$, $x_{a_2}$, and $x_{a_3}$ form a triangle.

We say $R^{(1)}\in V^{(1)}$ is consistent with $x\in f^{-1}(1)$ if edges indexed by $\mathcal{C}(x)$ are vertex disjoint from edges indexed by $\bigcup R^{(1)}$. Let $n_x$ be the number of $R^{(1)}\in V^{(1)}$ that are consistent with $x$, and define
\begin{equation}
    q\coloneqq \bigl(\min_{x\in f^{-1}(1)}\{n_x\}\bigr)^{-1},
\end{equation}
so that for all $x\in f^{-1}(1)$, there exist at least $q^{-1}$ vertices in $V^{(1)}$ that are consistent with $x$. By \cref{fact:vertex-avoid}, when $rd\leq o(m)$ --- a condition satisfied since we only need to consider $d\leq\log_2(m)$ and use the $r$ specified in \cref{subsubsection:learning_graph_complexity} --- it follows that $q^{-1}\geq\Omega(\bigl|V^{(1)}\bigr|)$.

Let $\Cst(x)$ denote a fixed but arbitrary set of $q^{-1}$ vertices $R^{(1)}\in V^{(1)}$ that are consistent with $x$. For each $R^{(1)}\in \Cst(x)$, $\Act(x,R^{(1)})$ is a set of active arcs consisting of the following arcs:
\begin{itemize}
     \setlength{\itemsep}{0.5pt}
    \renewcommand{\labelitemi}{\tiny$\blacksquare$}
    \item In stage I.1 and I.2, all arcs along the unique shortest length $r$ path from source $\emptyset$ to $R^{(1)}$.
    \item 
    In stage II.1, all arcs from $R^{(1)}$ that loads $a_1$ into $R^{(1)}$. Notice that there are $2^{2d}-1$ such arcs, one for each choice of $\emptyset\neq \Gamma\subseteq [2d]$.
    \item 
    In stage II.2, for each vertex $R^{(2)}\in V^{(2)}$ that has an incoming active arc from $\Act(x,R^{(1)})$, the arc that loads $a_2$ into $R^{(2)}$.
    \item 
    In stage II.3, for each vertex $R^{(3)}\in V^{(3)}$ that has an incoming active arc from $\Act(x,R^{(1)})$, the arc that loads $a_3$ into $R^{(3)}$.
\end{itemize}
Notice that in stage II.2 and II.3, for each vertex that has an incoming active arc, there is only one outgoing active arc, so the total number of active arcs in stages II.2 and II.3 is $2(2^{2d}-1)$.

Lastly, the set of all active arcs of $x$ is defined by
\begin{equation}
    \mathrm{Act}(x)\coloneqq \dot{\bigcup}_{R^{(1)}\in \Cst(x)}\Act(x,R^{(1)}).
\end{equation}

\subsubsection{Matrices for the adversary bound}
For every vertex $R$ in the learning graph, an \emph{assignment on $R$} refers to a function $\alpha_R: \bigcup R\to \binom{[n]}{2}\cup\{*\}$ such that:
\begin{enumerate}
    \item for all $j\in \bigcup R_1$,
    \begin{equation}
   \alpha_R(j)\neq *;
\end{equation}
\item for all $\gamma\in [2d]$ and $j\in  R_2(\gamma)$, 
\begin{equation}
   \alpha_R(j)\in\{*\}\cup\Bigl\{e\in \textstyle\binom{[n]}{2}\Bigm| \exists k\in \bigcup_{\Gamma\ni \gamma}R_1(\Gamma)\text{ s.t. $e$  is incident to $\alpha_R(k)$} \Bigr\}.
\end{equation}
Note that the special symbol $*$ is not incident to any edge.
\end{enumerate}
We say an input $z\in \binom{[n]}{2}^m$ \emph{satisfies assignment $\alpha_{R}$} if, for all $t\in \bigcup R$,
\begin{equation}\label{eq:alpha_def}
    \alpha_{R}(t) =
    \begin{cases}
        z_t, &\text{if $t\in \bigcup R_1$},\\
        z_t, &\text{if $t\in R_2(\gamma)$ and $\exists k\in \bigcup_{\Gamma\ni \gamma}R_1(\Gamma)$ s.t. $z_t$  is incident to $z_k$},\\
        *, &\text{otherwise}.
    \end{cases}
\end{equation}
For each vertex $R$ in the learning graph, we write $\alpha_R^x$ for the unique assignment on $R$ that $x$ satisfies. We say arc $A^{R,S}_j$ \emph{uncovers $j$} if $\alpha_S^x(j)\neq *$. We say inputs $x,y\in \binom{[n]}{2}^m$ \emph{agree on $R$} if they satisfy $\alpha_R^x = \alpha_R^y$; we also say they agree on a subset of $\bigcup R$ if the restrictions of $\alpha_R^x$ and $\alpha_R^y$ to that subset equal.

Define
\begin{equation}
    X_j^{R,S}\coloneqq \sum_{\alpha_R} Y_{\alpha_R},
\end{equation}
where the summation is over all assignments $\alpha_R$ on $R$. 

For each arc $A_j^{R,S}$ in stage I.1, define $Y_{\alpha_R}\coloneqq q\psi_{\alpha_R}\psi_{\alpha_R}^\dagger$, where $\psi_{\alpha_R}$ is a real vector indexed by $\binom{[n]}{2}^m$ and defined entry-wise by
\begin{equation}
    \psi_{\alpha_R}[z] \coloneqq
    \begin{cases}
        1/\sqrt{w}, & \text{if $f(z)=1$, $z$ satisfies $\alpha_R$, and $A^{R,S}_j\in \Act(z)$},\\
        \sqrt{w}, & \text{if $f(z)=0$,  and $z$ satisfies $\alpha_R$},\\
        0,&\text{otherwise.}
    \end{cases}
\end{equation}
Here, $w$ is a positive real number that will be specified later. With the above definition, we see that $X_j^{R,S}$ consists of blocks of the form:
\begin{equation}
\label{X_block_form_I.1}
\renewcommand{\arraystretch}{1.25}
\begin{array}{c|c|c}
    & x & y \\ \hline
x & q/w & q \\ \hline
y & q & qw
\end{array}
\end{equation}
where $x\in f^{-1}(1)$, $y\in f^{-1}(0)$, $A^{R,S}_j\in \Act(x)$, and both $x$ and $y$ agree on $R$.

For arcs $A_j^{R,S}$ in stage I.2, define $Y_{\alpha_R}\coloneqq q(\psi_{\alpha_R}\psi_{\alpha_R}^\dagger+\phi_{\alpha_R}\phi_{\alpha_R}^\dagger)$ where $\psi_{\alpha_R}$ and  $\phi_{\alpha_R}$ are real vectors indexed by $\binom{[n]}{2}^m$ and defined entry-wise by
\begin{equation}
    \psi_{\alpha_R}[z] \coloneqq
    \begin{cases}
        1/\sqrt{w_1}, & \text{if $f(z)=1$, $\alpha_S^z(j)\neq*$, $z$ satisfies $\alpha_R$, and $A^{R,S}_j\in \Act(z)$},\\
        \sqrt{w_1}, & \text{if $f(z)=0$,  and $z$ satisfies $\alpha_R$},\\
        0,&\text{otherwise;}
    \end{cases}
\end{equation}
and
\begin{equation}
    \phi_{\alpha_R}[z] \coloneqq
    \begin{cases}
        1/\sqrt{w_0}, & \text{if $f(z)=1$, $\alpha_S^z(j)=*$, $z$ satisfies $\alpha_R$, and $A^{R,S}_j\in \Act(z)$},\\
        \sqrt{w_0}, & \text{if $f(z)=0$, $\alpha_S^z(j)\neq*$, and $z$ satisfies $\alpha_R$},\\
        0, & \text{otherwise.}
    \end{cases}
\end{equation}
Here, $w_0$ and $w_1$ are positive real numbers that will be specified later. With the above definition, we see that $X_j^{R,S}$ consists of blocks of the form:
\begin{equation}
\label{X_block_form_I.2}
\begin{array}{l|c|c|c|c}
& x: x_j =\alpha_S^x(j)\neq * & x: \alpha_S^x(j)= * & y: y_j =\alpha_S^y(j)\neq * & y: \alpha_S^y(j)= * \\ \hline
x: x_j =\alpha_S^x(j)\neq * & q/w_1 & 0 & q & q \\ \hline
x: \alpha_S^x(j)= * & 0 & q/w_0 & q & 0 \\ \hline
y: y_j =\alpha_S^y(j)\neq * & q & q & q(w_0 + w_1) & qw_1 \\ \hline
y: \alpha_S^y(j)= * & q & 0 & qw_1 & qw_1 \\ 
\end{array}
\end{equation}
where $x\in f^{-1}(1)$, $y\in f^{-1}(0)$, $A^{R,S}_j\in \Act(x)$, and both $x$ and $y$ agree on $R$.

For arcs $A_j^{R,S}$ in stage II.$s$ where $s\in[3]$, there exists a unique $\emptyset\neq \Gamma\subseteq[2d]$ such that $S\in V^{(s+1)}(\Gamma)$. Define $Y_{\alpha_R}\coloneqq q\psi_{\alpha_R}\psi_{\alpha_R}^\dagger$ where
\begin{equation}
    \psi_{\alpha_R}[z] \coloneqq
    \begin{cases}
        1/\sqrt{w_2},& \text{if $f(z)=1$, $z$ satisfies  $\alpha_R$, and $A_j^{R,S}\in \Act(z)$},\\
        (-1)^{1+\abs{\Gamma}}\sqrt{w_2}, & \text{if $f(z)=0$, and $z$ satisfies $\alpha_R$},\\
        0,&\text{otherwise.}
    \end{cases}
\end{equation}
$X_j^{R,S}$ consists of blocks of the form:
\begin{equation}
\label{X_block_form_II}
\renewcommand{\arraystretch}{1.25}
\begin{array}{c|c|c}
    & x & y \\ \hline
x & q/w_2 & (-1)^{1+\abs{\Gamma}}q \\ \hline
y & (-1)^{1+\abs{\Gamma}}q & qw_2
\end{array}
\end{equation}
where $x\in f^{-1}(1)$, $y\in f^{-1}(0)$, $A^{R,S}_j\in \Act(x)$, and both $x$ and $y$ agree on $R$.

\subsubsection{Complexity}\label{subsubsection:learning_graph_complexity}
We show $Q(\tri_{m,d}) = O(2^{(12/7)d}d^{3/7}m^{5/7})$ by computing \cref{complexity_of_learning_graph}. Define $W$ as the maximum number of wedges in any input to $\tri_{m,d}$. Since the graph has maximum degree $d$, and each edge in the graph contributes to at most $2(d-1)$ wedges, we have $W\leq O(md)$.
\begin{enumerate}
    \item 
    For stage I.1, we set the weight $w=1$ for all arcs in this stage. There are $r_1(2^{2d}-1)\bigl|V^{(1)}\bigr|$ arcs in this stage. By \cref{X_block_form_I.1}, each of them contributes $\max\{qw, q/w\} = q$ to the complexity, so the complexity of this stage is $r_1(2^{2d}-1)\bigl|V^{(1)}\bigr|q = O(2^{2d}r_1)$.
    \item 
    For stage I.2, there are $2dr_2\bigl|V^{(1)}\bigr|$ arcs in this stage. For each input $z$, we need to bound the number of arcs that can uncover an element, so we get a refined bound on the contribution from \cref{X_block_form_I.2}. Suppose such an arc in this stage is on the shortest length-$r$ path from source $\emptyset$ to $R^{(1)}\in V^{(1)}$, and is uncovering $j$ in $R_2^{(1)}(\gamma)$ for some $\gamma\in[2d]$, then, $z_j$ must form a wedge with $z_i$ for some $i\in R^{(1)}_1(\Gamma)$ and $\gamma\in \Gamma\subseteq[2d]$. The number of such ordered tuples $(i,j)$ is at most $2W$, and the number of vertices in stage I.2 that loads $j$ in its outgoing arc and has loaded $i$ in stage I.1 is upper bounded by
    \begin{equation}
        (2^{2d}-1)(2d)\binom{m-2}{r_1-1,\underbrace{r_1,\cdots,r_1}_{2^{2d}-2},r_2-1,\underbrace{r_2,\cdots,r_2}_{2d-1}} = O\Bigl(\frac{r_1r_2d2^{2d}}{m^2}\bigl|V^{(1)}\bigr|\Bigr).
    \end{equation}
    
    Therefore, for a negative input, by \cref{X_block_form_I.2}, stage I.2's contribution to the complexity is
    \begin{equation}
        O\Bigl(qw_0W\frac{r_1r_2d2^{2d}}{m^2}\bigl|V^{(1)}\bigr|+qw_12dr_2\bigl|V^{(1)}\bigr|\Bigr)=O\Bigl(\frac{r_1r_2d^22^{2d}}{m}w_0+r_2dw_1\Bigr);
    \end{equation}
    and for a positive input, by \cref{X_block_form_I.2}, stage I.2's contribution to the complexity is
    \begin{equation}
        O\Bigl(\frac{q}{w_1}W\frac{r_1r_2d2^{2d}}{m^2}\bigl|V^{(1)}\bigr|+\frac{q}{w_0}2dr_2\bigl|V^{(1)}\bigr|\Bigr)=O\Bigl(\frac{r_1r_2d^22^{2d}}{mw_1}+\frac{r_2d}{w_0}\Bigr).
    \end{equation}
    If we set $w_0=\frac{1}{2^d}\sqrt{\frac{m}{r_1d}}$, and $w_1 = 1/w_0$, the total contribution to the complexity is $O(r_2d^{3/2}2^d\sqrt{\frac{r_1}{m}})$.
    \item 
    For stage II,
    \begin{itemize}
    \renewcommand{\labelitemi}{\tiny$\blacksquare$}
        \item 
        The total number of arcs in stage II.1 is
        \begin{equation}
            (m-r)(2^{2d}-1)\bigl|V^{(1)}\bigr|\leq  O\bigl(2^{2d}m\bigl|V^{(1)}\bigr|\bigr).
        \end{equation}
        \item
        The total number of arcs in stage II.2 is
        \begin{equation}
            (m-r-1)\bigl|V^{(2)}\bigr| = \frac{(m-r-1)(m-r)(2^{2d}-1)\bigl|V^{(1)}\bigr|}{r_1+1},
        \end{equation}
        since each vertex in $V^{(1)}$ has $(m-r)(2^{2d}-1)$ outgoing arcs into $V^{(2)}$, and each vertex in $V^{(2)}$ has $r_1+1$ incoming arcs from $V^{(1)}$. There are $O\left(2^{2d}m^2\bigl|V^{(1)}\bigr|/r_1\right)$ arcs in this stage.
        \item
        The total number of arcs in stage II.3 is
        \begin{equation}
            (m-r-2)\bigl|V^{(3)}\bigr| = \frac{(m-r-2)(m-r-1)(m-r)(2^{2d}-1)\bigl|V^{(1)}\bigr|}{(r_1+1)(r_2+1)},
        \end{equation}
        since each vertex in $V^{(2)}$ has $m-r-1$ outgoing arcs into $V^{(3)}$, and each vertex in $V^{(3)}$ has $r_2+1$ incoming arcs from $V^{(2)}$. There are $O(2^{2d}m^3\bigl|V^{(1)}\bigr|/(r_1r_2))$ arcs in this stage.
    \end{itemize}
    For every negative input, by \cref{X_block_form_II}, each arc contributes $qw_2$ to the complexity, so the total contribution is
    \begin{equation}
        O\biggl(\frac{2^{2d}m^3\bigl|V^{(1)}\bigr|}{r_1r_2}qw_2\biggr) = O\biggl(\frac{2^{2d}m^3w_2}{r_1r_2}\biggr).
    \end{equation}

    For every positive input, there are exactly $3q^{-1}(2^{2d}-1)$ active arcs in stage II, by \cref{X_block_form_II}, each of them contributes $q/w_2$ to the complexity, so the total contribution is
    \begin{equation}
        \frac{3(2^{2d}-1)}{q}\cdot \frac{q}{w_2} = \frac{3(2^{2d} -1)}{w_2} \leq O\left(\frac{2^{2d}}{w_2}\right).
    \end{equation}

    To balance the contribution from negative inputs and positive inputs, we can set $w_2=\sqrt{r_1r_2/m^3}$ so that the total contribution in stage II is
    \begin{equation}
        O\bigl(2^{2d}\sqrt{m^3/(r_1r_2)}\bigr).
    \end{equation}
\end{enumerate}

The total complexity of all stages is
\begin{equation}
    O\biggl(2^{2d}r_1 + r_2d^{3/2} \, 2^d\sqrt{\frac{r_1}{m}} + 2^{2d}\sqrt{\frac{m^3}{r_1r_2}}\biggr).
\end{equation}
To balance the summands, we set $r_1 = \ceil{m^{5/7}d^{3/7}/2^{(2/7)d}}$ and $r_2 = \ceil{2^{(6/7)d}m^{6/7}/d^{9/7}}$, which leads to a complexity of $O(2^{(12/7)d}d^{3/7}m^{5/7})$.

\subsubsection{Feasibility}
Fix two inputs $x\in f^{-1}(1)$ and $y\in f^{-1}(0)$, to prove \cref{feasibility_of_learning_graph} holds, it is equivalent to show
\begin{equation}
    \sum_{A_j^{R,S} \in \Act(x) \colon x_j \neq y_j} X_j^{R,S}[x,y] = 1.
\end{equation}
Since $\Act(x)= \dot{\bigcup}_{R^{(1)}\in \Cst(x)}\Act(x,R^{(1)})$, it suffices to prove
\begin{equation}\label{eq:feasibility}
    \sum_{A_j^{R,S} \in \Act(x,R^{(1)}) \colon x_j \neq y_j} X_j^{R,S}[x,y] = \frac{1}{\abs{\Cst(x)}} = q,
\end{equation}
for each $R^{(1)}\in \Cst(x)$.

Recall $\mathcal{C}(x) = \{a_1,a_2,a_3\}$ is a specific certificate for $x$ we choose. Let $t_1, t_2,\cdots,t_r\in[m]$ be the order of elements in $\bigcup R^{(1)}$ get loaded in $R^{(1)}$. Let $T_i, i\in\mathbb{Z}_{[0,r]}$, be the vertex that has loaded $i$ elements on the unique shortest length-$r$ path from source $\emptyset$ to $R^{(1)}$ ($\abs{\bigcup T_i} = i$, $T_0 = \emptyset$ and $T_r= R^{(1)}$), so that arcs from $\emptyset$ to $R^{(1)}$ is in the form of $A_{t_i}^{T_{i-1},T_i}\in \Act(x, R^{(1)})$.
Depending on if $x$ and $y$ agree on $R^{(1)}$,
\begin{itemize}
\renewcommand{\labelitemi}{\tiny$\blacksquare$}
    \item 
    If $x$ and $y$ disagree on $R^{(1)}$, there exists $i^* \in [r]$ such that $x$ and $y$ disagree on $T_i$ if and only if $i\geq i^*$. This follows from the fact that, for each arc $A_j^{R,S}$ in stage I, the following holds by \cref{eq:alpha_def}:
    \begin{equation}
        \forall k\in \bigcup R, \alpha_S^x(k) = \alpha_R^x(k) \text{ and } \alpha_S^y(k) = \alpha_R^y(k),
    \end{equation}
    so if $x$ and $y$ disagree on $R$, then they disagree on $S$.
    
    We show that $X_{t_{i^*}}^{T_{i^*-1}, T_{i^*}}[x,y] = q$, and for each $A_j^{R,S}\in \Act(x,R^{(1)})\backslash\{A_{t_{i^*}}^{T_{i^*-1}, T_{i^*}}\}$, either $x_j = y_j$ or $X_j^{R,S}[x,y] = 0$.
    \begin{itemize}
        \item 
        When we load $t_{i^*}$ into $T_{i^*-1}$, $x$ and $y$ agree on $T_{i^*-1}$ but not $T_{i^*}$. Since $x$ and $y$ agree on $T_{i^*-1}$, for all $i\in \bigcup T_{i^*-1}$, we have
        \begin{equation}
            \alpha_{T_{i^*}}^x(i) = \alpha_{T_{i^*-1}}^x(i) = \alpha_{T_{i^*-1}}^y(i) = \alpha_{T_{i^*}}^y(i),
        \end{equation}
        so we must have $\alpha_{T_{i^*}}^x(t_{i^*}) \neq \alpha_{T_{i^*}}^y(t_{i^*})$ as $x$ and $y$ disagree on $T_{i^*}$. Therefore, $x_{t_{i^*}} \neq y_{t_{i^*}}$ as otherwise $\alpha_{T_{i^*}}^x(t_{i^*}) = \alpha_{T_{i^*}}^y(t_{i^*})$. Hence, the term $X_{t_{i^*}}^{T_{i^*-1}, T_{i^*}}[x,y]$ is included in the summation in \cref{eq:feasibility}. Depending on which stage $A_{t_{i^*}}^{T_{i^*-1},T_{i^*}}$ is in, either by \cref{X_block_form_I.1} or by \cref{X_block_form_I.2},
        \begin{equation}
            X_{t_{i^*}}^{T_{i^*-1}, T_{i^*}}[x,y] = q.
        \end{equation}
        \item
        When we load $t_i$ into $T_{i-1}$ such that $i\in[i^*-1]$, $x$ and $y$ agree on both $T_{i-1}$ and $T_{i}$. Then, either $x_{t_i} = \alpha_{T_i}^x(t_i) = \alpha_{T_i}^y(t_i) = y_{t_i} \neq *$, or $\alpha_{T_i}^x(t_i) = \alpha_{T_i}^y(t_i) = *$ so that by \cref{X_block_form_I.2}, $X^{T_{i-1},T_i}_{t_i}[x,y] = 0$. Therefore, arc $A_{t_{i}}^{T_{i-1}, T_{i}}$ does not contribute to \cref{eq:feasibility}.
        \item 
        Every other arc $A^{R,S}_j\in \Act(x,R^{(1)})$ is either in stage I and of the form $A_{t_i}^{T_{i-1},T_i}$ for some $i>i^*$, $i\in[r]$, or the arc is in stage II. If the arc is in stage I, then $x$ and $y$ disagree on $R$. Depending on which stage $A_{t_i}^{T_{i-1},T_i}$ is in, either by \cref{X_block_form_I.1} or by \cref{X_block_form_I.2}, $X_j^{R,S}[x,y]=0$.
        
        If $A^{R,S}_j$ is in stage II, it suffices to show that if $x$ and $y$ disagree on $R$, then they also disagree on $S$. If this holds, a short induction establishes that $x$ and $y$ disagree on $R$ for all $A^{R,S}_j\in \Act(x, R^{(1)})$ in stage II. By \cref{X_block_form_II}, it then follows that $X_j^{R,S}[x,y]=0$.
        
        By the definition of consistency, after loading $j\in\mathcal{C}(x)$,
        \begin{equation}
            \forall i\in \bigcup R, \alpha_S^x(i) = \alpha_R^x(i).
        \end{equation}
        If $x$ and $y$ agree on $S$, we must have
        \begin{equation}\label{eq:feasibility_contradiciton_agree_on_S}
            \forall i\in \bigcup R, \alpha_S^y(i) = \alpha_S^x(i) = \alpha_R^x(i).
        \end{equation}
        However, there exists $k\in\bigcup R$ such that $\alpha_R^y(k) \neq \alpha_R^x(k) = \alpha_S^y(k)$ because $x$ and $y$ disagree on $R$. Since $\{\alpha_R^y(k),\alpha_S^y(k)\} \subseteq \{*, y_k\}$, if \cref{eq:feasibility_contradiciton_agree_on_S} holds, it follows that $\alpha_R^y(k) = *$ and $y_k = \alpha_S^y(k) = \alpha_S^x(k) = x_k$. Consequently, we must have $k\in \bigcup S_2$ and $j\in \bigcup S_1$, implying $j=a_1$, and $y_k = x_k$ is incident to $y_{a_1} = x_{a_1}$. This contradicts the definition of consistency. Therefore, $x$ and $y$ must disagree on $S$. 
    \end{itemize}
    \item 
    If $x$ and $y$ agree on $R^{(1)}$, we first show the contribution from arcs in $\Act(x,R^{(1)})$ in stage I to \cref{eq:feasibility} is 0. For each $i\in[r]$, when we load $t_i$ into $T_{i-1}$ , $x$ and $y$ agree on both $T_{i-1}$ and $T_{i}$. An identical argument as in a previous case when $i\in[i^*-1]$ shows arc $A_{t_{i}}^{T_{i-1}, T_{i}}$ does not contribute to \cref{eq:feasibility}.

    For stage II, let $k$ be the smallest number in $[3]$ that $x_{a_k} \neq y_{a_k}$. $k$ must exist because $y$ is a negative instance. For active arcs in stage II.$k'$ such that $k' < k$, those arcs are loading $a_{k'}$ but $x_{a_{k'}}=y_{a_{k'}}$, so those arcs' $X$ are not included in the summation in \cref{eq:feasibility}.

    For arcs in stage II.$k'$ such that $k' > k$, $x$ and $y$ disagree on the vertices before loading $a_{k'}$. This is because $a_{k}$ gets uncovered in $x$, so $y_{a_{k}}$ must be equal to $x_{a_k}$ for $x$ and $y$ to agree, which leads a contradiction. Then, by \cref{X_block_form_II}, those active arcs' contribution to \cref{eq:feasibility} is also 0.

    Next, we show the contribution from arcs in $\Act(x,R^{(1)})$ in stage II.$k$ to \cref{eq:feasibility} is exactly $q$.
    \begin{itemize}
        \item 
        If $k=1$, there are $2^{2d}-1$ arcs in $\Act(x,R^{(1)})$ in stage II.1. They are of the form $A^{R^{(1)}, R_\Gamma^{(2)}}_{a_1}$ where $R_\Gamma^{(2)}\in V^{(2)}(\Gamma)$, one for each $\emptyset\neq \Gamma\subseteq[2d]$. Since $x$ and $y$ agree on $R^{(1)}$ (and $x_{a_1}\neq y_{a_1}$), \cref{X_block_form_II} gives $X^{R^{(1)}, R_\Gamma^{(2)}}_{a_1}[x,y] = (-1)^{1+\abs{\Gamma}}q$. Therefore, the total contribution of stage II.1 to \cref{eq:feasibility} is
        \begin{equation}
            \sum_{\emptyset\neq \Gamma \subseteq [2d]} (-1)^{1+\abs{\Gamma}}q = q.
        \end{equation}
        \item 
        If $k=2$, there are $2^{2d}-1$ arcs in $\Act(x,R^{(1)})$ in stage II.2. They are of the form $A^{R_\Gamma^{(2)}, R_\Gamma^{(3)}}_{a_2}$ where $R_\Gamma^{(2)}\in V^{(2)}(\Gamma)$ and $R_\Gamma^{(3)}\in V^{(3)}(\Gamma)$, one for each $\emptyset\neq \Gamma\subseteq[2d]$. In this case, $x$ and $y$ may not agree on $R_\Gamma^{(2)}$ due to what are known as \emph{faults}.
        
        We say that an index $i\in \bigcup (R_\Gamma^{(2)})_2$ is \emph{faulty} if  $y_i$ is incident to $y_{a_1}$. For $\gamma\in [2d]$, we say the subset $(R_\Gamma^{(2)})_2(\gamma)$ is \emph{faulty} if $(R_\Gamma^{(2)})_2(\gamma)$ contains a faulty index.  Let $I$ denote the set of all $\gamma\in [2d]$ such that $(R_\Gamma^{(2)})_2(\gamma)$ is not faulty.

        We now show
        \begin{equation}
            \text{$x,y$ agree on $R^{(2)}_\Gamma$} \iff \emptyset\neq \Gamma \subseteq I.
        \end{equation}
        
        \begin{enumerate}
            \item ``$\implies$'': consider the contrapositive. Let $\gamma \in \Gamma-I$. Suppose $j\in (R^{(2)}_\Gamma)_2(\gamma)$ is faulty. Then by our definition of assignment, $\alpha_{R^{(2)}_\Gamma}^y(j) = y_j$. On the other hand, $\alpha_{R^{(2)}_\Gamma}^x(j)$ cannot be equal to $y_j$ since $y_j$ is incident to $y_{a_1} = x_{a_1}$ and $R^{(1)}$ is consistent with $x$.
            \item ``$\impliedby$'': note that $x$ and $y$ certainly agree on $\bigcup (R^{(2)}_\Gamma)_1$ since $x_{a_1} = y_{a_1}$ and $x$ and $y$ agreed on $\bigcup R^{(1)}_1$. But since $\Gamma \subseteq I$, and $x$ and $y$ agreed on $\bigcup R^{(1)}_2$ the definition of $I$ implies that $x$ and $y$ must agree on $\bigcup (R^{(2)}_\Gamma)_2$ as well.
        \end{enumerate}
        
        Since $\maxdeg(\calG(y))\leq d$, $y_{a_1}$ can be part of at most $2(d-1)$ wedges. Therefore, the set $I$ defined above has size $\abs{I} \geq 2d - 2(d-1) = 2>0$. Therefore, by \cref{X_block_form_II}, the total contribution of stage II.2 to \cref{eq:feasibility} is
        \begin{equation}
            \sum_{\emptyset\neq \Gamma \subseteq I} X^{R^{(2)}_\Gamma,R^{(3)}_\Gamma}_{a_2}[x,y] = \sum_{\emptyset\neq \Gamma \subseteq I} (-1)^{1+\abs{\Gamma}}q = q.
        \end{equation}
        \item 
        If $k=3$, arcs in $\Act(x,R^{(1)})$ in stage II.3 are of the form $A^{R_\Gamma^{(3)}, R_\Gamma^{(4)}}_{a_3}$ where $R_\Gamma^{(3)}\in V^{(3)}(\Gamma)$ and $R_\Gamma^{(4)}\in V^{(4)}(\Gamma)$, one for each $\emptyset\neq \Gamma\subseteq[2d]$.
        
        For each arc $A^{R_\Gamma^{(2)}, R_\Gamma^{(3)}}_{a_2}$ in stage II.2, since $x_{a_2} = y_{a_2}$, and loading $a_2$ does not uncover additional elements by the definition \cref{eq:alpha_def}, it follows that $x$ and $y$ agree on $R^{(3)}_\Gamma$ if and only if $x$ and $y$ agree on $R^{(2)}_\Gamma$.
        
        Therefore, using the same $I$ as defined in the previous case, the total contribution of stage II.3 to \cref{eq:feasibility} is
        \begin{equation}
            \sum_{\emptyset\neq \Gamma \subseteq I} X^{R^{(3)}_\Gamma,R^{(4)}_\Gamma}_{a_3}[x,y] = \sum_{\emptyset\neq \Gamma \subseteq I} (-1)^{1+\abs{\Gamma}}q = q.
        \end{equation}
    \end{itemize}
\end{itemize}

\section{Cycle Finding}
In this section, we study the $\kcycle{k}$ problem, which is a natural generalization of $\tri$. We give a nearly tight characterization of the quantum query complexity of $\kcycle{k}$ for graphs with low maximum degree such as random sparse graphs. We note that the proof of \Cref{prop:two_way_reduct} can be straightforwardly extended to any $k \geq 3$ to show that $\kcycle{k}$ is harder than $\kdist{k}$ but easier than $\ksum{k}$.   

\subsection{Cycle finding lower bound}

The most general form of our lower bound reads:
\begin{theorem}\label{thm:kcycle_avg_lower}
    Let $k,m\in \mathbb{N}$ with $k\geq 3$ and let $\epsilon>0$. Then, for all sufficiently large $n\in \mathbb{N}$ and every $T, \Delta\in \mathbb{N}$ with $1\leq T \leq n^{1-\epsilon}$ and $1\leq \Delta \leq n^{o(1)}$, the following holds for every $T$-query quantum algorithm $\calA_T$:
    \begin{equation}
    \begin{aligned}
        &\Pr[\calA_T(x) = (i_1,\dots,i_k)\in [m]^k \textup{ such that $x_{i_1},\dots,x_{i_k}$ form a $k$-cycle}]
        \\
        &\qquad\qquad\leq O\Bigl(\Delta^k\cdot
        \Bigl(\frac{T}{n^{3/4-1/(2^{k+2}-4)}}
       \Bigr)^{2-2/2^k} + \Delta^{k/2}\cdot T\frac{\sqrt{\log n}}{n}\\
       &\qquad\qquad\qquad\qquad+ Tn\Bigl(\frac{2em}{n\ceil{\log(n)}}\Bigr)^{\ceil{\log(n)}}+Tn\left(\frac{2em}{n\Delta}\right)^{\Delta}\Bigl(\frac{emn}{2\ceil{\log(n)}}\Bigr)^{\ceil{\log(n)}}\Bigr),\label{eq:kcycle_avg_lower}
    \end{aligned}
    \end{equation}
    where $\calA_T(x)$ denotes the output of $\calA_T$ when its queries are made to $x$, and the probability is over $x\leftarrow \binom{[n]}{2}^m$ and the randomness of $\calA_T$.
\end{theorem}

From \cref{thm:kcycle_avg_lower} and the standard fact that random graphs with $m=\Omega(n)$ edges contain at least one $k$-cycle with constant  probability, as formally stated in~\cref{fact:sparse-random-graph-restated}, we deduce that the search problem of finding $k$-cycle in a random graph is hard in the average case for $m=\Theta(n)$. It follows that the decision version of $k$-cycle is also hard on average, based on the search-to-decision reduction in the edge list model discussed in the preliminaries section. Formally, \cref{thm:kcycle_avg_lower} gives
\begin{corollary}\label{cor:kcycle_lower_bound}
    Let $k\in \mathbb{N}_{\geq 3}$. For all $m\in \mathbb{N}$ and $d = 2\log(m)/\log\log(m)$,
    \begin{equation}
        Q(\kcycle{k}_{m}) \geq  Q(\kcycle{k}_{m,d}) \geq \widetilde{\Omega}(m^{(3/4-1/(2^{k+2}-4))}).
    \end{equation}
\end{corollary}
\begin{proof}[Proof of \cref{cor:kcycle_lower_bound}]
    The first inequality follows by restriction so it suffices to prove the second inequality. Let $n\coloneqq m$. When $x$ is chosen uniformly at random from $\binom{[n]}{2}^m$, \cref{fact:sparse-random-graph-restated} gives
    \begin{equation}
         \Pr[\textup{$x$ contains a $k$-cycle and $x$ has $\maxdeg(x) \leq 2\log(m)/\log\log(m)$}] \geq \Omega(1).
    \end{equation}
    Therefore, every bounded-error quantum query algorithm for $\kcycle{k}_{m,d}$ can be used to find a $k$-cycle in $x$ with probability at least $\Omega(1)$, where the probability is over both the algorithm and the randomness in $x$. By the search-to-decision reduction (see preliminaries section), the quantum query complexity of the search algorithm is at most $Q(\kcycle{k}_{m,d})$ times $O(\log^2(m))$.

    Suppose for contradiction that $ Q(\kcycle{k}_{m,d}) \leq o(m^{(3/4-1/(2^{k+2}-4))}/\log^{2k/(2-2/2^k)+2}(m))$. Then, by the above remark, there exists a $T$-query quantum algorithm with 
    \begin{equation}\label{eq:kcycle_high_success}
        \Pr[\calA_T(x) = (i_1,\dots,i_k)\in [m]^k \textup{ such that $x_{i_1},\dots,x_{i_k}$ form a $k$-cycle}] \geq \Omega(1),
    \end{equation}
    where $T\leq o(m^{(3/4-1/(2^{k+2}-4))}/\log^{2k/(2-2/2^k)}(m))$.

    On the other hand, \cref{thm:kcycle_avg_lower} with $\Delta$ in its statement set equal to $\ceil{\log(n)}^2$ gives 
    \begin{equation}
    \begin{aligned}
        &\Pr[\calA_T(x) = (i_1,\dots,i_k)\in [m]^k \textup{ such that $x_{i_1},\dots,x_{i_k}$ form a $k$-cycle}]\\
        &\qquad\qquad \leq O\Bigl(\log^{2k}(n)\cdot
        \Bigl(\frac{T}{n^{3/4-1/(2^{k+2}-4)}}
       \Bigr)^{2-2/2^k} + \log^k(n)\cdot T \frac{\sqrt{\log n}}{n}\\
       &\qquad\qquad\qquad\qquad + Tn \Bigl(\frac{2e}{\ceil{\log n}}\Bigr)^{\ceil{\log(n)}}+Tn\Bigl(\frac{2^{\ceil{\log(n)}}e^{\ceil{\log(n)}+1}n^2}{2\ceil{\log(n)}^{2\ceil{\log(n)} + 1}}\Bigr)^{\ceil{\log(n)}}\Bigr).\label{eq:kcycle_low_success}
    \end{aligned}
    \end{equation}
    But recall $n=m$, so if $T\leq o(m^{(3/4-1/(2^{k+2}-4))}/\log^{2k/(2-2/2^k)}(m))$, each summand in the expression appearing in \cref{eq:kcycle_low_success} is at most $o(1)$. This contradicts \cref{eq:kcycle_high_success}. Hence the corollary.
\end{proof}

In the subsequent subsections, we prove \cref{thm:kcycle_avg_lower}.

\subsubsection{Setting up}
Let $k,m,\epsilon,n,T,\Delta$ be as given in the statement of \cref{thm:kcycle_avg_lower}. For convenience, we write 
\begin{equation}\label{eq:constant_factor}
N\coloneqq \scalebox{1.25}{$\binom{n}{2}$} \quad \text{and} \quad c \coloneqq 32e \cdot k\Delta^{k/2}.
\end{equation}
We also identify the symbol $[N]$ with $\binom{[n]}{2}$ for convenience. 

We follow the overall proof strategy as in the triangle case, except there are now more sub-graph structures to keep track of, including length-$l$ paths for $l=2,3,\ldots, k-1$, to measure progress in the algorithm.  We will argue that to find a $k$-cycle, a lot of queries should be spent to find many $(k-1)$-paths first, which in turn requires a large investment of queries to find many $(k-2)$-paths and so on. Thus, we define the following projectors to measure the progress of recording edges, paths, and cycles, etc.

\begin{definition}[Recording projectors]
\label{def:recording_projectors-cycle} 
    For  \(R\subseteq\mathbb{R}\), $l\in \mathbb{N}_{\geq 2}$, and \(y\in \binom{[n]}{2}\cup \{\bot\}\), we define the following projectors by giving the basis states onto which they project.\\[5pt]
\begin{tabular}{@{}l@{\ }l}
$\Pi^\square$                  &\,: every basis state $\ket{i,u,w}\ket{x}$ such that $\calG(x)$ contains a $k$-cycle. \\[4pt]
$\Pi^{\{1\}}_{R}$               &\,: every basis state $\ket{i,u,w}\ket{x}$ such that the number of edges in $\calG(x)$ is in $R$ \emph{or} $\calG(x)$ \\
&\hspace{6pt} contains a vertex of degree at least $\Delta$. \\[4pt]
$\Pi^{\{l\}}_{R}$  &\,: every basis state $\ket{i,u,w}\ket{x}$ such that the number of $l$-paths in $\calG(x)$ is in $R$. \\[4pt]
$\Pi_{y}$                       &\,: every basis state $\ket{i,u,w}\ket{x}$ such that $u\neq 0$ and $x_i = y$.
\end{tabular}
\end{definition}

We define formally $\cycleprogress_t$ to measure the progress of recording a $k$-cycle in $t$ queries as below
\begin{definition}(Progress measure for recording cycle) For $t\in \mathbb{Z}_{[0,T]}$, define
\begin{equation}
    \cycleprogress_t \coloneqq \norm{\Pi^\square \ket{\phi_t}}.
\end{equation}
\end{definition}

In the triangle case, a technical obstacle is that high-degree vertices can distort our analysis by overestimating the increase in the number of new wedges recorded per query. To address this, we developed the Exclusion Lemma to filter out high-degree vertices. In the $k$-cycle case, an unexpectedly large number of $(l-1)$-paths becomes a new source of overestimating the increase of $l$-paths recorded per query. This issue, in principle, can be addressed by excluding these unexpected cases; which, however, introduces some extra technical nuances. In particular, we define the following family of functions $r_j:\mathbb N\to\mathbb N$, for $j\in \mathbb N$,
    \begin{equation}\label{eq:rjt_definition}
        r_j(t) \coloneqq \max\left\{\biggl\lceil \Bigl(\frac{ct}{\sqrt{n}}\Bigr)^{2-4/2^j} (t+1)^{2/2^j}\biggr\rceil, \ceil{\log n}\right\}\cdot(j\Delta^{j-1}).
    \end{equation}
A routine calculation reveals the following analytical property of $r_j(t)$.
\begin{claim}\label{claim:inductive_step_technical}
    Let $l\in \mathbb{N}_{\geq 2}$ and $j \in \{1,\dots, \floor{\frac{l-1}{2}}\}$. Then, for all sufficiently large $n\in \mathbb{N}$ and every $t\in \mathbb{Z}_{[0,T]}$,
    \begin{equation}
        r_j(t)r_{l-1-j}(t) \leq r_{l-1}(t) n.
    \end{equation}
\end{claim}
\begin{proof}
By definition, $r_j(t)$ and $r_{l-1-j}(t)$ are each a max of two terms. Since $\max(a,b)\cdot \max(c,d)\leq \max(ac,ad,bc,bd)$ for all $a,b,c,d\in \mathbb{R}$,  it suffices to show that each of the following four terms is at most $r_{l-1}(t)n$.
\begin{enumerate}
    \item Term $\biggl\lceil \Bigl(\frac{ct}{\sqrt{n}}\Bigr)^{2-4/2^j} (t+1)^{2/2^j}\biggr\rceil \cdot \biggl\lceil \Bigl(\frac{ct}{\sqrt{n}}\Bigr)^{2-4/2^{l-1-j}} (t+1)^{2/2^{l-1-j}}\biggr\rceil \cdot(l^2\Delta^l)$. This term is at most $r_{l-1}(t)n$ because the term divided by $r_{l-1}(t)n$ is at most
    \begin{align*}
        &\frac{\biggl\lceil\Bigl(\frac{ct}{\sqrt{n}}\Bigr)^{2-4/2^j} (t+1)^{2/2^j}\biggr\rceil \cdot \biggl\lceil \Bigl(\frac{ct}{\sqrt{n}}\Bigr)^{2-4/2^{l-1-j}} (t+1)^{2/2^{l-1-j}}\biggr\rceil \cdot(l^2\Delta^l)}{\biggl\lceil\Bigl(\frac{ct}{\sqrt{n}}\Bigr)^{2-4/2^{l-1}} (t+1)^{2/2^{l-1}}\biggr\rceil\cdot n}
        \\
        &\qquad\qquad\leq O\Bigl(\poly(\Delta)\cdot \Bigl(\frac{t}{n}\Bigr)^{2-2/2^j - 2/2^{l-1-j} +2/2^{l-1}}\Bigr) \leq 1,
    \end{align*}
    where the last inequality holds for sufficiently large $n$ and uses $t\leq T \leq n^{1-\epsilon}$ and $\Delta \leq n^{o(1)}$. 
    \item Term $\biggl\lceil \Bigl(\frac{ct}{\sqrt{n}}\Bigr)^{2-4/2^j} (t+1)^{2/2^j}\biggr\rceil \cdot \ceil{\log n} \cdot(l^2\Delta^l)$. Since $t\leq T \leq n^{1-\epsilon}$, this term is at most $O(n^{1-\Omega(1)}\, \poly(\log n))$, which is at most $r_{l-1}(t)n$ because $r_{l-1}(t)n \geq \Omega(n \log n)$.
    \item Term $\biggl\lceil \Bigl(\frac{ct}{\sqrt{n}}\Bigr)^{2-4/2^{l-1-j}} (t+1)^{2/2^{l-1-j}}\biggr\rceil \cdot \ceil{\log n} \cdot(l^2\Delta^l)$. The  previous argument applies.
    \item Term $\ceil{\log n}^2 \cdot(l^2\Delta^l)$. This is at most $r_{l-1}(t)n$ because $r_{l-1}(t)n \geq \Omega(n \log n)$ and $\Delta \leq n^{o(1)}$.\qedhere
\end{enumerate}
\end{proof}
The mysterious-looking family of functions $r_j(t)$ corresponds to upper bounds on the number of $j$-paths that we can record with $t$ queries, and is calculated from an optimization that only becomes clear after the analysis is carried out: we will show that it is difficult to record more than $r_j(t)$ number of $j$-paths with $t$ queries. The above discussion motivates the following excluding projectors.

\begin{definition}[Excluding projectors]\label{def:excluding_projectors-cycle}
We define the following excluding projectors.

\begin{itemize}
    \item For $r\geq0$, let $\Pi^{\deg}_{\geq r}$ be the projector onto every basis state $\ket{i,u,w}\ket{x}$ such that there exists a vertex in $\calG(x)$ with degree at least $r$.

    \item Let $t\in \mathbb{Z}_{[0,T]}$. For $l\in \mathbb{N}$, define $\Pi_{\Bad,t}^{(l)}$ to be the projector onto every basis state $\ket{i,u,w}\ket{x}$ such that $\calG(x)$ contains a vertex of degree at least $\Delta$ \emph{or} there exists an integer $1\leq j\leq l$ such that $\calG(x)$ contains at least $r_{j}(t)$ length-$j$ paths. For convenience, we also define $\Pi_{\Bad,t}^{(0)}\coloneqq 0$.
\end{itemize}
\end{definition}
$\Pi_{\ge r}^{\deg} $ is the same operator we used to exclude high-degree vertices. 
By the same argument as in the triangle case we can exclude high-degree graphs with the degree threshold set to $\Delta$.

$\Pi^{(l)}_{\Bad, t}$ is the new operator that in addition excludes the case that there are unexpectedly many (i.e. at least $r_j(t)$) $j$-paths for some $1\le j \le l$ with $t$ queries.
By our definition, for every $t\in \mathbb{Z}_{[0,T]}$ and $l\in \mathbb{N}$,
\begin{equation}\label{eq:bad_recurrence}
    \I - \Pi^{(l)}_{\Bad, t} = \bigl(\I - \Pi^{\{l\}}_{[r_l(t),\infty)}\bigr) \bigl(\I - \Pi^{(l-1)}_{\Bad, t}\bigr).
\end{equation}

Recall that, for $t\in \mathbb{Z}_{[0,T]}$, 
\begin{equation}
    \ket{\phi_t} \coloneqq U_t \calR U_{t-1}\calR \dots U_1 \calR U_0 (\ket{0}\ket{\bot^m}).
\end{equation}

For $t\in \mathbb{Z}_{[0,T]}$ and $l\in \mathbb{Z}_{\geq 0}$, let
\begin{equation}\label{eq:phi_i^l}
    \ket{\phi_t^{(l)}} \coloneqq U_t \calR(\I-\Pi_{\Bad,t-1}^{(l)}) U_{t-1} \calR (\I-\Pi_{\Bad,t-2}^{(l)}) \dots (\I-\Pi_{\Bad,1}^{(l)})U_1 \calR U_0 (\ket{0}\ket{\bot^m}).
\end{equation}
The above definition implies:
\begin{equation}\label{eq:phi_i^l_properties}
\begin{aligned}
    &\ket{\phi_t^{(0)}} = \ket{\phi_t} \quad \text{and} 
    \\
    &\ket{\phi_{t+1}^{(l)}} =U_{t+1} \calR (\I - \Pi^{(l)}_{\Bad,t})\ket{\phi_t^{(l)}} \quad \text{for all integers $0\leq t \leq T-1$.}
\end{aligned}
\end{equation}

Now, to measure progress in recording length-$l$ paths, we use the following somewhat technical but convenient definition, that excludes unlikely events in advance.
\begin{definition}[Progress measures for recording paths]
    For $t\in \mathbb{Z}_{[0,T]}$, $l\in \mathbb{N}$, and $r\geq 0$, define
    \begin{equation}\label{eq:P^l_{i,p}}
    \pathsprogress^{\{l\}}_{t,r} \coloneqq \norm{\Pi_{[r,\infty)}^{\{l\}} \ket{\phi_t^{(l-1)}}}.
\end{equation}
\end{definition}

In the next two subsections, we analyze our progress measures $\pathsprogress_{t,r}^{\{l\}}$ and $\cycleprogress_t$ in detail.

\subsubsection{Progress in recording paths}
In this part, we analyze the progress of recording $l$-paths. We start with the base case $l=1$, the progress of recording $1$-paths, i.e., edges, \emph{or} a vertex of degree at least $\Delta$. Recall that the definition of $\Pi^{\{1\}}_R$ involves a degree condition that is not present in the definitions of $\Pi^{\{l\}}_R$ for $l\in \mathbb{N}_{\geq 2}$. The degree condition is crucial for carrying out the induction. 
\begin{proposition}\label{prop:p_base_bound}
    If $n\geq 2$ and $t\in \mathbb{Z}_{[0,T]}$,
    \begin{equation}
        \pathsprogress^{\{1\}}_{t,{r_1(t)}} \leq 2n  \Bigl(\frac{2em}{n\ceil{\log(n)}}\Bigr)^{\ceil{\log(n)}}+2n\left(\frac{2em}{n\Delta}\right)^{\Delta}\Bigl(\frac{emn}{2\ceil{\log(n)}}\Bigr)^{\ceil{\log(n)}}.
    \end{equation}
\end{proposition}
\begin{proof}
    First, we recall the definitions
    \begin{equation}
        r_1(t) \coloneqq \max\left\{t+1, \ceil{\log n}\right\} \quad \text{and} \quad  \pathsprogress^{\{1\}}_{t,{r_1(t)}} = \norm{\Pi_{[r_1(t),\infty)}^{\{1\}} \ket{\phi_t}},
    \end{equation}
    where $\Pi_{[r_1(t),\infty)}^{\{1\}}$ is the projector onto every basis state $\ket{i,u,w}\ket{x}$ such that $\calG(x)$ contains at least $r_1(t)\geq t+1$ edges \emph{or} a vertex of degree at least $\Delta$.

    By~\cref{fact:t_non_bots}, $\ket{\phi_t}$ is only supported on the basis states $\ket{i,u,w}\ket{x}$ such that $\calG(x)$ contains at most $t$ edges. Therefore, 
    \begin{equation}
        \norm{\Pi_{[r_1(t),\infty)}^{\{1\}} \ket{\phi_t}} = \norm{\Pi^{\deg}_{\geq \Delta} \ket{\phi_t}} \leq 2n  \Bigl(\frac{2em}{n\ceil{\log(n)}}\Bigr)^{\ceil{\log(n)}}+2n\left(\frac{2em}{n\Delta}\right)^{\Delta}\Bigl(\frac{emn}{2\ceil{\log(n)}}\Bigr)^{\ceil{\log(n)}},
    \end{equation}
    where we used~\cref{lem:cycle-degree-event-bound} for the inequality, as required.
\end{proof}

Next, we analyze the progress recording $l$-paths for $l\ge 2$, relying on  the  following recurrence.

\begin{lemma}[$l$-paths progress recurrence]\label{lem:inductive_step_p_lpaths_bound}
    Let $l\in \mathbb{N}_{\geq 2}$. For all $t\in \mathbb{Z}_{[0,T]}$ and $r\geq 0$, we have
    \begin{equation}\label{eq:paths_recurrence}
            \pathsprogress^{\{l\}}_{t+1,r} \leq \pathsprogress_{t,r}^{\{l\}} + 4 \sqrt{\frac{B_{l}(t)}{N}} \pathsprogress^{\{l\}}_{t,r-l\Delta^{l-1}},
    \end{equation}
where  
    \begin{equation}\label{eq:blt}
        B_l(t) \coloneqq 2r_{l-1}(t) n + 4\sum_{j=1}^{\floor{\frac{l-1}{2}}} r_j(t)r_{l-1-j}(t).
    \end{equation}
\end{lemma}

\begin{proof}[Proof of \cref{lem:inductive_step_p_lpaths_bound}]
Let $l \in \mathbb{N}_{\geq 2}$. For $t\in \mathbb{Z}_{[0,T]}$ and $r\geq0$, 
\begin{align}
     \pathsprogress^{\{l\}}_{t+1,r} =\;& \norm{\Pi_{[r,\infty)}^{\{l\}} \ket{\phi_{t+1}^{(l-1)}}} &&\text{definition \cref{eq:P^l_{i,p}}}\notag
     \\
     =\;&\norm{\Pi_{[r,\infty)}^{\{l\}}U_{t+1}\calR\bigl(\I-\Pi^{(l-1)}_{\Bad,t}\bigr)\ket{\phi_t^{(l-1)}}} &&\text{\cref{eq:phi_i^l_properties}}\notag
     \\
     =\;&\norm{\Pi_{[r,\infty)}^{\{l\}}\calR\bigl(\I-\Pi^{(l-1)}_{\Bad,t}\bigr)\ket{\phi_t^{(l-1)}}} && \text{$\Pi_{[r,\infty)}^{\{l\}}$, $U_{t+1}$ commute}\notag
     \\
     \leq\;& \norm{\Pi_{[r,\infty)}^{\{l\}}\ket{\phi_t^{(l-1)}}}+ \norm{\Pi_{[r,\infty)}^{\{l\}} \calR \bigl(\I-\Pi^{(l-1)}_{\Bad,t}\bigr)\bigl(\I-\Pi_{[r,\infty)}^{\{l\}}\bigr)\ket{\phi_t^{(l-1)}}} &&\text{triangle inequality}\notag
     \\
     =\;& \pathsprogress_{t,r}^{\{l\}} + \norm{\Pi_{[r,\infty)}^{\{l\}} \calR \bigl(\I-\Pi^{(l-1)}_{\Bad,t}\bigr)\bigl(\I-\Pi_{[r,\infty)}^{\{l\}}\bigr)\ket{\phi_t^{(l-1)}}}&&\text{definition \cref{eq:P^l_{i,p}}}\notag
     \\
     =\;& \pathsprogress_{t,r}^{\{l\}} + \norm{\Pi_{[r,\infty)}^{\{l\}} \calR \bigl(\I-\Pi^{(l-1)}_{\Bad,t}\bigr)\Pi_{[r-l\Delta^{l-1},r)}^{\{l\}}\ket{\phi_t^{(l-1)}}}, \label{eq:paths_recurrence_terms}
\end{align}
where the last equality uses the fact that adding a new edge of a graph of maximum degree at most $\Delta$ can introduce at most
$\sum_{i=0}^{l-1} \Delta^{l-1-i}\Delta^i = l \Delta^{l-1}$ new length-$l$ paths.

Abbreviate
\begin{equation}
    \ket{\rho} \coloneqq \bigl(\I-\Pi^{(l-1)}_{\Bad,t}\bigr)\Pi_{[r-l\Delta^{l-1},r)}^{\{l\}}\ket{\phi_{t}^{(l-1)}}. 
\end{equation}
Then, we can write the second term in \cref{eq:paths_recurrence_terms} as
\begin{equation}
    \norm{\Pi_{[r,\infty)}^{\{l\}} \calR \ket{\rho}} \le \sum_{y\in [N] \cup \{\bot\}}   \norm{\Pi_{[r,\infty)}^{\{l\}} \calR \Pi_y \ket{\rho}},
\end{equation}
since $\id = \ketbra{0}{0}_P + \sum_{y\in[N]\cup\{\perp\}}\Pi_y$, and $\Pi_{[r,\infty)}^{\{l\}} \calR \ketbrasame{0}_P \ket{\rho} = 0$ because $\ket{\rho}$ has recorded fewer than $r$ length-$l$ paths and $\calR\ketbrasame{0}_P = \ketbrasame{0}_P$.
    
We will show
\begin{equation}\label{eq:path_progress_bound}
    \norm{\Pi_{[r,\infty)}^{\{l\}} \calR \Pi_{\bot}\ket{\rho}} \leq \sqrt{\frac{B_{l}(t)}{N}}\pathsprogress^{\{l\}}_{t,r-l\Delta^{l-1}} \quad \text{and} \quad
    \sum_{y\in [N]}\norm{\Pi_{[r,\infty)}^{\{l\}} \calR \Pi_y\ket{\rho}} \leq 3\sqrt{\frac{B_{l}(t)}{N}}\pathsprogress^{\{l\}}_{t,r-l\Delta^{l-1}}.
\end{equation}

\paragraph{First bound in \cref{eq:path_progress_bound}.} 
        When $y=\bot$, consider each basis state $\ket{i,u,w}\ket{x}$ in the support of
        \begin{equation}
            \Pi_\bot\bigl(\I-\Pi^{(l-1)}_{\Bad,t}\bigr)\Pi_{[r-l\Delta^{l-1},r)}^{\{l\}}\ket{\phi^{(l-1)}_t}.
        \end{equation}
         Due to $\Pi_\bot$, we have $x_i = \bot$ and $u\neq 0$, so \cref{lemma:recording_operator_effect} states
        \begin{equation}
            \calR\ket{i,u,w}\ket{x} = \ket{i,u,w}\biggl(\sum_{y\in [N]}\frac{\omega_N^{u y}}{\sqrt{N}}\ket{y}_{I_i}\biggr)\otimes\bigotimes_{j\neq i}\ket{x_{j}}_{I_{j}}.
        \end{equation}
        Since $\ket{i,u,w}\ket{x}$ is also in the support of $(\I-\Pi^{(l-1)}_{\Bad,t})$, $x$ records at most $r_j(t)$ length-$j$ paths for every $1\leq j \leq l-1$. Due to $\Pi_{[r-l\Delta^{l-1},r)}^{\{l\}}$, $x$ records fewer than $r$ length-$l$ paths.  Therefore, in order for $\ket{y}_{I_i} \otimes \bigotimes_{j\neq i}\ket{x_{j}}_{I_{j}}$ to record at least $r$ length-$l$ paths, edge $y$ must join one end of a length-$j$ path with one end of a length-$(l-1-j)$ path for some $0\leq j\leq \floor{\frac{l-1}{2}}$, where a length-$0$ path refers to a vertex. When $j=0$, the number of such $y$s is at most $2 r_{l-1}(t) n$, where the factor of $2$ accounts for the two ends of the length-$(l-1)$ path. When $j>0$, the number of such $y$s is at most $4r_j(t) r_{l-1-j}(t)$, where the factor of $4$ accounts for the four ways to connect a length-$j$ path and a length-$(l-1-j)$ path. Therefore,
        \begin{equation}
        \norm{\Pi_{[r,\infty)}^{\{l\}}\calR\ket{i,u,w}\ket{x}}\leq\sqrt{\frac{B_{l}(t)}{N}}.
        \end{equation}
        Since any two distinct basis states in the support of $\Pi_\bot\bigl(\I-\Pi^{(l-1)}_{\Bad,t}\bigr)\Pi_{[r-l\Delta^{l-1},r)}^{\{l\}}\ket{\phi^{(l-1)}_t}$ remain orthogonal after $\Pi^{\{l\}}_{[r,\infty)}\calR$ is applied, we have
        \begin{equation}
        \begin{aligned}            
            & \norm{\Pi_{[r,\infty)}^{\{l\}} \calR \Pi_{\bot}\bigl(\I-\Pi^{(l-1)}_{\Bad,t}\bigr)\Pi_{[r-l\Delta^{l-1},r)}^{\{l\}}\ket{\phi_t^{(l-1)}}}
            \\
            &\qquad \leq \sqrt{\frac{B_{l}(t)}{N}}\norm{\Pi_{\bot}\bigl(\I-\Pi^{(l-1)}_{\Bad,t}\bigr)\Pi_{[r-l\Delta^{l-1},r)}^{\{l\}}\ket{\phi_t^{(l-1)}}} \leq \sqrt{\frac{B_l(t)}{N}}\pathsprogress^{\{l\}}_{t,r-l\Delta^{l-1}}.
        \end{aligned}
        \end{equation}
        \paragraph{Second bound in \cref{eq:path_progress_bound}.} 

        When $y\in [N]$, consider each basis state $\ket{i,u,w}\ket{x}$ in the support of
        \begin{equation}
            \Pi_y\bigl(\I-\Pi^{(l-1)}_{\Bad,t}\bigr)\Pi_{[r-l\Delta^{l-1},r)}^{\{l\}}\ket{\phi^{(l-1)}_t}.
        \end{equation}
        Due to $\Pi_y$, we have $x_i = y$ and $u\neq 0$, so \cref{lemma:recording_operator_effect} states
        \begin{equation}
        \begin{aligned}
            \calR\ket{i,u,w}\ket{x}
            =\ket{i,u,w}&\biggl(\frac{\omega^{u y}_N}{\sqrt{N}}\ket{\bot}_{I_i}+\frac{1+\omega^{u y}_N(N-2)}{N}\ket{y}_{I_i}\\
            &\qquad\qquad+\sum_{y'\in [N]\backslash\{y\}}\frac{1-\omega^{u y'}_N-\omega^{u y}_N}{N}\ket{y'}_{I_i}\biggr)\otimes\bigotimes_{j\neq i}\ket{x_{j}}_{I_{j}}.
        \end{aligned}
        \end{equation}
        Applying similar reasoning as in the proof of the first bound in \cref{eq:path_progress_bound}, we deduce
        \begin{equation}
            \norm{\Pi^{\{l\}}_{[r,\infty)}\calR\ket{i,u,w}\ket{x}}\leq 3\frac{\sqrt{B_{l}(t)}}{N}.
        \end{equation}
        Since any two distinct basis states in the support of $\Pi_y(\I-\Pi^{(l-1)}_{\Bad,t})\Pi_{[r-l\Delta^{l-1},r)}^{\{l\}}\ket{\phi^{(l-1)}_t}$ remain orthogonal after $\Pi^{\{l\}}_{[r,\infty)}\calR$ is applied, we have
        \begin{equation}
            \norm{\Pi_{[r,\infty)}^{\{l\}} \calR \Pi_{y}\bigl(\I-\Pi^{(l-1)}_{\Bad,t}\bigr)\Pi_{[r-l\Delta^{l-1},r)}^{\{l\}}\ket{\phi_t^{(l-1)}}} \leq 3\frac{\sqrt{B_l(t)}}{N}\norm{\Pi_{y}\bigl(\I-\Pi^{(l-1)}_{\Bad,t}\bigr)\Pi_{[r-l\Delta^{l-1},r)}^{\{l\}}\ket{\phi_t^{(l-1)}}}.
        \end{equation}
        Therefore, the Cauchy-Schwarz inequality gives
        \begin{align}
            &\sum_{y\in [N]} \norm{\Pi_{[r,\infty)}^{\{l\}} \calR \Pi_{y}\bigl(\I-\Pi^{(l-1)}_{\Bad,t}\bigr)\Pi_{[r-l\Delta^{l-1},r)}^{\{l\}}\ket{\phi_t^{(l-1)}}}
            \nonumber\\
            &\qquad \le 3\sqrt{\frac{B_{l}(t)}{N}}\sqrt{\sum_{y\in [N]}\norm{\Pi_{y}\bigl(\I-\Pi^{(l-1)}_{\Bad,t}\bigr)\Pi_{[r-l\Delta^{l-1},r)}^{\{l\}}\ket{\phi_t^{(l-1)}}}^2}
            \leq 3\sqrt{\frac{B_{l}(t)}{N}}\pathsprogress^{\{l\}}_{t,r-l\Delta^{l-1}}.\qedhere
    \end{align}
\end{proof}

Solving the recurrence, we deduce
\begin{lemma}[$l$-paths progress]\label{prop:p_lpaths_bound}
Let $l \in \mathbb{N}_{\geq 2}$. For all sufficiently large $n\in \mathbb{N}$ and every $t\in \mathbb{Z}_{[0,T]}$,
\begin{equation}
    \pathsprogress^{\{l\}}_{t,{r_l(t)}} \leq \frac{1}{n}.
\end{equation}
\end{lemma}

\begin{proof}[Proof of \cref{prop:p_lpaths_bound}]
We start by solving the recurrence \cref{eq:paths_recurrence} in \cref{lem:inductive_step_p_lpaths_bound} to obtain
    \begin{equation}
        \pathsprogress_{t,r \cdot l\Delta^{l-1}}^{\{l\}} \leq \Bigl(\frac{et}{r}\cdot 4\sqrt{\frac{B_l(t)}{N}}\Bigr)^{r}.
    \end{equation}
    Therefore, setting $r$ equal to the non-negative integer 
    \begin{equation}
        s_{l}(t) \coloneqq r_{l}(t)/(l\Delta^{l-1}) = \max\left\{\biggl\lceil\Bigl(\frac{ct}{\sqrt{n}}\Bigr)^{2-4/2^l} (t+1)^{2/2^l}\biggr\rceil, \ceil{\log n}\right\}
    \end{equation}
    gives
    \begin{align*}
        \pathsprogress_{t,r_l(t)}^{\{l\}} &\leq \biggl(\frac{et}{s_l(t)}\cdot 4\sqrt{\frac{B_l(t)}{N}}\biggr)^{s_{l}(t)} 
        \\
        &= \biggl(\frac{et}{s_l(t)}\cdot 4\sqrt{\frac{2l r_{l-1}(t)n}{N}}\biggr)^{s_{l}(t)} &&\text{\cref{claim:inductive_step_technical},~\cref{eq:blt}}
        \\
        &\leq \biggl(\frac{8\sqrt{2}et}{s_l(t)}\cdot \sqrt{\frac{l r_{l-1}(t)}{n}}\biggr)^{s_{l}(t)} &&\text{$N \geq n^2/4$}.
    \end{align*}
    Now, we recall the definition
    \begin{equation}
        r_{l-1}(t) \coloneqq \max\left\{\lceil \alpha \rceil, \ceil{\log n}\right\}\cdot(l-1)\Delta^{l-2},\quad \text{where} \quad \alpha \coloneqq \Bigl(\frac{ct}{\sqrt{n}}\Bigr)^{2-4/2^{l-1}} (t+1)^{2/2^{l-1}}.  
    \end{equation}

    There are two cases depending on where the max is attained in the expression $\max\left\{\lceil \alpha \rceil, \ceil{\log n}\right\}$.

    \underline{Case 1}: $\ceil{\log n}\ge \ceil{\alpha}$.
    
    In this case, 
        \begin{equation}
           \ceil{\log n}\geq \biggl\lceil \Bigl(\frac{ct}{\sqrt{n}}\Bigr)^{2-4/2^{l-1}} (t+1)^{2/2^{l-1}}\biggr\rceil \geq \Omega\Bigl(n \Bigl(\frac{t}{n}\Bigr)^{2-4/2^l}\Bigr),
        \end{equation}
        rearranging gives
        \begin{equation}
            t \leq \Bigl(n \cdot \Bigl(\frac{O(\log(n))}{n}\Bigr)^{1/(2-4/2^l)}\Bigr) \leq n^{1/2 - \Omega(1)}.
        \end{equation}
        Therefore, for sufficiently large $n$, we have
        \begin{equation}
            \pathsprogress_{t,r_l(t)}^{\{l\}} \leq \Bigl(\frac{8\sqrt{2}et}{s_l(t)}\cdot \sqrt{\frac{l r_{l-1}(t)}{n}}\Big)^{s_{l}(t)} \leq O(n^{-\Omega(\log(n))})\leq \frac{1}{n},
        \end{equation}

        \underline{Case 2}: $\ceil{\alpha}\ge\ceil{\log n}$.
        Note that $\alpha\geq 1$ since we assume $n$ is sufficiently large. Then
        \begin{equation}
             \pathsprogress_{t,r_l(t)}^{\{l\}}\leq \Bigl(\frac{8\sqrt{2}et}{s_l(t)}\cdot \sqrt{\frac{l r_{l-1}(t)}{n}}\Bigr)^{s_{l}(t)} \leq \Bigl(\frac{16e\sqrt{l(l-1)\Delta^{l-2}}}{c}\Bigr)^{s_l(t)} \leq 2^{-s_{l}(t)} \leq \frac{1}{n},
        \end{equation}
        where we used $r_{l-1}(t) = \ceil{\alpha} (l-1)\Delta^{l-2} \leq 2 \alpha (l-1)\Delta^{l-2}$, which holds since $\alpha\geq 1$, in the second inequality, and recall we define $c = 32e k\Delta^{k/2}$.
\end{proof}

\subsubsection{Progress in recording a $k$-cycle}

For $t\in \mathbb{Z}_{[0,T]}$, define
\begin{equation}\label{eq:def_cycleprogress}
    \cycleprogress'_t \coloneqq \norm{\Pi^{\square}\ket{\phi_t^{(k-1)}}},
\end{equation}
which corresponds to the progress of recording a $k$-cycle given that not many length-$(k-1)$ paths have ever been recorded, and serves as a good estimator of the actual progress.
\begin{lemma}\label{lem:with_vs_without_long_paths}
    For all sufficiently large $n$,
    \begin{equation}
        \cycleprogress_T \leq \cycleprogress'_T + T\cdot \Bigl(2n  \Bigl(\frac{2em}{n\ceil{\log(n)}}\Bigr)^{\ceil{\log(n)}}+2n\left(\frac{2em}{n\Delta}\right)^{\Delta}\Bigl(\frac{emn}{2\ceil{\log(n)}}\Bigr)^{\ceil{\log(n)}} + \frac{k}{n}\Bigr).
    \end{equation}
\end{lemma}
\begin{proof}
    Using the reverse triangle inequality and the fact that $\Pi^\square$ is a contraction, we have
    \begin{equation}
        \cycleprogress_T - \cycleprogress'_T = \bigl\|\Pi^{\square}\ket{\phi_T}\bigr\| - \bigl\|\Pi^{\square}\ket{\phi_T^{(k-1)}}\bigr\| \leq \big\|\Pi^{\square}\ket{\phi_T} - \Pi^{\square}\ket{\phi_T^{(k-1)}}\bigr\| \leq \bigl\|\ket{\phi_T} - \ket{\phi_T^{(k-1)}}\bigr\|.
    \end{equation}
The RHS of the above equation can bounded by the triangle inequality applied to the telescoping sum    \begin{equation}\label{eq:overall_exclusion_bound}
    \norm{\ket{\phi_T} - \ket{\phi_T^{(k-1)}}}\leq \norm{\ket{\phi_T^{(0)}} - \ket{\phi_T^{(1)}}} + \dots + \norm{\ket{\phi_T^{(k-2)}} - \ket{\phi_T^{(k-1)}}}.
\end{equation}
Recall that \cref{eq:bad_recurrence}, for each $l\in [k-1]$ and $t\in \mathbb{Z}_{[0,T]}$,
\begin{equation}
    \I - \Pi^{(l)}_{\Bad, t} = \bigl(\I - \Pi^{\{l\}}_{r_l(t)}\bigr) \bigl(\I - \Pi^{(l-1)}_{\Bad, t}\bigr).
\end{equation}
Applying the Exclusion Lemma (\cref{lem:exclusion_lemma}) and notice the projectors on basis states commute, we obtain
\begin{equation}
    \norm{\ket{\phi_T^{(l-1)}} - \ket{\phi_T^{(l)}}} \leq \sum_{t=1}^{T-1} \norm{\Pi^{\{l\}}_{r_l(t)} \ket{\phi_t^{(l-1)}}} = \sum_{t=1}^{T-1} \pathsprogress^{(l)}_{t,{r_l(t)}}.
\end{equation}
Therefore,
\begin{align}\label{eq:stepwise_exclusion_bound}
     &\norm{\ket{\phi_T^{(l-1)}} - \ket{\phi_T^{(l)}}} 
     \\
     \leq&
     \begin{cases}
         T\cdot \left(2n\Bigl(\frac{2em}{n\ceil{\log(n)}}\Bigr)^{\ceil{\log(n)}}+2n\left(\frac{2em}{n\Delta}\right)^{\Delta}\Bigl(\frac{emn}{2\ceil{\log(n)}}\Bigr)^{\ceil{\log(n)}}\right), &\text{if $l=1$ using \cref{prop:p_base_bound}},
         \notag
         \\
        T \cdot \frac{1}{n},&\text{if $l\in \mathbb{N}_{\geq 2}$ using \cref{prop:p_lpaths_bound}}.
     \end{cases}
\end{align}

The lemma follows from substituting \cref{eq:stepwise_exclusion_bound} into \cref{eq:overall_exclusion_bound}.
\end{proof}

We can bound $\cycleprogress_T'$ by first showing that it satisfies a recurrence in the next lemma and then solving it in the following corollary.
\begin{lemma}[$k$-cycle progress recurrence]\label{lem:kcycle_progress_without_long_paths}
    For all $t\in \mathbb{Z}_{[0,T]}$,
    \begin{equation}\label{eq:ct_recurrence}
        \cycleprogress'_{t+1} \leq \cycleprogress'_{t} + 4 \sqrt{\frac{r_{k-1}(t)}{N}},
    \end{equation}
    and $\cycleprogress_0' = 0$.
\end{lemma}

Solving this recurrence gives
\begin{corollary}[$k$-cycle progress]\label{cor:kcycle_progress_without_long_paths}
  We have
    \begin{equation}
        \cycleprogress'_{T} \leq O\Bigl(\Delta^k\cdot
        \Bigl(\frac{T}{n^{3/4-1/(2^{k+2}-4)}}
       \Bigr)^{2-2/2^k} + \Delta^{k/2}\cdot T\frac{\sqrt{\log n}}{n}\Bigr).
    \end{equation}
\end{corollary}
\begin{proof}
    From \cref{lem:kcycle_progress_without_long_paths}, we obtain
    \begin{equation}\label{eq:cycleprogress_sum_bound}
         \cycleprogress'_T \leq 4 \sum_{t=0}^T \sqrt{\frac{r_{k-1}(t)}{N}}.
    \end{equation}
    We now recall the definition from
    \cref{eq:rjt_definition} and \cref{eq:constant_factor} that
    \begin{equation}
         r_{k-1}(t) \coloneqq \max\left\{\biggl\lceil \Bigl(\frac{ct}{\sqrt{n}}\Bigr)^{2-4/2^{k-1}} (t+1)^{2/2^{k-1}}\biggr\rceil, \ceil{\log n}\right\}\cdot((k-1)\Delta^{k-2}),
    \end{equation}
    where $c \coloneqq 32e \cdot k\Delta^{k/2}$, and so
    \begin{equation}\label{eq:r_{k-1}_bound}
         r_{k-1}(t) \leq O\Bigl(\bigl(\Delta^{k} \cdot t^{2-4/2^k}/ n^{1-4/2^k}+\log(n)\bigr)\cdot \Delta^k\Bigr).
    \end{equation}

    Therefore,
    \begin{equation}
         \cycleprogress'_T \leq O\Bigl( \Delta^k\cdot \frac{T^{2-2/2^k}}{n^{3/2-2/2^k}} + \Delta^{k/2}\cdot T\frac{\sqrt{\log n}}{n}\Bigr) 
       = O\Bigl(\Delta^k\cdot
        \Bigl(\frac{T}{n^{3/4-1/(2^{k+2}-4)}}
       \Bigr)^{2-2/2^k} + \Delta^{k/2}\cdot T\frac{\sqrt{\log n}}{n}\Bigr),
    \end{equation}
    where the equality uses $(3/4-1/(2^{k+2}-4))(2-2/2^k) = 3/2 - 2/2^k$.
\end{proof}

\begin{proof}[Proof of \cref{lem:kcycle_progress_without_long_paths}]
    The boundary condition $\cycleprogress_0'=0$ in \cref{eq:ct_recurrence} is trivial. To prove the recurrence, observe that
    \begin{align}
        \cycleprogress_{t+1}' &= \norm{\Pi^{\square} \ket{\phi_{t+1}^{(k-1)}}} &&\text{definition \cref{eq:def_cycleprogress}}
        \notag
        \\
        &= \norm{\Pi^{\square} U_{t+1} \calR \bigl(\I - \Pi^{(k-1)}_{\Bad,t}\bigr)\ket{\phi_t^{(k-1)}}} &&\text{\cref{eq:phi_i^l_properties}}
        \notag
        \\
        &=\norm{\Pi^{\square} \calR \bigl(\I - \Pi^{(k-1)}_{\Bad,t}\bigr)\ket{\phi_t^{(k-1)}}}  &&\text{$\Pi^{\square}$, $U_{t+1}$ commute}
        \notag
        \\
        &\leq \norm{\Pi^{\square} \ket{\phi_t^{(k-1)}}} + \norm{\Pi^{\square} \calR \bigl(\I - \Pi^{(k-1)}_{\Bad,t}\bigr)\bigl(\I-\Pi^{\square}\bigr)\ket{\phi_t^{(k-1)}}}  &&\text{triangle inequality}
        \notag
        \\
        &=  \cycleprogress_t' + \norm{\Pi^{\square} \calR \bigl(\I - \Pi^{(k-1)}_{\Bad,t}\bigr)\bigl(\I-\Pi^{\square}\bigr)\ket{\phi_t^{(k-1)}}} &&\text{definition \cref{eq:def_cycleprogress}}\label{eq:cycle_recurrence_terms}
    \end{align}

    Abbreviate
    \begin{equation}
        \ket{\rho} \coloneqq \bigl(\I - \Pi^{(k-1)}_{\Bad,t}\bigr)\bigl(\I-\Pi^{\square}\bigr)\ket{\phi_t^{(k-1)}}.
    \end{equation}
    Then, we can write the second term in \cref{eq:cycle_recurrence_terms} as
    \begin{equation}
        \norm{\Pi^{\square} \calR \ket{\rho}} \leq \sum_{y\in [N]\cup \{\bot\}} \norm{\Pi^{\square} \calR \Pi_y \ket{\rho}},
    \end{equation}
    using $\Pi^{\square} \calR \ketbrasame{0}_P \ket{\rho} = 0$ because $\ket{\rho}$ has recorded no $k$-cycle and $\calR\ketbrasame{0}_P = \ketbrasame{0}_P$.

    We will show 
    \begin{equation}\label{eq:cycle_progress_bound}
    \norm{\Pi^{\square} \calR \Pi_{\bot}\ket{\rho}} \leq \sqrt{\frac{r_{k-1}(t)}{N}} \quad \text{and} \quad
    \sum_{y\in [N]}\norm{\Pi^{\square} \calR \Pi_y \ket{\rho}} \leq 3\sqrt{\frac{r_{k-1}(t)}{N}}.
    \end{equation}

    \paragraph{First bound in \cref{eq:cycle_progress_bound}.} 
        When $y=\bot$, consider each basis state $\ket{i,u,w}\ket{x}$ in the support of
        \begin{equation}
            \Pi_\bot \bigl(\I - \Pi^{(k-1)}_{\Bad,t}\bigr)\bigl(\I-\Pi^{\square}\bigr)\ket{\phi_t^{(k-1)}}. 
        \end{equation}
         Due to $\Pi_\bot$, we have $x_i = \bot$ and $u\neq 0$, so \cref{lemma:recording_operator_effect} states
        \begin{equation}
            \calR\ket{i,u,w}\ket{x} = \ket{i,u,w}\biggl(\sum_{y\in [N]}\frac{\omega_N^{u y}}{\sqrt{N}}\ket{y}_{I_i}\biggr)\otimes\bigotimes_{j\neq i}\ket{x_{j}}_{I_{j}}.
        \end{equation}
        Since $\ket{i,u,w}\ket{x}$ is also in the support of $(\I-\Pi^{(k-1)}_{\Bad,t})$, $x$ records at most $r_{k-1}(t)$ length-$(k-1)$ paths. Due to $\bigl(\I-\Pi^{\square}\bigr)$, $x$ does not record a length-$k$ cycle. Therefore, in order for $\ket{y}_{I_i} \otimes \bigotimes_{j\neq i}\ket{x_{j}}_{I_{j}}$ to record a length-$k$ cycle, edge $y$ must connect the two ends of a length-$(k-1)$ path. The number of such $y$s is at most $r_{k-1}(t)$. Therefore,
        \begin{equation}
        \norm{\Pi^{\square}\calR\ket{i,u,w}\ket{x}}\leq\sqrt{\frac{r_{k-1}(t)}{N}}.
        \end{equation}
        Since any two distinct basis states in the support of $\Pi_\bot \bigl(\I - \Pi^{(k-1)}_{\Bad,t}\bigr)\bigl(\I-\Pi^{\square}\bigr)\ket{\phi_t^{(k-1)}}$ remain orthogonal after $\Pi^{\square}\calR$ is applied, we have
        \begin{equation}
        \begin{aligned}            
             \norm{\Pi^{\square} \calR \Pi_{\bot}\ket{\rho}} &= \norm{\Pi^{\square} \calR \Pi_{\bot}\bigl(\I-\Pi^{(k-1)}_{\Bad,t}\bigr)\bigl(\I-\Pi^{\square}\bigr)\ket{\phi_t^{(k-1)}}} \leq\sqrt{\frac{r_{k-1}(t)}{N}}. 
        \end{aligned}
        \end{equation}
        
        \paragraph{Second bound in \cref{eq:cycle_progress_bound}.} 

        When $y\in [N]$, consider each basis state $\ket{i,u,w}\ket{x}$ in the support of
        \begin{equation}
            \Pi_y \bigl(\I - \Pi^{(k-1)}_{\Bad,t}\bigr)\bigl(\I-\Pi^{\square}\bigr)\ket{\phi_t^{(k-1)}}. 
        \end{equation}
        Due to $\Pi_y$, we have $x_i = y$ and $u\neq 0$, so \cref{lemma:recording_operator_effect} states
        \begin{equation}
        \begin{aligned}
            \calR\ket{i,u,w}\ket{x}
            =\ket{i,u,w}&\biggl(\frac{\omega^{u y}_N}{\sqrt{N}}\ket{\bot}_{I_i}+\frac{1+\omega^{u y}_N(N-2)}{N}\ket{y}_{I_i}\\
            &\qquad\qquad+\sum_{y'\in [N]\backslash\{y\}}\frac{1-\omega^{u y'}_N-\omega^{u y}_N}{N}\ket{y'}_{I_i}\biggr)\otimes\bigotimes_{j\neq i}\ket{x_{j}}_{I_{j}}.
        \end{aligned}
        \end{equation}
        Applying similar reasoning as in the proof of the first bound in \cref{eq:cycle_progress_bound}, we deduce
        \begin{equation}
            \norm{\Pi^{\square}\calR\ket{i,u,w}\ket{x}}\leq 3\frac{\sqrt{r_{k-1}(t)}}{N}.
        \end{equation}
        Since any two distinct basis states in the support of $\Pi_y \bigl(\I - \Pi^{(k-1)}_{\Bad,t}\bigr)\bigl(\I-\Pi^{\square}\bigr)\ket{\phi_t^{(k-1)}}$ remain orthogonal after $\Pi^{\square}\calR$ is applied, we have
        \begin{equation}
            \norm{\Pi^{\square} \calR \Pi_{y}\bigl(\I-\Pi^{(k-1)}_{\Bad,t}\bigr)\bigl(1-\Pi^{\square}\bigr)\ket{\phi_t^{(k-1)}}} \leq 3\frac{\sqrt{r_{k-1}(t)}}{N}\norm{\Pi_{y}\bigl(\I-\Pi^{(k-1)}_{\Bad,t}\bigr)\bigl(1-\Pi^{\square}\bigr)\ket{\phi_t^{(k-1)}}}.
        \end{equation}
        Therefore, the Cauchy-Schwarz inequality gives
        \begin{align}
            &\sum_{y\in [N]} \norm{\Pi^{\square} \calR \Pi_{y}\bigl(\I-\Pi^{(k-1)}_{\Bad,t}\bigr)\bigl(1-\Pi^{\square}\bigr)\ket{\phi_t^{(k-1)}}} 
            \nonumber\\
            &\qquad \le 3\sqrt{\frac{r_{k-1}(t)}{N}}\sqrt{\sum_{y\in [N]}\norm{\Pi_{y}\bigl(\I-\Pi^{(k-1)}_{\Bad,t}\bigr)\bigl(1-\Pi^{\square}\bigr)\ket{\phi_t^{(k-1)}}}^2}
            \leq 3\sqrt{\frac{r_{k-1}(t)}{N}}.\qedhere
        \end{align}
\end{proof}

Combining the previous two lemmas gives:
\begin{proposition}[Hardness of recording a $k$-cycle]\label{lem:recording_probability-cycle}
    For all sufficiently large $n\in \mathbb{N}$,
    \begin{equation}
    \begin{aligned}
        \cycleprogress_T \leq& O\Bigl(\Delta^k
        \Bigl(\frac{T}{n^{3/4-1/(2^{k+2}-4)}}
       \Bigr)^{2-2/2^k} + \Delta^{k/2}T\frac{\sqrt{\log n}}{n}\\
       &\qquad\qquad+ Tn\Bigl(\frac{2em}{n\ceil{\log(n)}}\Bigr)^{\ceil{\log(n)}}+Tn\left(\frac{2em}{n\Delta}\right)^{\Delta}\Bigl(\frac{emn}{2\ceil{\log(n)}}\Bigr)^{\ceil{\log(n)}}\Bigr).
    \end{aligned}
    \end{equation}
\end{proposition}

\begin{proof}
The proposition follows by combining \cref{lem:with_vs_without_long_paths,cor:kcycle_progress_without_long_paths}.
\end{proof}

\subsubsection{Completing the proof of \cref{thm:kcycle_avg_lower}}

We finally bound the probability of the quantum query algorithm succeeding even if it does not record a $k$-cycle. Let $\Pi_{\success}$ denote the projector onto basis states $\ket{i,u,w}\ket{x}$ such that $w$ contains an output substring (i.e., a substring located at some fixed output register) of the form $(a_1,\dots,a_k,b_1,\dots,b_k) \in [m]^k \times \binom{[n]}{2}^k$ where the $a_i$s are distinct, the $b_i$s form a $k$-cycle, and $x_{a_i} = b_i$ for all $i$.

\begin{lemma}[Hardness of guessing a $k$-cycle]\label{lem:guessing_probability-cycle} Suppose $N\geq k$. Then, for every state $\ket{\phi}$, $\norm{\Pi_{\success}\calT (\I - \Pi^{\square}) \ket{\phi}} \leq O(\frac{1}{n})$.
\end{lemma}

With the above lemma, \cref{thm:kcycle_avg_lower} becomes a simple corollary: we have
\begin{align*}
    \norm{\Pi_{\success}\ket{\psi_T}}=& \norm{\Pi_{\success}\calT\ket{\phi_T}}
    \notag
    \\
    \leq& \norm{\Pi^\square \ket{\phi_T}} + \norm{\Pi_{\success}\calT(\id-\Pi^\square)\ket{\phi_T}}
    \notag
    \\
    \leq& O\Bigl(\Delta^k
    \Bigl(\frac{T}{n^{3/4-1/(2^{k+2}-4)}}
    \Bigr)^{2-2/2^k} + \Delta^{k/2}T\frac{\sqrt{\log n}}{n} + Tn\Bigl(\frac{2em}{n\ceil{\log(n)}}\Bigr)^{\ceil{\log(n)}}\\
    &\qquad\qquad+Tn\left(\frac{2em}{n\Delta}\right)^{\Delta}\Bigl(\frac{emn}{2\ceil{\log(n)}}\Bigr)^{\ceil{\log(n)}}\Bigr),
\end{align*}
where, in the second inequality, we used \cref{lem:recording_probability-cycle} to bound $\Pi^\square \ket{\phi_T}$ and \cref{lem:guessing_probability-cycle} to bound $\norm{\Pi_{\success}\calT(\id-\Pi^\square)\ket{\phi_T}}$, and observed that $\Delta^{k/2}\cdot T\frac{\sqrt{\log n}}{n}\geq \frac{1}{n}$ for $T,\Delta\in \mathbb{N}$ and $n\geq 2$.

\begin{proof}[Proof of \cref{lem:guessing_probability-cycle}]
    For $\alpha,\beta\in \mathbb{Z}_{\geq 0}$ with $\alpha+\beta\leq k$, we define the projector $P_{\alpha,\beta}$ to be onto basis states $\ket{i,u,w}\ket{x}$, where
\begin{enumerate}
    \item $w$ contains the output substring $(a_1,\dots,a_k,b_1,\dots,b_k) \in [m]^k \times \binom{[n]}{2}^k$ where the $a_i$s are distinct and the $b_i$s form a $k$-cycle;
    \item there are exactly $\alpha$ indices $i\in [k]$ such that $x_{a_i} = \bot$;
    \item there are exactly $\beta$ indices $i \in [k]$ such that $x_{a_i} \neq \bot$ and $x_{a_i} \neq b_i$.
\end{enumerate}
Observe that if $\ket{i,u,w}\ket{x}$ is in the support of $\id-\Pi^\square$, then $P_{0,0} \ket{i,u,w}{\ket{x}} = 0$. This is because if $\ket{i,u,w}\ket{x}$ is in the support of $P_{0,0}$, then $x_{a_i} = b_i$ for all $i$ and the $b_i$s form a $k$-cycle, which contradicts the assumption.

For a state $\ket{i,u,w}\ket{x}$ in the support of $P_{k,l}$, we have 
\begin{equation}\label{eq:succ_single_basis_kl-cycle}
    \norm{\Pi_{\success} \calT \ket{i,u,w}\ket{x}} \leq (1/\sqrt{N})^\alpha(1/N)^\beta,
\end{equation}
using the definition of $\calT$.

For $\ket{i,u,w}\ket{x}$ in the support of $P_{\alpha,\beta}$, we write $w_{a}\coloneqq \{a_1,\dots,a_k\}$ for the set containing the first $k$ elements of the output substring. Observe that $\ket{i,u,w}\ket{x}$ and $\ket{i',u',w'}\ket{x'}$ remain orthogonal after applying $\Pi_{\success}\calT$ unless $i=i',u=u',w=w'$ and $x_s = x'_s$ for all $s\in [m]-\{a_1,\dots,a_k\}$.

Therefore, for a state $\ket{\chi} \coloneqq \sum_{i,u,w,x}c_{i,u,w,x} \ket{i,u,w}\ket{x}$ in the support of $P_{\alpha,\beta}$, we have
\begin{align}
    \norm{\Pi_{\success}\calT \ket{\chi}}^2 &= \sum_{i,u,w,(x_{a'})_{a'\notin w_a}}\biggl\|\sum_{(x_{a'})_{a'\in w_a}} c_{i,u,w,x} \Pi_{\success} \calT \ket{i,u,w}\ket{x}\biggr\|^2
    \notag
    \\
    &\leq  \sum_{i,u,w,(x_{a'})_{a'\notin w_a}}\Bigl(\sum_{(x_{a'})_{a'\in w_a}} \abs{c_{i,u,w,x}}^2 \Bigr)\Bigl(\sum_{(x_{a'})_{a'\in w_a}}\1[c_{i,u,w,x}\neq 0]\cdot \norm{\Pi_{\success} \calT \ket{i,u,w}\ket{x}}^2\Bigr)
    \notag
    \\
    &\leq   \binom{k}{\alpha,\beta,k-(\alpha+\beta)}(N-1)^\beta \cdot \Bigl(\frac{1}{N}\Bigr)^\alpha\Bigl(\frac{1}{N^2}\Bigr)^\beta \norm{\ket{\chi}}^2 \leq \frac{1}{\alpha!\beta!}\Bigl(\frac{k}{N}\Bigr)^{\alpha+\beta}\norm{\ket{\chi}}^2,\label{eq:succ_ab}
\end{align}
where the first inequality is Cauchy-Schwarz and the second inequality is \cref{eq:succ_single_basis_kl-cycle}.

Finally, write $P$ for the projector onto basis states $\ket{i,u,w}\ket{x}$ not satisfying the first condition defining $P_{\alpha,\beta}$ so that $\id = P + \sum_{\alpha,\beta\in \mathbb{Z}_{\geq 0}\colon \alpha+\beta \leq k}P_{\alpha,\beta}$. Therefore
\begin{align}
    \norm{\Pi_{\success}\calT (\id-\Pi^\square)\ket{\phi}}^2 &= \biggl\|\Pi_{\success}\calT (\id-\Pi^{\square})\Bigl(P + \sum_{\alpha,\beta\in \mathbb{Z}_{\geq 0}\colon \alpha+\beta \leq k} P_{\alpha,\beta}\Bigr)\ket{\phi}\biggr\|^2
    \notag
    \\
    &=\biggl\|\Pi_{\success}\calT (\id-\Pi^{\square})\Bigl(\sum_{\alpha,\beta\in \mathbb{Z}_{\geq 0}\colon 0< \alpha+\beta \leq k} P_{\alpha,\beta}\Bigr)\ket{\phi}\biggr\|^2
    \notag
    \\
    &= \biggl\|\Pi_{\success}\calT \Bigl(\sum_{\alpha,\beta\in \mathbb{Z}_{\geq 0}\colon 1\leq \alpha+\beta \leq k} P_{k,l}\Bigr) (\id-\Pi^{\square})\ket{\phi}\biggr\|^2
    \notag
    \\
    &\leq \frac{k(k+3)}{2}\cdot \frac{k}{N} \leq O\Bigl(\frac{1}{N}\Bigr),\label{eq:guessing_prob-cycle}
\end{align}
where the second equality uses $\Pi_{\success}P = 0$ and $(\id-\Pi^{\square}) P_{0,0} = 0$ as observed previously, and the second-to-last inequality uses \cref{eq:succ_ab}, the restriction on the sum of $P_{\alpha,\beta}$s to $1 \leq \alpha+\beta$, and $N\geq k$. 

The lemma follows from \cref{eq:guessing_prob-cycle} after taking square roots on both sides and recalling $N = \binom{n}{2}$.
\end{proof}

\subsection{Cycle finding upper bound}
In this section, we prove the following theorem.
\begin{theorem}\label{thm:worst-case-decision-k}
    Let $k\in \mathbb{N}$. For all $m,d \in \mathbb{N}$, $Q(\kcycle{k}_{m,d}) = O((2d)^{(k^2+k-6)/4}2^{2dk}m^{3/4-1/(2^{k+2}-4))})$.
\end{theorem}
To prove \cref{thm:worst-case-decision-k}, when $d\geq \log_2(m)$, the trivial bound $Q(\kcycle{k}_{m,d})\leq m$ applies, ensuring \cref{thm:worst-case-decision-k} is satisfied. When $d\leq \log_2(m)$, we employ the learning graph method to confirm the validity of \cref{thm:worst-case-decision-k}.
\begin{corollary}\label{cor:worst-case-decision-k}
    For all $m\in \mathbb{N}$, $d\leq O(\log(m)/\log\log(m))$, and $k\in O(1)$, $Q(\kcycle{k}_{m,d}) = O(m^{3/4-1/(2^{k+2}-4) + o(1)})$. In view of \cref{fact:sparse-random-graph-restated}, if $n\geq \Omega(m)$ there exists a quantum query algorithm using $O(m^{3/4-1/(2^{k+2}-4) + o(1)})$ queries that finds a triangle in a random sparse graph $x \leftarrow \binom{[n]}{2}^m$, or decides it does not exist, with success probability $\geq 1 - o(1)$.
\end{corollary}

Unlike the triangle-finding variant, where some problem-specific simplifications and improvements naturally emerged, our generic $k$-cycle algorithm retains a nearly \emph{identical} structure to Belovs's original learning graph algorithm for $k$-distinctness \cite{kdist_learning_graphs_12}. In fact, by simply setting $d$ equal to $k$ and replacing occurrences of the notion of ``incident to'' by ``equal to'', our algorithm essentially solves $k$-distinctness itself with a complexity of $O(m^{3/4-1/(2^{k+2}-4)})$. 

Although the proof closely resembles those of our triangle-finding algorithm and Belovs's $k$-distinctness algorithm, we give the proof in full for completeness and clarity.

\subsubsection{Learning graph construction}
We first describe the construction of the learning graph. For simplicity, we use $f$ to denote $\kcycle{k}_{m,d}$. The vertices of the learning graph are partitioned by
\begin{equation}
    \dot{\bigcup}_{i\in\mathbb{Z}_{[0,k+1]}}V^{(i)}.
\end{equation}
A vertex $S^{(i)}$ in $V^{(i)}$ is (labeled by) an array containing pairwise-disjoint subsets of $[m]$; we refer to entries of $S^{(i)}$ by
\begin{equation}
    S^{(i)} = \dot{\bigoplus}_{j=1}^{k}\left(\dot{\bigoplus}_{d_1,\dots,d_{j-1}\in[2d],\emptyset\neq D\subseteq [2d]}S_j^{(i)}(d_1,\dots,d_{j-1},D)\right)\in\mathcal{P}([m])^{\sum_{l=1}^{k}(2d)^{l-1}(2^{2d}-1)},
\end{equation}
where $\oplus$ means array concatenation, $\dot{\oplus}$ means the subsets being concatenated are pairwise disjoint, and $\mathcal{P}([m])$ denotes the power set of $[m]$.

We further impose conditions on the sizes of these sets using some $r_1,\dots,r_{k} \in \mathbb{N}$ that will be chosen later, at the end of \cref{subsubsection:learning_graph_complexity_for_k}. To describe these conditions, it is convenient to write $\mu(A)$ for the minimum element of a finite non-empty subset $A \subset \mathbb{Z}$. 

Then, for every $i\geq 1$, we impose that  $V^{(i)}$ can be partitioned by labels consisting of a length-$(i-1)$ sequence of non-empty subsets of $[2d]$, $D_1,D_2,\dots,D_{i-1}$, such that $V^{(i)}(D_1,\dots,D_{i-1})$ consists of all vertices $S^{(i)}$ such that the size of $S^{(i)}_j(\cdot)$ satisfies
\begin{equation}
    \abs{S^{(i)}_j(d_1,\dots,d_{j-1},D)} =
    \begin{cases}
        r_j+1, & j < i, d_1 = \mu(D_1),\dots,d_{j-1}=\mu(D_{j-1}), D = D_j,\\
        r_j, & \text{otherwise}.
    \end{cases}
\end{equation}

For each vertex $R$ in the learning graph, define
\begin{equation}
    \bigcup R_i\coloneqq \bigcup_{d_1,\dots,d_{i-1}\in[2d],\emptyset\neq D\subseteq [2d]} R_i(d_1,\dots,d_{i-1},D),
    \quad\text{and}\quad
    \bigcup R\coloneqq \bigcup_{i\in[k]}\left(\bigcup R_i\right).
\end{equation}
We refer to indices contained in $\bigcup R_i$ and $\bigcup R$ as having been loaded in $R_i$ and $R$, respectively. If there is an arc from vertex $R$ to vertex $S$ in the learning graph, by the construction in the following paragraphs, it must satisfy $\bigcup S = \bigcup R \cup \{j\}$. We denote the arc as $A^{R,S}_j$, and accordingly, we say $A^{R,S}_j$ is associated with the loading of $j$.

There exists a unique vertex $R^{(0)}\in V^{(0)}$ such that $\bigcup R^{(0)} = \emptyset$. We refer to this vertex as the source, denoted by $\emptyset$. Each vertex $R^{(1)}\in V^{(1)}$ loads exactly
\begin{equation}
    r \coloneqq \sum_{i=1}^{k}r_i(2d)^{i-1}(2^{2d}-1)
\end{equation}
indices. We fix an arbitrary ordering $t_1,\dots, t_r$ of indices in $\bigcup R^{(1)}$ such that all indices in $\bigcup R^{(1)}_i$ precede those in $\bigcup R^{(1)}_j$ if $i<j$. There exists a length $r$ path from the source $\emptyset$ to $R^{(1)}$, with all intermediate vertices lying in $V^{(0)}$. Along this path, the element $t_i$ is loaded on the $i$-th arc for $i\in[r]$ in stage I, and specifically, if $t_i\in \bigcup R_s^{(1)}$, it is loaded in stage I.$s$. All such paths are disjoint except at the source $\emptyset$. Therefore, there are $r\bigl|V^{(1)}\bigr|$ arcs in stage I.

In stage II.$s$ for $s\in[k]$, for each $R^{(s)}\in V^{(s)}(D_1,\dots,D_{s-1})$, $j\not\in\bigcup R^{(s)}$ and $\emptyset\neq D_s\subseteq[2d]$, there is an arc loading $j$, denoted  $A^{R^{(s)}, R^{(s+1)}}_j$, from $R^{(s)}$ to $R^{(s+1)} \in V^{(s+1)}(D_1,\dots,D_s)$ defined by
\begin{equation*}
\begin{cases}
    R^{(s+1)}_s(\mu(D_1),\dots,\mu(D_{s-1}),D_s) = R^{(s)}_s(\mu(D_1),\dots,\mu(D_{s-1}),D_s)\cup\{j\},\\
    R^{(s+1)}_i(d_1,\dots,d_{i-1},D) = R^{(s)}_i(d_1,\dots,d_{i-1},D) \quad\text{if }(i,d_1,\dots,d_{i-1},D)\neq (s,\mu(D_1),\dots,\mu(D_{s-1}),D_s).
\end{cases}
\end{equation*}

\subsubsection{Active arcs}
For every $x\in f^{-1}(1)$, let $\mathcal{C}(x) = \{a_1, a_2, \dots, a_k\}\subseteq[m]$ be a specific certificate for $x$ such that $x_{a_1}, \dots, x_{a_k}$ form a $k$-cycle, and $x_{a_i}$ is incident to $x_{a_{i+1}}$ for all $i\in[k-1]$.

We say $R^{(1)}\in V^{(1)}$ is consistent with $x\in f^{-1}(1)$ if edges indexed by $\mathcal{C}(x)$ are vertex disjoint from edges indexed by $\bigcup R^{(1)}$. Let $n_x$ be the number of $R^{(1)}\in V^{(1)}$ that are consistent with $x$, and define
\begin{equation}
    q\coloneqq \Bigl(\min_{x\in f^{-1}(1)}\{n_x\}\Bigr)^{-1}.
\end{equation}
so that for all $x\in f^{-1}(1)$, there exist at least $q^{-1}$ vertices in $V^{(1)}$ that are consistent with $x$. By \cref{fact:vertex-avoid}, when $kdr\leq o(m)$ --- a condition satisfied since we only need to consider $k\in O(1)$, $d\in\log_2(m)$, and use the $r$ specified in \cref{subsubsection:learning_graph_complexity_for_k} --- it follows that $q^{-1}\geq\Omega(\bigl|V^{(1)}\bigr|)$.

Let $\Cst(x)$ denote a fixed but arbitrary set of $q^{-1}$ vertices $R^{(1)}\in V^{(1)}$ that are consistent with $x$. For each $R^{(1)}\in \Cst(x)$, $\Act(x,R^{(1)})$ is a set of active arcs consisting of the following arcs:
\begin{itemize}
     \setlength{\itemsep}{0.5pt}
    \renewcommand{\labelitemi}{\tiny$\blacksquare$}
    \item
    In stage I, all arcs along the unique shortest length $r$ path from source $\emptyset$ to $R^{(1)}$.
    \item 
    In stage II.$s$, we define recursively so that if $R^{(s)}\in V^{(s)}$ has an incoming arc from $\Act(x, R^{(1)})$, all arcs $A_{a_s}^{R^{(s)},R^{(s+1)}}$ in the learning graph such that $R^{(s+1)} \in V^{(s+1)}$ are in $\Act(x, R^{(1)})$.
\end{itemize}

Lastly, the set of all active arcs of $x$ is defined by
\begin{equation}
    \mathrm{Act}(x)\coloneqq \dot{\bigcup}_{R^{(1)}\in \Cst(x)}\Act(x,R^{(1)}).
\end{equation}

\subsubsection{Matrices for the adversary bound}
For every vertex $R$ in the learning graph, an \emph{assignment on $R$} refers to a function $\alpha_R: \bigcup R\to \binom{[n]}{2}\cup\{*\}$ such that:
\begin{enumerate}
\item for all $j\in \bigcup R_1$,
\begin{equation}
   \alpha_R(j)\neq *;
\end{equation}
\item for all $j \in R_i(d_1,\dots,d_{i-1},D)$ such that $i>1$,
define
\begin{equation}
    N(R, d_1,\dots,d_{i-1})\coloneqq \bigcup_{d_{i-1}\in D'\subseteq[2d]}R_{i-1}(d_1, \dots, d_{i-2}, D'),
\end{equation}
then $\alpha_R(j)$ must satisfy
\begin{equation}
   \alpha_R(j)\in\{*\}\cup \Bigl\{e\in \scalebox{1.25}{$\binom{[n]}{2}$} \Bigm| \exists k\in N(R, d_1,\dots,d_{i-1})\text{ s.t. $e$  is incident to $\alpha_R(k)$} \Bigr\}.
\end{equation}
Note that the special symbol $*$ is not incident to any edge.
\end{enumerate}
We say an input $z\in \binom{[n]}{2}^m$ \emph{satisfies assignment $\alpha_{R}$} if, for all $t\in \bigcup R$,
\begin{equation}\label{eq:alpha_def_for_k}
\alpha_{R}(t) =
\begin{cases}
  z_t, 
    & \text{if $t \in \bigcup R_1$}, \\[1ex]
  z_t, 
    &
    \begin{aligned}[t]
      &\text{if $i>1$, $t \in R_i(d_1,\dots,d_{i-1},D)$ and}\\
      &\text{$\exists k\in N(R,d_1,\dots,d_{i-1})$  s.t. 
        $z_t$ is incident to $\alpha_R(k)$},
    \end{aligned}\\[1ex]
  *, 
    &\text{otherwise}.
\end{cases}
\end{equation}
For each vertex $R$ in the learning graph, we write $\alpha_R^x$ for the unique assignment on $R$ that $x$ satisfies. We say arc $A^{R,S}_j$ \emph{uncovers $j$} if $\alpha_S^x(j)\neq *$. We say inputs $x,y\in \binom{[n]}{2}^m$ \emph{agree on $R$} if they satisfy $\alpha_R^x = \alpha_R^y$; we also say they agree on a subset of $\bigcup R$ if the restriction of $\alpha_R^x$ and $\alpha_R^y$ to that subset equal.

Define
\begin{equation}
    X_j^{R,S}\coloneqq \sum_{\alpha_R} Y_{\alpha_R},
\end{equation}
where the summation is over all assignments $\alpha_R$ on $R$. 

For each arc $A_j^{R,S}$ in stage I.1, define $Y_{\alpha_R}\coloneqq q\psi_{\alpha_R}\psi_{\alpha_R}^\dagger$, where $\psi_{\alpha_R}$ is a real vector indexed by $\binom{[n]}{2}^m$ and defined entry-wise by
\begin{equation}
    \psi_{\alpha_R}[z] \coloneqq
    \begin{cases}
        1, & \text{if $f(z)=1$, $z$ satisfies $\alpha_R$, and $A^{R,S}_j\in \Act(z)$},\\
        1, & \text{if $f(z)=0$ and $z$ satisfies $\alpha_R$},\\
        0,&\text{otherwise}.
    \end{cases}
\end{equation}
With the above definition, we see that $X_j^{R,S}$ consists of blocks of the form:
\begin{equation}
\label{X_block_form_I.1_for_k}
\renewcommand{\arraystretch}{1.25}
\begin{array}{c|c|c}
    & x & y \\ \hline
x & q & q \\ \hline
y & q & q
\end{array}
\end{equation}
where $x\in f^{-1}(1)$, $y\in f^{-1}(0)$, $A^{R,S}_j\in \Act(x)$, and both $x$ and $y$ agree on $R$.

For arcs $A_j^{R,S}$ in stage I.$s$ where $s>1$, define $Y_{\alpha_R}\coloneqq q(\psi_{\alpha_R}\psi_{\alpha_R}^\dagger+\phi_{\alpha_R}\phi_{\alpha_R}^\dagger)$ where $\psi_{\alpha_R}$ and  $\phi_{\alpha_R}$ are real vectors indexed by $\binom{[n]}{2}^m$ and defined entry-wise by
\begin{equation}
    \psi_{\alpha_R}[z] \coloneqq
    \begin{cases}
        1/\sqrt{w_1}, & \text{if $f(z)=1$, $\alpha_S^z(j)\neq*$, $z$ satisfies $\alpha_R$, and $A^{R,S}_j\in \Act(z)$},\\
        \sqrt{w_1}, & \text{if $f(z)=0$,  and $z$ satisfies $\alpha_R$},\\
        0,&\text{otherwise;}
    \end{cases}
\end{equation}
and
\begin{equation}
    \phi_{\alpha_R}[z] \coloneqq
    \begin{cases}
        1/\sqrt{w_0}, & \text{if $f(z)=1$, $\alpha_S^z(j)=*$, $z$ satisfies $\alpha_R$, and $A^{R,S}_j\in \Act(z)$},\\
        \sqrt{w_0}, & \text{if $f(z)=0$, $\alpha_S^z(j)\neq*$, and $z$ satisfies $\alpha_R$},\\
        0, & \text{otherwise.}
    \end{cases}
\end{equation}
Here, $w_0$ and $w_1$ are positive real numbers that will be specified later. With the above definition, we see that $X_j^{R,S}$ consists of blocks of the form:
\begin{equation}
\label{X_block_form_I.2_for_k}
\begin{array}{l|c|c|c|c}
& x: x_j =\alpha_S^x(j)\neq * & x: \alpha_S^x(j)= * & y: y_j =\alpha_S^y(j)\neq * & y: \alpha_S^y(j)= * \\ \hline
x: x_j =\alpha_S^x(j)\neq * & q/w_1 & 0 & q & q \\ \hline
x: \alpha_S^x(j)= * & 0 & q/w_0 & q & 0 \\ \hline
y: y_j =\alpha_S^y(j)\neq * & q & q & q(w_0 + w_1) & qw_1 \\ \hline
y: \alpha_S^y(j)= * & q & 0 & qw_1 & qw_1 \\ 
\end{array}
\end{equation}
where $x\in f^{-1}(1)$, $y\in f^{-1}(0)$, $A^{R,S}_j\in \Act(x)$, and both $x$ and $y$ agree on $R$.

For arcs $A_j^{R,S}$ in stage II.$s$ where $s\in[k]$, suppose $R\in V^{(s)}(D_1,\dots,D_{s-1})$ and $S\in V^{(s+1)}(D_1,\dots, D_s)$, define $Y_{\alpha_R}\coloneqq q\psi_{\alpha_R}\psi_{\alpha_R}^\dagger$ where
\begin{equation}
    \psi_{\alpha_R}[z] \coloneqq
    \begin{cases}
        1/\sqrt{w_2},& \text{if $f(z)=1$, $z$ satisfies  $\alpha_R$, and $A_j^{R,S}\in \Act(z)$},\\
        (-1)^{s+\sum_{i=1}^{s}\abs{D_i}}\sqrt{w_2}, & \text{if $f(z)=0$, and $z$ satisfies $\alpha_R$},\\
        0,&\text{otherwise.}
    \end{cases}
\end{equation}
$X_j^{R,S}$ consists of blocks of the form:
\begin{equation}
\label{X_block_form_II_for_k}
\renewcommand{\arraystretch}{1.25}
\begin{array}{c|c|c}
    & x & y \\ \hline
x & q/w_2 & (-1)^{s+\sum_{i=1}^{s}\abs{D_i}}q \\ \hline
y & (-1)^{s+\sum_{i=1}^{s}\abs{D_i}}q & qw_2
\end{array}
\end{equation}
where $x\in f^{-1}(1)$, $y\in f^{-1}(0)$, $A^{R,S}_j\in \Act(x)$, and both $x$ and $y$ agree on $R$.

\subsubsection{Complexity}\label{subsubsection:learning_graph_complexity_for_k}

Assume constant $k\in \mathbb{N}_{\geq 3}$, we show $Q(\kcycle{k}_{m,d}) \leq O((2d)^{(k^2+k-6)/4}\, 2^{2dk} \, m^{3/4-1/(2^{k+2}-4)})$  by computing \cref{complexity_of_learning_graph}. Furthermore, if $d\leq O(\log(m)/\log\log(m))$, we have $Q(\kcycle{k}_{m,d}) \leq O(m^{3/4-1/(2^{k+2}-4)+o(1)})$.
\begin{enumerate}
    \item 
    For stage I.1, there are $r_1(2^{2d}-1)\bigl|V^{(1)}\bigr|$ arcs in this stage. By \cref{X_block_form_I.1_for_k}, each of them contributes $q$ to the complexity, so the complexity of this stage is $r_1(2^{2d}-1)\bigl|V^{(1)}\bigr|q = O(2^{2d}r_1)$.
    \item 
    For stage I.$s$ with $s>1$, there are $r_s(2d)^{s-1}(2^{2d}-1)\bigl|V^{(1)}\bigr|$ arcs in this stage. For each input $z$, we need to bound the number of arcs that can uncover an element, so we get a refined bound on the contribution from \cref{X_block_form_I.2_for_k}. Suppose such an arc in this stage is on the shortest length-$r$ path from source $\emptyset$ to $R^{(1)}\in V^{(1)}$, and is uncovering $b_s$ in $\bigcup R_s^{(1)}$, then, there must exist a length-$s$ sequence $(z_{b_1}, z_{b_2}, \dots, z_{b_s})$ such that $b_i\in \bigcup R^{(1)}_i$ and $z_{b_i}$ is incident to $z_{b_{i+1}}$ for all $i\in[s-1]$. The number of such sequences $(b_1,\dots,b_s)$ is at most $m(2d)^{s-1}$, and the number of choices of $R^{(1)}$ loads $b_i$ in stage I.$i$ for all $i\in[s]$ is upper bounded by
    \begin{equation}
    \begin{aligned}
        &(2d)^{s(s-1)/2}(2^{2d}-1)^s\binom{m-s}{r_1-1,\underbrace{r_1,\cdots,r_1}_{(2^{2d}-1)-1},r_2-1,\underbrace{r_2,\cdots,r_2}_{(2d)(2^{2d}-1)-1},\dots, r_s-1,\underbrace{r_s,\cdots,r_s}_{(2d)^{s-1}(2^{2d}-1)-1}}\\
        =& O\Bigl(\frac{r_1r_2\cdots r_s(2d)^{\frac{s(s-1)}{2}} 2^{2ds}}{m^s}\bigl|V^{(1)}\bigr|\Bigr).
    \end{aligned}
    \end{equation}
    Therefore, for a negative input, by \cref{X_block_form_I.2_for_k}, stage I.$s$'s contribution to the complexity is
    \begin{equation}
    \begin{aligned}
        &O\Bigl(qw_0(m(2d)^{s-1})\Bigl(\frac{r_1\cdots r_s(2d)^{s(s-1)/2} 2^{2ds}}{m^s}\bigl|V^{(1)}\bigr|\Bigr)+qw_1\Bigl(r_s(2d)^{s-1}(2^{2d}-1)\bigl|V^{(1)}\bigr|\Bigr)\Bigr)\\
        =&O\Bigl(\frac{r_1\cdots r_s(2d)^{(s^2+s-2)/2}2^{2ds}}{m^{s-1}}w_0+r_s(2d)^{s-1}2^{2d}w_1\Bigr);
    \end{aligned}
    \end{equation}
    and for a positive input, by \cref{X_block_form_I.2_for_k}, stage I.$s$'s contribution to the complexity is
    \begin{equation}
    \begin{aligned}
        &O\Bigl(\frac{q}{w_1}(m(2d)^{s-1})\Bigl(\frac{r_1\cdots r_s(2d)^{s(s-1)/2} 2^{2ds}}{m^s}\bigl|V^{(1)}\bigr|\Bigr)+\frac{q}{w_0}\Bigl(r_s(2d)^{s-1}(2^{2d}-1)\bigl|V^{(1)}\bigr|\Bigr)\Bigr)\\
        =&O\Bigl(\frac{r_1\cdots r_s(2d)^{(s^2+s-2)/2}2^{2ds}}{m^{s-1}w_1}+\frac{r_s(2d)^{s-1}2^{2d}}{w_0}\Bigr).
    \end{aligned}
    \end{equation}
    If we set $w_0=((2d)^{s(s-1)/4}2^{d(s-1)})^{-1}\sqrt{\frac{m^{s-1}}{r_1\cdots r_{s-1}}}$, and $w_1 = 1/w_0$, the total contribution to the complexity is
    \begin{equation}
        O\Bigl((2d)^{(s^2+3s-4)/4}2^{d(s+1)}r_s\sqrt{\frac{r_1\cdots r_{s-1}}{m^{s-1}}}\Bigr).
    \end{equation}
    \item 
    For stage II.$s$ when $s\in[k]$, each vertex in $V^{(s)}$ has $(m-r-s+1)(2^{2d}-1)$ outgoing arcs, and if $s>1$, each vertex in $V^{(s)}$ has $r_{s-1}+1$ incoming arcs. Therefore, the total number of arcs on stage II.$s$ is
    \begin{equation}
        \frac{\prod_{i=1}^s(m-r-i+1)(2^{2d}-1)}{\prod_{j=2}^{s}(r_{j-1}+1)}\bigl|V^{(1)}\bigr|\leq O\Bigl(\frac{m^s2^{2ds}\bigl|V^{(1)}\bigr|}{r_1\cdots r_{s-1}}\Bigr).
    \end{equation}
    As we will see later, $r_i\in o(m)$ for all $i\in[k]$, so the total number of arcs in stage II is upper bounded by $O(m^k2^{2dk}\bigl|V^{(1)}\bigr|/(r_1\cdots r_{k-1}))$. Among them, if a vertex in $V^{(s)}$ has an incoming active arc, that vertex has $2^{2d}-1$ outgoing active arcs. Therefore, there are at most $q^{-1}\sum_{i=1}^k(2^{2d}-1)^i = O(2^{2dk}q^{-1})$ active arcs in stage II in total. For every negative input, by \cref{X_block_form_II_for_k}, each arc contributes $qw_2$ to the complexity, so the total contribution is
    \begin{equation}
        O\Bigl(\frac{m^k2^{2dk}\bigl|V^{(1)}\bigr|}{r_1\cdots r_{k-1}}qw_2\Bigr) = O\Bigl(2^{2dk}\frac{m^kw_2}{r_1\cdots r_{k-1}}\Bigr).
    \end{equation}
    For every positive input, by \cref{X_block_form_II_for_k}, each active arc contributes $q/w_2$ to the complexity, so the total contribution is upper bounded by
    \begin{equation}
        O(2^{2dk}q^{-1}) \frac{q}{w_2} \leq O(2^{2dk}/w_2).
    \end{equation}
    To balance the contribution from negative inputs and positive inputs, we can set $w_2=\sqrt{\frac{r_1\cdots r_{k-1}}{m^k}}$ so that the total contribution in stage II is
    \begin{equation}
        O\biggl(2^{2dk}\sqrt{\frac{m^k}{r_1\cdots r_{k-1}}}\, \biggr).
    \end{equation}
\end{enumerate}

The total complexity of all stages is
\begin{equation}
    O\biggl(2^{2d}r_1 + \Bigl(\sum_{i=2}^k (2d)^{(i^2+3i-4)/4} \, 2^{d(i+1)}r_i\sqrt{\frac{r_1\cdots r_{i-1}}{m^{i-1}}}\Bigr) + 2^{2dk}\sqrt{\frac{m^k}{r_1\cdots r_{k-1}}}\, \biggr).
\end{equation}
To (approximately) balance the summands, we set $r_i = \ceil{m^{1-1/2^{i+1}-1/(2^{i+1}(2^k-1))}}$ for all $i\in[k-1]$ and $r_k=0$. The above equation leads to a complexity of $O((2d)^{(k^2+k-6)/4}\, 2^{2dk} \, m^{3/4-1/(2^{k+2}-4)})$.

\subsubsection{Feasibility}
Fix two inputs $x\in f^{-1}(1)$ and $y\in f^{-1}(0)$, to prove \cref{feasibility_of_learning_graph} holds, it is equivalent to show
\begin{equation}
    \sum_{A_j^{R,S} \in \Act(x) \colon x_j \neq y_j} X_j^{R,S}[x,y] = 1.
\end{equation}
Since $\Act(x)= \dot{\bigcup}_{R^{(1)}\in \Cst(x)}\Act(x,R^{(1)})$, it suffices to prove that for all $R^{(1)}\in \Cst(x)$,
\begin{equation}\label{eq:feasibility_for_k}
    \sum_{A_j^{R,S} \in \Act(x,R^{(1)}) \colon x_j \neq y_j} X_j^{R,S}[x,y] = \frac{1}{\abs{\Cst(x)}} = q,
\end{equation}

Recall $\mathcal{C}(x) = \{a_1,\dots, a_k\}$ is a specific certificate for $x$ that we chose. Let $t_1, t_2,\cdots,t_r\in[m]$ be the elements in $\bigcup R^{(1)}$ given in the order in which they are loaded in $R^{(1)}$. For $i\in\mathbb{Z}_{[0,r]}$, let $T_i$ be the vertex that has loaded $i$ elements on the unique shortest length-$r$ path from source $\emptyset$ to $R^{(1)}$ ($\abs{\bigcup T_i} = i$, $T_0 = \emptyset$ and $T_r= R^{(1)}$), so that arcs from $\emptyset$ to $R^{(1)}$ is in the form of $A_{t_i}^{T_{i-1},T_i}\in \Act(x, R^{(1)})$.

We perform the following case analysis depending on whether $x$ and $y$ agree on $R^{(1)}$.
\begin{itemize}
\renewcommand{\labelitemi}{\tiny$\blacksquare$}
    \item 
    If $x$ and $y$ disagree on $R^{(1)}$, there exists $i^* \in [r]$ such that $x$ and $y$ disagree on $T_i$ if and only if $i\geq i^*$. This follows from the fact that, for each arc $A_j^{R,S}$ in stage I, the following holds by \cref{eq:alpha_def_for_k}:
    \begin{equation}
        \forall l\in \bigcup R, \alpha_S^x(l) = \alpha_R^x(l) \text{ and } \alpha_S^y(l) = \alpha_R^y(l),
    \end{equation}
    so if $x$ and $y$ disagree on $R$, then they disagree on $S$.
    
    We show that $X_{t_{i^*}}^{T_{i^*-1}, T_{i^*}}[x,y] = q$, and for each $A_j^{R,S}\in \Act(x,R^{(1)})\backslash\{A_{t_{i^*}}^{T_{i^*-1}, T_{i^*}}\}$, either $x_j = y_j$ or $X_j^{R,S}[x,y] = 0$.
    \begin{itemize}
        \item 
        When we load $t_{i^*}$ into $T_{i^*-1}$, $x$ and $y$ agree on $T_{i^*-1}$ but not $T_{i^*}$. Since $x$ and $y$ agree on $T_{i^*-1}$, for all $i\in \bigcup T_{i^*-1}$, we have
        \begin{equation}
            \alpha_{T_{i^*}}^x(i) = \alpha_{T_{i^*-1}}^x(i) = \alpha_{T_{i^*-1}}^y(i) = \alpha_{T_{i^*}}^y(i),
        \end{equation}
        so we must have $\alpha_{T_{i^*}}^x(t_{i^*}) \neq \alpha_{T_{i^*}}^y(t_{i^*})$ as $x$ and $y$ disagree on $T_{i^*}$. Therefore, $x_{t_{i^*}} \neq y_{t_{i^*}}$ as otherwise $\alpha_{T_{i^*}}^x(t_{i^*}) = \alpha_{T_{i^*}}^y(t_{i^*})$. Hence, the term $X_{t_{i^*}}^{T_{i^*-1}, T_{i^*}}[x,y]$ is included in the summation in \cref{eq:feasibility_for_k}. Depending on which stage $A_{t_{i^*}}^{T_{i^*-1}, T_{i^*}}$ is in, either by \cref{X_block_form_I.1_for_k} or by \cref{X_block_form_I.2_for_k},
        \begin{equation}
            X_{t_{i^*}}^{T_{i^*-1}, T_{i^*}}[x,y] = q,
        \end{equation}
        
        \item
        When we load $t_i$ into $T_{i-1}$ such that $i\in[i^*-1]$, $x$ and $y$ agree on both $T_{i-1}$ and $T_{i}$. Then, either $x_{t_i} = \alpha_{T_i}^x(t_i) = \alpha_{T_i}^y(t_i) = y_{t_i} \neq *$, or $\alpha_{T_i}^x(t_i) = \alpha_{T_i}^y(t_i) = *$ so that by \cref{X_block_form_I.2_for_k}, $X^{T_{i-1},T_i}_{t_i}[x,y] = 0$. Therefore, arc $A_{t_{i}}^{T_{i-1}, T_{i}}$ does not contribute to \cref{eq:feasibility_for_k}.
        \item 
        Every other arc $A^{R,S}_j\in \Act(x,R^{(1)})$ is either in stage I and of the form $A_{t_i}^{T_{i-1},T_i}$ for some $i>i^*$, $i\in[r]$, or the arc is in stage II. If the arc is in stage I, then $x$ and $y$ disagree on $R$. Depending on which stage $A_{t_{i}}^{T_{i-1}, T_{i}}$ is in, either by \cref{X_block_form_I.1_for_k} or by \cref{X_block_form_I.2_for_k}, $X_j^{R,S}[x,y]=0$.
        
        If $A^{R,S}_j$ is in stage II, it suffices to show that if $x$ and $y$ disagree on $R$, then they also disagree on $S$. If this holds, a short induction establishes that $x$ and $y$ disagree on $R$ for all $A^{R,S}_j\in \Act(x, R^{(1)})$ in stage II. By \cref{X_block_form_II_for_k}, it then follows that $X_j^{R,S}[x,y]=0$.
        
        By the definition of consistency, after loading $j\in\mathcal{C}(x)$,
        \begin{equation}
            \forall i\in \bigcup R, \alpha_S^x(i) = \alpha_R^x(i).
        \end{equation}
        If $x$ and $y$ agree on $S$, we must have
        \begin{equation}\label{eq:feasibility_contradiciton_agree_on_S_for_k}
            \forall i\in \bigcup R, \alpha_S^y(i) = \alpha_S^x(i) = \alpha_R^x(i).
        \end{equation}
        However, there exists $l\in\bigcup R$ such that $\alpha_R^y(l) \neq \alpha_R^x(l) = \alpha_S^y(l)$ because $x$ and $y$ disagree on $R$. Since $\{\alpha_R^y(l),\alpha_S^y(l)\} \subseteq \{*, y_l\}$, if \cref{eq:feasibility_contradiciton_agree_on_S_for_k} holds, it follows that $\alpha_R^y(l) = *$ and $y_l = \alpha_S^y(l) = \alpha_S^x(l) = x_l$. Furthermore, one such $l$ must satisfy $l\in \bigcup S_{i+1}$ and $j\in \bigcup S_i$ for some $i\in[k-1]$, implying $j=a_i$, and $y_l = x_l$ is incident to $y_{a_i} = x_{a_i}$. This contradicts the definition of consistency. Therefore, $x$ and $y$ must disagree on $S$.
    \end{itemize}
    \item 
    If $x$ and $y$ agree on $R^{(1)}$, we first show the contribution from arcs in $\Act(x,R^{(1)})$ in stage I to \cref{eq:feasibility_for_k} is 0. For each $i\in[r]$, when we load $t_i$ into $T_{i-1}$, $x$ and $y$ agree on both $T_{i-1}$ and $T_{i}$. An identical argument as in a previous case when $i\in[i^*-1]$ shows arc $A_{t_{i}}^{T_{i-1}, T_{i}}$ does not contribute to \cref{eq:feasibility_for_k}.

    For stage II, let $l$ be the smallest number in $[l]$ that $x_{a_l} \neq y_{a_l}$. Such $l$ must exist because $y$ is a negative instance. For active arcs in stage II.$l'$ such that $l' < l$, those arcs are loading $a_{l'}$ but $x_{a_{l'}}=y_{a_{l'}}$, so those arcs' $X$ are not included in the summation in \cref{eq:feasibility_for_k}.

    For arcs in stage II.$l'$ such that $l' > l$, $x$ and $y$ disagree on the vertices before loading $a_{l'}$. This is because $a_{l}$ gets uncovered in $x$, so $y_{a_{l}}$ must be equal to $x_{a_l}$ for $x$ and $y$ to agree, which leads a contradiction. Then, by \cref{X_block_form_II_for_k}, those active arcs' contribution to \cref{eq:feasibility_for_k} is also 0.

    Next, we show the contribution from arcs in $\Act(x,R^{(1)})$ in stage II.$l$ to \cref{eq:feasibility_for_k} is exactly $q$. For each $s\in[l]$ and sequence $(D_1, \dots, D_{s-1})$, there exists exactly one vertex $T^{(s)}\in V^{(s)}(D_1,\dots, D_{s-1})$ such that $\bigcup T^{(s)} = \bigcup R^{(1)}\cup\{a_1,\dots,a_{s-1}\}$.

    Let $B(s, D_1,\dots, D_{s-1})$ be the set of vertices in $V^{(l)}$ reachable from $T^{(s)}$ by arcs in $\Act(x, R^{(1)})$. There is a bijection between vertices in $B(s, D_1,\dots, D_{s-1})$ and choices $(D_s, \dots, D_{l-1})$; explicitly, each $R\in B(s, D_1,\dots, D_{s-1})$ corresponds uniquely to a single choice of $(D_s, \dots, D_{l-1})$, satisfying $R\in V^{(l)}(D_1,\dots, D_{l-1})$. Define the contribution from $T^{(s)}$ to \cref{eq:feasibility_for_k} as
    \begin{equation}
        C(T^{(s)}) \coloneqq \sum_{R\in B(s, D_1,\dots, D_{s-1}), A_{a_l}^{R,S}\in\Act(x, R^{(1)})}X_{a_l}^{R,S}[x,y].
    \end{equation}
    Then, our goal becomes showing $C(R^{(1)}) = q$, and we do this by induction on $s$ from $l$ down to 1, and showing
    \begin{equation}
        C(T^{(s)}) =
        \begin{cases}
            (-1)^{1+s+\sum_{i=1}^{s-1}\abs{D_i}}q & \text{if $x$ and $y$ agree on $T^{(s)}$,}\\
            0 & \text{if $x$ and $y$ disagree on $T^{(s)}$.}
        \end{cases}
    \end{equation}
    \begin{itemize}
    \item 
    If $s=l$, $T^{(s)}$ vertex has $2^{2d}-1$ outgoing active arcs in $\Act(x, R^{(1)})$, one for each $\emptyset\neq D_l\subseteq [2d]$, to vertex $T^{(l+1)}\in V^{(l+1)}(D_1,\dots,D_{l-1},D_l)$. If $x$ and $y$ disagree on $T^{(s)}$, all these arcs' contribution to \cref{eq:feasibility_for_k} is 0 by \cref{X_block_form_II_for_k}. If $x$ and $y$ agree on $T^{(s)}$, since $x_{a_l}\neq y_{a_l}$, all these outgoing active arcs' contribution are included in \cref{eq:feasibility_for_k}. By \cref{X_block_form_II_for_k}, the total contribution from these arcs to \cref{eq:feasibility_for_k} is
    \begin{equation}
        \sum_{\emptyset\neq D_l\subseteq [2d]}(-1)^{l+\sum_{i=1}^l\abs{D_i}}q = (-1)^{1+l+\sum_{i=1}^{l-1}\abs{D_i}}q.
    \end{equation}
    This proves the base case.
    \item 
    When $s<l$, for each $\emptyset\neq D_s\subseteq [2d]$, there exists a vertex in $T^{(s+1)}(D_s)\in V^{(s+1)}(D_1,\dots,D_s)$ such that $A_{a_s}^{T^{(s)}, T^{(s+1)}(D_s)}\in\Act(x, R^{(1)})$. Notice 
    \begin{equation}
        B(s,D_1,\dots, D_{s-1}) = \dot{\bigcup_{\emptyset\neq D_s\subseteq[2d]}} B(s+1,D_1,\dots, D_s),
    \end{equation}
    so
    \begin{equation}
        C(T^{(s)}) = \sum_{\emptyset\neq D_s\subseteq[2d]} C(T^{(s+1)}(D_s)).
    \end{equation}
    \begin{itemize}
        \item 
        If $x$ and $y$ disagree on $T^{(s)}$, since we are in stage II, as proved in a previous part, $x$ and $y$ disagree on $T^{(s+1)}(D_s)$ for all $\emptyset\neq D_s \subseteq[2d]$. Assuming the induction hypothesis, we have $C(T^{(s)}) = 0$.
        \item 
        If $x$ and $y$ agree on $T^{(s)}$, although $x_{a_s} = y_{a_s}$, $x$ and $y$ may not agree on $T^{(s+1)}(D_s)$ due to what are known as \emph{faults}. We say that an index $j\in \bigcup R^{(1)}_{s+1}$ is \emph{faulty} if $y_j$ is incident to $y_{a_s}$. Let $I$ denote the set of all $i\in[2d]$ such that
        \begin{equation}
            \bigcup_{\emptyset\neq D\subseteq[2d]} T^{(s)}_{s+1}(\mu(D_1),\dots,\mu(D_{s-1}), i, D),
        \end{equation}
        does not contain a faulty index. By the definitions of assignment and consistency, and the assumption that $x$ and $y$ agree on $T^{(s)}$, since $x_{a_s} = y_{a_s}$, $x$ and $y$ agree on $T^{(s+1)}(D_s)$ if and only if $\emptyset\neq D_s\subseteq I$.        
        There are at most $2(d-1)$ faulty indices since the graph has a max degree of $d$, so at most $2(d-1)$ other edges in $y$ can be incident to $y_{a_s}$ --- a necessary condition for causing a fault. Therefore, $\abs{I} \geq 2d-2(d-1) > 0$. By the induction hypothesis,
        \begin{align}
            C(T^{(s)}) =& \sum_{\emptyset\neq D_s\subseteq I} C(T^{(s+1)}(D_s)) \notag
            \\
            =& \sum_{\emptyset\neq D_s\subseteq I}(-1)^{1+(s+1)+\sum_{i=1}^{s}\abs{D_i}}q = (-1)^{1+s+\sum_{i=1}^{s-1}\abs{D_i}}q.
        \end{align}
    \end{itemize}
    \end{itemize}
    The above proof completes the induction and shows that $C(R^{(1)})$, the contribution from arcs in $\Act(x,R^{(1)})$ in stage II.$l$ to \cref{eq:feasibility_for_k}, is exactly $q$.
\end{itemize}

\subsection{Towards a tight lower bound for \texorpdfstring{$k$-distinctness}{kdist}}

In this subsection, we show that the next conjecture suffices to lift our nearly tight lower bound for $k$-cycle finding in the low maximum degree regime to a nearly tight lower bound for $k$-distinctness.

 \begin{conjecture} \label{conj:graph_property_partition}
    Let $m,n,b\in \mathbb{N}$ be such that $b=\ceil{\log(n)}$ and $n$ is divisible by $b$. Let $k\in \mathbb{N}_{\geq 3}$ and $\kcycle{k}\colon \binom{[n]}{2}^m \to \{0,1\}$ denote the $k$-cycle function. Let $U\subseteq \binom{[n]}{2}^m$ consist of all $x\in \binom{[n]}{2}^m$ such that $x$ does not contain repeated edges. 

    \begin{enumerate}
        \item Let $D\subseteq U$ consist of all $x\in U$ such that there \emph{exists} a partition of the vertex set $[n]$ into $n/b$ parts of size $b$ each such that every edge in $x$ connects two vertices \emph{within} the same part.
    
        \item Let $\calP$ be a partition of the vertex set $[n]$ into $n/b$ parts of size $b$ each, and let $D[\calP]$ consist of all $x\in D$ such that every edge in $x$ connects two vertices \emph{within} the same part \emph{of $\calP$}.
    \end{enumerate}
    Then the following relationship holds between $\kcycle{k}$ restricted to domain $D$ and $D[\calP]$:
    \begin{equation}
        Q(\kcycle{k}|_{D}) \leq  O(Q(\kcycle{k}|_{D[\calP]})).
    \end{equation}
\end{conjecture}

By definition, $D$ is the union of $D[\calP]$ over all partitions $\calP$. Therefore, the reverse inequality, $Q(\kcycle{k}|_{D[\calP]}) \leq  Q(\kcycle{k}|_{D})$, trivially holds. Why is the conjecture plausible? Intuitively, in $\kcycle{k}|_{D}$ we are computing $\kcycle{k}$ without access to the partition of the input graph, while in $\kcycle{k}|_{D[\calP]}$ we are given this extra information. However, the symmetries of the $\kcycle{k}$ problem appear to make this extra information hard to exploit. This is somewhat reminiscent of the fact that increasing the alphabet size of a highly symmetric function does not increase its approximate degree (a measure of complexity) as established by Ambainis in \cite{ambainis_symmetrize}. Indeed, the conjecture may also hold for functions other than $\kcycle{k}$, in particular, functions satisfying graph symmetries.

We now prove that the conjecture implies a tight lower bound for $k$-distinctness. The idea is that (i) $\kcycle{k}|_{D[\calP]}$ is $k$-distinctness in disguise and (ii) our lower bound for $k$-cycle in fact also holds for $\kcycle{k}|_{D}$. The next two lemmas show (i) and (ii) respectively.

\begin{lemma}\label{eq:conjecture_lemma1}
Let $k,b,m, D[\calP]$ and $\kcycle{k}$ be as defined in \cref{conj:graph_property_partition}. Then $Q(\kcycle{k}|_{D[\calP]}) \leq O(Q(\OR_{b^k} \circ \kdist{k}_{m})) \leq O(\sqrt{b^k}\cdot Q(\kdist{k}_m))$.
\end{lemma}
\begin{proof}
    The second inequality follows from well-known composition theorems for quantum query complexity \cite{reichardt_direct_sum} so we focus on proving the first.

    Given an input $x \in D[\calP]$ of length $m$, we proceed to construct $b^k$ many length-$m$ inputs to $\kdist{k}_m$, denoted $y^{(1)},\dots,y^{(b^k)}$, such that $\kdist{k}_m(y^{(j)})=1$ for some $j \in [b^k]$ if and only if $\kcycle{k}(x) = 1$.

    Let $P_1, P_2, \dots, P_{n/b}$ be an arbitrary but fixed enumeration of the parts in $\calP$. For each $\ell \in [n/b]$, let $P_\ell^{k} \coloneqq \{(v_1, v_2, \dots, v_k, v_1):v_1,v_2,\dots,v_k \in P_\ell\}$. For each $j \in [b^k]$, let $P_\ell^{k}[j]$ denote the $j$th element of $P_\ell^{k}$ under an arbitrary but fixed enumeration. For an edge $\{u,v\}\in \binom{[n]}{2}$, we write $\{u,v\} \in_{\text{adj}} P_\ell^{k}[j]$ if vertices $u$ and $v$ are consecutive elements in $P_\ell^{k}[j]$. For $j \in [b^k]$ and $i \in [n/b]$, let $y^{(j)}_i$ be $P_\ell^{k}[j]$ if $x_i \in_{\text{adj}} P_\ell^{k}[j]$ for some $\ell \in [n/b]$, and $x_i$ otherwise. (To be clear, $P_\ell^{k}[j]$ is a tuple, so it can never equal an edge like $x_i$, which is a set.) Since $\calP$ is given, querying $y_i^{(j)}$ only requires querying $x_i$ for all $i \in [n/b]$ and $j \in [b^k]$.

    We see $\kdist{k}_m(y^{(j)})=1$ for some $j \in [b^k]$ implies $\kcycle{k}(x) = 1$ as follows. Since $y^{(j)}$ is a positive instance of $\kdist{k}_m$, there exists $\ell \in [n/b]$ such that $P_\ell^{k}[j]$ appears (at least) $k$ times, which means that there exist indices $i_1, \dots, i_k$ such that edges $x_{i_1}, \dots, x_{i_k} \in_{\text{adj}} P_\ell^{k}[j]$. Let $v_1, \dots, v_k$ be such that $P_\ell^{k}[j] = (v_1, v_2, \dots, v_k, v_1)$. Since there are no repeated edges in $x\in D[\calP]$ by the definition of $D[\calP]$, the set of edges $\{x_{i_1}, \dots, x_{i_k}\}$ must be the same as $\{\{v_1, v_2\}, \{v_2, v_3\}, \dots, \{v_{k-1}, v_k\}, \{v_k, v_1\}\}$ by the definition of $\in_{\text{adj}}$. Hence $x$ is a positive instance of $\kcycle{k}$.

     We see $\kcycle{k}(x) = 1$ implies $\kdist{k}_m(y^{(j)})=1$ for some $j \in [b^k]$ as follows. Since $x$ is a positive instance of $\kcycle{k}$, there exist vertices $v_1, v_2, \dots, v_k$ such that $x$ contains the edges $\{v_1, v_2\}, \{v_2, v_3\}, \dots, \{v_{k-1}, v_k\}, \{v_k, v_1\}$. Since $x \in D[\calP]$, the vertices $v_1, v_2, \dots, v_k$ must be in $P_\ell$ for some $\ell \in [n/b]$. Therefore, there exists $j \in [b^k]$ such that $P_\ell^k[j]$ is $(v_1, v_2, \dots, v_k, v_1)$. It follows that $\{v_1, v_2\}, \{v_2, v_3\}, \dots, \{v_{k-1}, v_k\}, \{v_k, v_1\} \in_{\text{adj}} P_\ell^k[j]$, which implies that $P_\ell^k[j]$ appears $k$ times in $y^{(j)}$. Hence $y^{(j)}$ is a positive instance of $\kdist{k}_m$.  
\end{proof}

Proving the next lemma requires some more basic facts about random graphs, which we collect here. The proof of this fact is given in \cref{app:random_graphs}.
\begin{fact}\label{fact:cycle_partition_duplicates}
    Let $n,m,b\in \mathbb{N}$ be such $b=\ceil{\log(n)}$ and $n$ is divisible by $b$. Let $k\in \mathbb{N}_{\geq 3}$. Suppose $m = \Theta(n/b)$ and $x\leftarrow \binom{[n]}{2}^m$, then:
    \begin{enumerate}
    \item with probability at least $\Omega(1/\log^k(n))$, $x$ contains a $k$-cycle;
    \item with probability at least $1 - o(1/\log^{k}(n))$, there exists a partition of the vertex set $[n]$ into $n/b$ sets of size $b$ each such that every edge in $x$ connects two vertices \emph{within} the same part;
    \item with probability at least $1 - o(1/\log^k(n))$, there are fewer than $k$ edges in $x$ whose deletion leaves $x$ not containing any repeated edges.
\end{enumerate}
In particular, the probability that all three events above occur together is at least $\Omega(1/\log^k(n))$.
\end{fact}

We remark that for the second and third items, we can in fact lower bound the probabilities by quantities considerably larger than $1 - o(1/\log^k(n))$, as will be clear from the proof, but the results as stated suffice for our purpose of proving the next lemma.

\begin{lemma}\label{eq:conjecture_lemma2}
    Let $k,m,n,b,D$ and $\kcycle{k}$ be as defined in \cref{conj:graph_property_partition}. Suppose $n$ is such that $m = \Theta(n/b)$. Then $Q(\kcycle{k}|_{D}) \geq \widetilde{\Omega}(m^{3/4-1/(2^{k+2}-4)})$.
\end{lemma}
\begin{proof}
    Assume for contradiction that $Q(\kcycle{k}|_{D}) \leq o(m^{3/4-1/(2^{k+2}-4)}/\log^c(m))$ for some constant $c$ that will be set later.
    
    Consider $x\leftarrow \binom{[n]}{2}^m$. By \cref{fact:cycle_partition_duplicates}, with probability at least $\Omega(1/\log^k(n))$, all three events ---- (i), (ii), and (iii) ---  described in \cref{fact:cycle_partition_duplicates} happen together. Suppose we are in this ``good'' scenario for the rest of this paragraph except during the last sentence. Since (iii) holds, we can use the element distinctness algorithm $k$ times to preprocess the input $x$ to remove all of its repeated edges at a cost of $O(k m^{2/3})$, which is $O(m^{2/3})$ since $k$ is constant. This preprocessing, together with the assumption that (ii) holds, means that $x$ is now in $D$, the domain of $\kcycle{k}|_{D}$. Therefore, using our assumption at the start of the proof together with the search-to-decision reduction (see preliminaries section), we can find the $k$-cycle with constant probability using $o(m^{3/4-1/(2^{k+2}-4)}/\log^{c-2}(m)) + O(m^{2/3})$ queries, which can be simply written as $o(m^{3/4-1/(2^{k+2}-4)}/\log^{c-2}(m))$ since $2/3<3/4-1/(2^{k+2}-4)$ for all $k\geq 3$. Since the good scenario occurs with probability at least $\Omega(1/\log^k(n))$, this means we have constructed a quantum algorithm $\calA$ using $o(m^{3/4-1/(2^{k+2}-4)}/\log^{c-2}(m))$ queries that finds a $k$-cycle in $x$ with probability at least $\Omega(1/\log^k(n))$, where the probability is over the randomness in $\calA$ and the distribution $x\leftarrow \binom{[n]}{2}^m$.

    We will now reach a contradiction by using our lower bound theorem for $k$-cycle finding, \cref{thm:kcycle_avg_lower}. To apply \cref{thm:kcycle_avg_lower}, we set the $\Delta$ in its statement to $\log(n)$, and choose $c$ satisfying $(c-2)(2-2/2^k)-k\geq k+1$. Since $\calA$ uses only $o(m^{3/4-1/(2^{k+2}-4)}/\log^{c-2}(m))$ queries, the theorem shows that the probability of $\calA$ finding a $k$-cycle is at most $O(1/\log^{k+1}(n))$. Again, the probability is over the randomness in $\calA$ and the distribution $x\leftarrow \binom{[n]}{2}^m$.

    The last two paragraphs give $\Omega(1/\log^k(n)) \leq O(1/\log^{k+1}(n))$, which is the desired contradiction.
\end{proof}

\begin{proposition}\label{prop:conjecture_implies_lower}
    Assuming \cref{conj:graph_property_partition}, $Q(\kdist{k}_m) \geq \widetilde{\Omega}(m^{3/4-1/(2^{k+2}-4)})$.
\end{proposition}
\begin{proof}
    Let $n\in \mathbb{N}$ be such that $n$ is divisible by $b\coloneqq \ceil{\log(n)}$ and $m = \Theta(n/b)$. (A simple argument shows such $n$ exists for any given $m$.) Let $k\in \mathbb{N}_{\geq 3}$ and $\kcycle{k}\colon \binom{[n]}{2}^m \to \{0,1\}$ denote the $k$-cycle function. Let $D$ and $\calP$ be as defined in \cref{conj:graph_property_partition}.

    Then, we have
    \begin{align*}
        \widetilde{\Omega}(m^{3/4-1/(2^{k+2}-4)})
        \leq&~ Q(\kcycle{k}|_{\calD}) &&\text{\cref{eq:conjecture_lemma2}}
        \\
        \leq&~O(Q(\kcycle{k}|_{\calD[\calP]})) &&\text{\cref{conj:graph_property_partition}}
        \\
        \leq&~ O(\sqrt{b^k}\cdot Q(\kdist{k}_m) 
        \leq \widetilde{O}(Q(\kdist{k}_m)), &&\text{\cref{eq:conjecture_lemma1}}
    \end{align*}
    which establishes the proposition.
\end{proof}

\section*{Acknowledgments}
We thank François Le Gall for suggesting the study of quantum algorithms in the edge list model. We thank Yassine Hamoudi for helpful discussions on \cite{Hamoudi_2023} and for simplifying our original proof of the Mirroring Lemma (\cref{lem:mirror}) in the binary-alphabet case. We thank anonymous reviewers for identifying errors in the Mirroring Lemma in submitted versions of this paper; these have been corrected in the current version.

Amin Shiraz Gilani is supported by the U.S. Department of Energy ASCR Quantum Testbed Pathfinder program (awards DE-SC0019040 and DE-SC0024220) and the U.S. Department of Energy, Office of Science, Accelerated Research in Quantum Computing, Fundamental Algorithmic Research toward Quantum Utility (FAR-Qu) program. Daochen Wang is supported by NSERC
Grants CRC-2023-00039, RGPIN-2024-06493, and DGECR-2024-00113. Xingyu Zhou is supported by NSERC
Grant RGPIN-2024-06493.

\phantomsection
\addcontentsline{toc}{section}{References}

\bibliography{references}
\bibliographystyle{alphaurl}

\newpage
\appendix
\section{Appendix: Facts about random graphs}\label{app:random_graphs}
For completeness, we prove facts on random graphs used in the main body in this appendix.

\begin{fact}[Sparse random graphs]
\label{fact:sparse-random-graph-restated}
Let $x\leftarrow \binom{[n]}{2}^m$.
\begin{enumerate}
    \item \label{enu:random-graph-maxdeg-restated} (Low max degree) For $m  \leq O(n)$, 
    \begin{equation}
        \Pr[\maxdeg(x) \leq 2\log (n) /\log\log (n))] \geq  1-o(1).
    \end{equation}
    \item \label{enu:random-graph-triangle-restated} (Existence of a $k$-cycle) For $m \geq \Omega(n)$, 
    \begin{equation}
        \Pr[x \textup{ contains a $k$-cycle}] \geq  \Omega(1).
    \end{equation}
    \item \label{enu:random-graph-triangle-sparse} (Existence of a $k$-cycle) For $m \leq o(n)$, 
    \begin{equation}
        \Pr[x \textup{ contains a $k$-cycle}] \geq  \Omega((m/n)^k).
    \end{equation}
\end{enumerate} 
In particular, the first two items imply that, for $m = \Theta(n)$, 
\begin{equation}
    \Pr[\maxdeg(x) \leq 2\log(n) /\log\log(n) \textup{ and $x$ contains a triangle}] \geq \Omega(1).
\end{equation}
\end{fact}
\begin{proof}
    \cref{enu:random-graph-maxdeg-restated} follows from the Chernoff bound followed by the union bound, like in a textbook analysis of the maximum load in a balls-into-bins experiment \cite{raab1998balls}. We give the details for completeness.

    Let $d\geq 0$. For a fixed vertex $v\in [n]$, the Chernoff bound (\cref{lemma_chernoff}) shows that
    \begin{equation}
        \Pr[\text{degree of $v$ in $x$ is at least $d$}] \leq \Bigl (\frac{e}{d}\cdot \frac{2m}{n}\Bigr)^d.
    \end{equation}
    Therefore, taking the union bound over all $v\in [n]$, and then using $m\leq O(n)$, gives
    \begin{equation}
        \Pr[\maxdeg(x) \geq d] \leq n \, \Bigl(\frac{e}{d}\cdot \frac{2m}{n}\Bigr)^d \leq n \, \Bigl(\frac{C}{d} \Bigr)^d
    \end{equation}
    for some constant $C>0$.

    Now, $n$ can be expressed as $\log(n)^{\log(n)/\log\log(n)}$, so 
    \begin{equation}
        \Pr[\maxdeg(x) \geq 2\log(n)/\log\log n] \leq \Bigl(\frac{C\sqrt{\log(n)} \cdot \log\log(n)}{2\log(n)} \Bigr)^{2\log(n)/\log\log(n)} < o(1),
    \end{equation}
    from which the result follows.

    \cref{enu:random-graph-triangle-restated,enu:random-graph-triangle-sparse} can be shown using the second-moment method as follows. The proofs for both are the same until after \cref{eq:expect_Xsquare_intermediate}.
   
    Consider the random variable
    \begin{equation}
        X \coloneqq \sum_{0\leq i_1<\dots< i_k\leq m} \1[x_{i_1},\dots x_{i_k} \text{ forms a \emph{sequential} $k$-cycle}],
    \end{equation}
    where the word ``sequential'' imposes the extra condition that for every $l\in \{1,\dots,k-1\}$, the edges $x_{i_1},\dots,x_{i_l}$ form a length-$l$ path. Clearly, it suffices to prove the claimed lower bounds on $\Pr[X>0]$. We restrict attention to sequential $k$-cycles as that makes $\expect[X]$ and $\expect[X^2]$ easier to analyze.
    
    We have
    \begin{equation}\label{eq:cycle_expectation}
         \expect[X] =\binom{m}{k}\frac{N\cdot2(n-2)\cdots 2(n-k+1)}{N^k},
    \end{equation}
    which gives $\expect[X] = \Theta((m/n)^k)$, and
    \begin{align}
        &\expect[X^2] 
        \\
        =& \expect[X] + \sum_{l=0}^{k-1}\sum_{\substack{i_1,\dots, i_k, j_1,\dots,j_k \in [m]\colon\\ \ \abs{\{i_1,\dots,i_k\}\cap\{j_1,\dots,j_k\}}=l}} \Pr[\text{$x_{i_1},\dots, x_{i_k}$ and $x_{j_1},\dots, x_{j_k}$ both form a $k$-cycle}]
        \notag
        \\
        =& \expect[X] + \binom{m}{k}\binom{m-k}{k}\Bigl(\frac{N\cdot2(n-2)\cdots 2(n-k+1)}{N^k}\Bigr)^2
        \notag
        \\
        &\ +\sum_{l=1}^{k-1}\binom{m}{l}\binom{m-l}{k-l}\binom{m-k}{k-l}\cdot \frac{N\cdot2(n-2)\cdots 2(n-k+1)}{N^k} \cdot \frac{2(n-l+1)\cdots 2(n-k+3)}{N^{k-l}}
        \notag
        \\
        \leq&\expect[X] + \expect[X]^2 
        \notag
        \\
        &\ +\sum_{l=1}^{k-1}\binom{m}{l}\binom{m-l}{k-l}\binom{m-k}{k-l}\cdot \frac{N\cdot2(n-2)\cdots 2(n-k+1)}{N^k} \cdot \frac{2(n-l+1)\cdots 2(n-k+3)}{N^{k-l}}, \label{eq:expect_Xsquare_initial}
    \end{align}
    where the term $(2(n-l+1)\cdots 2(n-k+3))$ is taken to mean $1$ when $l = k-1$.

    Continuing \cref{eq:expect_Xsquare_initial} gives
    \begin{equation}\label{eq:expect_Xsquare_intermediate}
        \expect[X^2] \leq \expect[X] + \expect[X]^2  + \frac{1}{n^2}\cdot O\Bigl(\sum_{l=1}^{k-2}\Bigl(\frac{m}{n}\Bigr)^{2k-l}\Bigr) + \frac{1}{n}\cdot O\Bigl(\Bigl(\frac{m}{n}\Bigr)^{k+1}\Bigr)\Bigr).
    \end{equation}
    \begin{enumerate}
        \item To prove \cref{enu:random-graph-triangle-restated}, which is for $m\geq \Omega(n)$, we may first wlog assume $m=\Theta(n)$ since the probability we are lower bounding increases with $m$. Therefore, $\expect[X] = \Theta(1)$ and the second-moment method gives
    \begin{equation}
        \Pr[X > 0] \geq \expect[X]^2/\expect[X^2] \geq \expect[X]^2/(\expect[X] + \expect[X]^2 +\Theta(1/n)) \geq \Omega(1),
    \end{equation}
    
    \item To prove \cref{enu:random-graph-triangle-sparse}, which is for $m\leq o(n)$, we first observe that since $\expect[X] = \Theta((m/n)^k)$, \cref{eq:expect_Xsquare_intermediate} gives
    \begin{equation}
        \expect[X^2] \leq \expect[X](1+O(1/n)) + \expect[X]^2.
    \end{equation}
        Therefore, the second-moment method gives
    \begin{equation}
    \begin{aligned}
        \Pr[X > 0] \geq \expect[X]^2/\expect[X^2] \geq&~ 1 - \Bigl(1 + \frac{\expect[X]}{(1+O(1/n))}\Bigr)^{-1} 
        \\
        \geq&~ \frac{1}{2}\frac{\expect[X]}{(1+O(1/n))}\geq \frac{1}{4}\expect[X] \geq \Omega((m/n)^k),
    \end{aligned}
    \end{equation}
    where the third inequality uses $(1+x)^{-1} \leq 1-x/2$ for $x\in [0,1]$ which applies as $\expect[X] \leq o(1)$.\qedhere
    \end{enumerate}
\end{proof}

\begin{fact}\label{fact:sublinear-graphs}
    Let $n,m,b\in \mathbb{N}$ be such $b=\ceil{\log(n)}$ and $n$ is divisible by $b$. Let $k\in \mathbb{N}_{\geq 3}$. Suppose $m = \Theta(n/b)$ and $x\leftarrow \binom{[n]}{2}^m$, then:
    \begin{enumerate}
    \item \label{enu:sublinear-kcycle} with probability at least $\Omega(1/\log^k(n))$, $x$ contains a $k$-cycle;
    \item \label{enu:sublinear-partition} with probability at least $1 - o(1/\log^{k}(n))$, there exists a partition of the vertex set $[n]$ into $n/b$ sets of size $b$ each such that every edge in $x$ connects two vertices \emph{within} the same part;
    \item \label{enu:sublinear-dup} with probability at least $1 - o(1/\log^k(n))$, there are fewer than $k$ edges in $x$ whose deletion leaves $x$ not containing any repeated edges.
\end{enumerate}
In particular, the probability that all three events above occur together is at least $\Omega(1/\log^k(n))$.
\end{fact}

\begin{proof}
The ``in particular'' part follows immediately from the three items of the fact. We prove each item as follows. We assume $\log(n)$ is a positive integer for notational convenience.

\cref{enu:sublinear-kcycle} This follows from the third item of \cref{fact:sparse-random-graph-restated}.

\cref{enu:sublinear-partition} 
We consider the Erd\H{o}s--R\'enyi $G_{n,p}$ model, where $p=C(n\log n)^{-1}$ for some suitably large constant $C>0$. This suffices because for large $C$, by concentration, the probability that the number of edges in $x\leftarrow G_{n,p}$ is fewer than $m$ is at most $\exp(- \Theta(n/\log n))$, which is exponentially small. The partition property the statement is looking for is a monotone property, the smaller the number of edges the easier the property is satisfied.\footnote{More formally, because the underlying distribution of $G_{n,p}$ is just $\expect_{m}[G_{n,m}]$ for some distribution on the edge count $m$.}

For a graph $x\leftarrow G_{n,p}$, let $\chi(x)$ denote the largest connected component in $x$. It is well known for sparse (sublinear in our case) random graphs that there are no giant connected components.
For us, it suffices to show the following bound
\begin{align}\label{eq:sublinear-chi-bound}
    \Pr_{x\leftarrow G_{n,p}} [\chi(x) \ge \log n] \le O((\log n) ^{-\log n}),
\end{align}
which is easy to see as follows. Suppose that $\chi(x)=s$, then for some subset of $s$ vertices, there are at least $(s-1)$ edges, the probability of which is bounded by
\begin{equation}
   \Pr_{x\leftarrow G_{n,p}}[\chi(x)\ge s] \leq \binom{n}{s} \binom{s(s-1)/2}{ s-1 } p^{s-1}\leq \Bigl(\frac{en}{s}\Bigr)^s \cdot \Bigl(\frac{es}{2}\Bigr)^{s-1} \cdot \Bigl( \frac{C}{n\log n}\Bigr) ^{s-1} \leq \frac{en}{s}\cdot\Bigl(\frac{Ce^2}{2\log n}\Bigr)^{s-1}.
\end{equation}
Setting $s=\log n$, we obtain~\cref{eq:sublinear-chi-bound}.

As long as $\chi(x)\le \log n$, there always is a partition with the required property. To obtain such a partition, for each connected component with more than one edge, we take some isolated vertices and group them into a part of size $\log n$. Since there are only $m$ edges, in the worst case, there are at most $m$ nontrivial connected components each demanding $\log(n) -2$ single vertices to become a part of size $\log n$. In total, we need at most $m (\log(n)- 2) = n-2m$ isolated vertices. Since we always have at least $n-2m$ isolated vertices, this concludes the proof.

\cref{enu:sublinear-dup} We bound the probability of there being a $3$-collision in $x$ by
\begin{align}   
    \Pr[\exists \text{ distinct } &i_1, i_2, i_3~ \text{s.t.}~ x_{i_1} = x_{i_2} = x_{i_3}] 
    \notag \\
    &\leq \sum_{\text{distinct } i_1,i_2,i_3 \in [m]} \Pr[x_{i_1} = x_{i_2} = x_{i_3}] = \textstyle{\binom{m}{3}}/\textstyle{\binom{n}{2}^2} = O(1/n).
\end{align}
Therefore, with probability $1 - O(1/n)$, there are no $3$-collisions in $x$. 

Now, we compute the probability of there being (at least) $k$ $2$-collisions in $x$: 
\begin{align}
    \Pr[\exists ~&\text{distinct }i_1, j_1, \dots, i_k,j_k: x_{i_1} = x_{j_1}, \dots, x_{i_k}=x_{j_k}] 
    \notag \\
    &= \prod_{\ell \in [k]} \Pr\bigl[\exists \text{ distinct }i_\ell, j_\ell \not \in \{i_1,j_1, \dots, i_{\ell-1}, j_{\ell-1}\} \text{ s.t. } x_{i_\ell} = x_{j_\ell} 
    \notag\\
    &\qquad\qquad\qquad\qquad \bigm| \exists \text{ distinct }i_1, j_1, \dots, i_{\ell-1},j_{\ell-1} \text{ s.t. } x_{i_1} = x_{j_1}, \dots, x_{i_{\ell-1}}=x_{j_{\ell-1}}\bigr] 
    \notag\\
    &\leq \prod_{\ell \in [k]} \sum_{\substack{\text{ distinct }i_\ell, j_\ell\\ \not \in  \{i_1,j_1, \dots, i_{\ell-1}, j_{\ell-1}\}}} \Pr[x_{i_\ell} = x_{j_\ell}] = \Bigl(\textstyle{\binom{n/\log(n) - 2(\ell-1)}{2} \binom{n}{2}^{-1}}\Bigr)^{k} \leq  O\bigl(1/\log^{2k}(n)\bigr).
\end{align}
Therefore, with probability at least $1 - O(1/\log^{2k}(n))$, there are fewer than $k$ $2$-collisions in $x$. 

It follows that with probability at least $1 - O(1/\log^{2k}(n))$, there are fewer than $k$ elements in $x$ whose deletion leaves $x$ with no repeated edges.
\end{proof}

\end{document}